\documentclass[final]{IEEEtran}
\usepackage{listings}
\usepackage{amsmath}
\usepackage{amsthm}
\usepackage{tikz}
\usepackage{caption}
\usepackage{array}
\usepackage{mdwmath}
\usepackage{multirow}
\usepackage{mdwtab}
\usepackage{eqparbox}
\usepackage{amsfonts}
\usepackage{tikz}
\usepackage{multirow,bigstrut,threeparttable}
\usepackage{amsthm}
\usepackage{array}
\usepackage{bbm}
\usepackage{subfigure}
\usepackage{epstopdf}
\usepackage{mdwmath}
\usepackage{mdwtab}
\usepackage{eqparbox}
\usepackage{tikz}
\usepackage{latexsym}
\usepackage{cite}
\usepackage{amssymb}
\usepackage{bm}
\usepackage{amssymb}
\usepackage{graphicx}
\usepackage{mathrsfs}
\usepackage{epsfig}
\usepackage{psfrag}
\usepackage{setspace}
\usepackage{hyperref}
\usepackage{algorithm}
\usepackage{algpseudocode}
\usepackage{stfloats}

\newtheorem{theorem}{Theorem}

\newtheorem{lemma}{Lemma} 
\newtheorem{corollary}{Corollary}
\newtheorem{question}{Question}

\def \bP {\mathbb{P}}
\def \bE {\mathbb{E}}

\def \cL {\mathcal{L}}
\def \spo {\mathsf{Poi}}

\DeclareMathOperator*{\argmin}{\arg\!\min}
\DeclareMathOperator*{\argmax}{\arg\!\max}
\begin{document}

\title{Minimax Estimation of Functionals of Discrete Distributions}

\author{Jiantao~Jiao,~\IEEEmembership{Student Member,~IEEE},~Kartik~Venkat,~\IEEEmembership{Student Member,~IEEE},~Yanjun~Han,~\IEEEmembership{Student Member,~IEEE}, and Tsachy~Weissman,~\IEEEmembership{Fellow,~IEEE}
\thanks{Manuscript received Month 00, 0000; revised Month 00, 0000; accepted Month 00, 0000. Date of current version Month 00, 0000. This work was supported in part by two Stanford Graduate Fellowships, and by the Center for Science of Information (CSoI), an NSF Science and Technology Center, under grant agreement CCF-0939370.}
\thanks{Jiantao Jiao, Kartik Venkat, and Tsachy Weissman are with the Department of Electrical Engineering, Stanford University, CA, USA. Email: \{jiantao,kvenkat,tsachy\}@stanford.edu.} 
\thanks{Yanjun Han is with the Department of Electronic Engineering, Tsinghua University, Beijing, China. Email: hanyj11@mails.tsinghua.edu.cn.}
\thanks{Communicated by H. Permuter, Associate Editor for Shannon Theory.}%
\thanks{Color versions of one or more of the figures in this paper are available 
online at \url{http://ieeexplore.ieee.org}.}%
\thanks{The Matlab code package of this paper can be downloaded from \url{http://web.stanford.edu/\~tsachy/software.html}.}%
\thanks{Copyright (c) 2014 IEEE. Personal use of this material is permitted.  However, permission to use this material for any other purposes must be obtained from the IEEE by sending a request to pubs-permissions@ieee.org.}
\thanks{Digital Object Identifier 10.1109/TIT.2015.0000000}%
}
%
%

\date{\today}

\vspace{-10pt}

\maketitle

\begin{abstract}
We propose a general methodology for the construction and analysis of essentially minimax estimators for a wide class of functionals of finite dimensional parameters, and elaborate on the case of discrete distributions, where the support size $S$ is unknown and may be comparable with or even much larger than the number of observations $n$. We treat the respective regions where the functional is ``nonsmooth'' and ``smooth'' separately. In the ``nonsmooth'' regime, we apply an unbiased estimator for the best polynomial approximation of the functional whereas, in the  ``smooth'' regime, we apply a bias-corrected version of the Maximum Likelihood Estimator (MLE). 

We illustrate the merit of this approach by thoroughly analyzing the performance of the resulting schemes for estimating two important information measures: the entropy $H(P) = \sum_{i = 1}^S -p_i \ln p_i$ and $F_\alpha(P) = \sum_{i = 1}^S p_i^\alpha,\alpha>0$. We obtain the minimax $L_2$ rates for estimating these functionals. In particular, we demonstrate that our estimator achieves the optimal sample complexity $n \asymp S/\ln S$ for entropy estimation. We also demonstrate that the sample complexity for estimating $F_\alpha(P),0<\alpha<1$, is $n\asymp S^{1/\alpha}/ \ln S$, which can be achieved by our estimator but not the MLE. For $1<\alpha<3/2$, we show the minimax $L_2$ rate for estimating $F_\alpha(P)$ is $(n\ln n)^{-2(\alpha-1)}$ for infinite support size, while the maximum $L_2$ rate for the MLE is $n^{-2(\alpha-1)}$. For all the above cases, the behavior of the minimax rate-optimal estimators with $n$ samples is essentially that of the MLE (plug-in rule) with $n\ln n$ samples, which we term ``effective sample size enlargement''. 

We highlight the practical advantages of our schemes for the estimation of entropy and mutual information. We compare our performance with various existing approaches, and demonstrate that our approach reduces running time and boosts the accuracy. Moreover, we show that the minimax rate-optimal mutual information estimator yielded by our framework leads to significant performance boosts over the Chow--Liu algorithm in learning graphical models. The wide use of information measure estimation suggests that the insights and estimators obtained in this work could be broadly applicable. 

\end{abstract}


\section{Introduction and main results}

Given $n$ independent samples from an unknown discrete probability distribution $P = (p_1, p_2, \ldots, p_S)$, with \emph{unknown} support size $S$, consider the problem of estimating a functional of the distribution of the form:
\begin{equation}\label{eqn.generalf}
F(P) = \sum_{i=1}^{S} f(p_i),
\end{equation}
where $f: (0,1] \to \mathbb{R}$ is a continuous function. Among the most fundamental of such functionals is the entropy \cite{Shannon1948},
\begin{equation}\label{eqn.entropy}
H(P) \triangleq \sum_{i = 1}^S - p_i \ln p_i,
\end{equation}
which plays significant roles in information theory~\cite{Shannon1948}. Another information theoretic quantity which is closely related to the entropy is the mutual information, which for discrete random variables can be defined as
\begin{equation}
I(X;Y) = I(P_{XY}) = H(P_X) + H(P_Y) - H(P_{XY}),
\end{equation}
where $P_X, P_Y, P_{XY}$ denote, respectively, the distributions of random variables $X$, $Y$, and the pair $(X,Y)$.

We are also interested in the family of information measures $F_\alpha(P)$:
\begin{equation}
F_\alpha (P) \triangleq \sum_{i = 1}^S p_i^\alpha, \alpha>0.
\end{equation}
The significance of functional $F_\alpha(P)$ can be seen via the connection $H_\alpha(P) = \frac{\ln F_\alpha(P)}{1-\alpha}$, where $H_\alpha(P)$ is the R\'enyi entropy \cite{Renyi1961measures}, emerging in answering fundamental questions in information theory~\cite{Csiszar1995generalized},\cite{Courtade--Verdu2014cumulant,Courtade--Verdu2014variable}. The functional $1-F_2(P)$ is also called Gini impurity, which is widely used in machine learning~\cite{Breiman--Friedman--Stone--Olshen1984classification}.

Over the years, the use of information theoretic measures, especially entropy and mutual information, has extended far beyond the information theory community, and is deeply imbued in fundamental concepts from various disciplines. In statistics, one of the popular criteria for objective Bayesian modeling~\cite{Lehmann--Casella1998theory} is to design a prior on the parameter to maximize the mutual information between the parameter and the observations. In machine learning, the so-called \emph{infomax}\cite{Linsker1988self} criterion states that the function that maps a set of input values to a set of output values should be chosen or learned so as to maximize the mutual information between the input and output, subject to a set of specified constraints. This principle has been widely adopted in practice, for example, in decision tree based algorithms in machine learning such as C4.5~\cite{Quinlan1993c4}, one tries to select the feature at each step of tree splitting to maximize the mutual information (called \emph{information gain} principle~\cite{Nowozin2012improved}) between the output and the feature conditioned on previous chosen features. Other measures in feature selection have been proposed, such as the Gini impurity (used in CART~\cite{Breiman--Friedman--Stone--Olshen1984classification}), variance reduction~\cite{Breiman--Friedman--Stone--Olshen1984classification}, and many of them can be incorporated as special cases of what we study in this paper. We emphasize that in some applications, mutual information arises naturally as the only answer, for example, the well known Chow--Liu algorithm \cite{Chow--Liu1968} for learning tree graphical models relies on estimation of the mutual information, which is a natural consequence of maximum likelihood estimation. Recently, it was shown~\cite{Jiao--Courtade--Venkat--Weissman2014Justification} that mutual information is the unique measure of relevance for inference in the presence of side information to satisfy a natural data processing property. 

We also mention genetics~\cite{Olsen--Meyer--Bontempi2009impact}, image processing~\cite{Pluim--Maintz--Viergever2003mutual}, computer vision~\cite{Viola--Wells1997alignment}, secrecy~\cite{Batina--Gierlichs--Prouff--Rivain--Standaert--Veyrat2011mutual}, ecology~\cite{Hill1973diversity}, and physics~\cite{Franchini--Its--Korepin2008Renyi} as fields in which information theoretic measures are widely used. There are some other functionals that can be loosely categorized as information theoretic measures, such as the association measures quantifying certain dependency relations of random variables~\cite{Schmid--Schmidt--Blumentritt--Gaisser--Ruppert2010copula}, and divergence measures~\cite{Basseville2010}.

In most applications, the underlying distribution is unknown, so we cannot compute these information theoretic measures exactly. Hence, in nearly every problem that uses information theoretic measures, we need to estimate these quantities from the data, which is what we study in this paper. 
Our contributions are threefold. (i) We show that when the number of observations $n$ is comparable to the parameter dimension (a relevant regime in the ``big data'' era), the prevailing approaches (such as plug-in of the maximum likelihood estimator) can be highly sub-optimal. (ii) We propose new and computationally efficient algorithms that are essentially optimal in terms of the worst case squared error risk. That is, we both characterize the fundamental limits on estimation performance, and propose practical algorithms that essentially achieve them. Our results establish that for such functional estimation scenarios, replacing the plug-in (maximum likelihood) estimator by our practical and essentially minimax optimal estimators yields an effective enlargement of the sample size from $n$ to $n \ln n$, which can make a significant difference in practice. (iii) We demonstrate the efficacy of our schemes by a comparison with existing procedures in the literature, as well as illustrate performance boosts over traditional schemes on both real and simulated data. 

Notation:  We use the notation $a_\gamma \lesssim  b_\gamma$ to denote that there exists a universal constant $C$ such that $\sup_{\gamma } \frac{a_\gamma}{b_\gamma} \leq C$. Notation $a_\gamma \asymp b_\gamma$ is equivalent to $a_\gamma \lesssim  b_\gamma$ and $b_\gamma \lesssim  a_\gamma$. Notation $a_\gamma \gg b_\gamma$ means that $\liminf_\gamma \frac{a_\gamma}{b_\gamma} = \infty$, and $a_\gamma \ll b_\gamma$ is equivalent to $b_\gamma \gg a_\gamma$. The sequences $a_\gamma,b_\gamma$ are non-negative.

\subsection{Our estimators}\label{subsec.ourestimators}

Our main goal in this work is to present a general approach to the construction of minimax rate-optimal estimators for functionals of the form (\ref{eqn.generalf}) under $L_2$ loss. To illustrate our approach, we describe and analyze explicit constructions for the specific cases of entropy $H(P)$ and $F_\alpha(P)$, from which the construction for any other functional of the form~(\ref{eqn.generalf}) will be clear. Our estimators for each of these two functionals are agnostic with respect to the support size $S$, and achieve the minimax $L_2$ rates (i.e. the performance of our approaches when we do not know the support size $S$ does not degrade compared with the case where the support size $S$ is known).

Our approach is to tackle the estimation problem separately for the cases of ``small $p$'' and ``large $p$'' in $H(P)$ and $F_\alpha(P)$ estimation, corresponding to treating regions where the functional is nonsmooth and smooth in different ways. As we describe in detail in the sections to follow, where we give a full account of our estimators, in the nonsmooth region, we rely on the best polynomial approximation of the function $f$, by employing an unbiased estimator for this approximation. The best polynomial approximation for a function $f(x)$ on domain $A$ with order no more than $K$ is defined as
\begin{equation}
P_K^*(x) \triangleq \argmin_{P \in \mathsf{poly}_K} \max_{x\in A}|f(x) - P(x)|,
\end{equation}
where $\mathsf{poly}_K$ is the collection of polynomials with order at most $K$ on $A$. The part pertaining to the smooth region is estimated by a bias-corrected maximum likelihood estimator. We apply this procedure coordinate-wise based on the empirical distribution of each observed symbol, and finally sum the respective estimates.

We now look at the specific cases of entropy and $F_\alpha(P)$ separately. For the entropy, after we obtain the empirical distribution $P_n$, for each coordinate $P_n(i)$, if $P_n(i) \ll \ln n/n$, we (i) compute the best polynomial approximation for $-p_i \ln p_i$ in the regime $0\leq p_i \ll \ln n/n$, (ii) use the unbiased estimators for integer powers $p_i^k$ to estimate the corresponding terms in the polynomial approximation for $-p_i \ln p_i$ up to order $K_n \sim \ln n$, and (iii) use that polynomial as an estimate for $-p_i \ln p_i$. If $P_n(i) \gg \ln n/n$, we use the estimator $-P_n(i) \ln P_n(i) + \frac{1}{2n}$ to estimate $-p_i \ln p_i$. Then, we add the estimators corresponding to each coordinate. Our estimator for $F_\alpha(P)$ is very similar to that of entropy, with the only difference that we conduct polynomial approximation for $x^\alpha$ with order $K_n \sim \ln n$, and use the estimator $\left( 1 + \frac{\alpha(1-\alpha)}{2n P_n(i)} \right) P_n^\alpha(i)$ when $P_n(i) \gg \ln n/n$.

We remark that our estimator is both conceptually and algorithmically simple, with complexity linear in the number of samples $n$. Indeed, the only non-trivial computation required is the best polynomial approximation for functions, which is data independent and can be done \emph{offline} before obtaining any samples from the experiment. Moreover, the coefficients of the best polynomial approximation of different orders can be preprocessed and stored in advance in the implementation of our approach. We demonstrate in Section~\ref{sec.experiments} that the best polynomial approximation step can be performed efficiently using modern machinery from approximation theory and numerical analysis.

\subsection{Main results}\label{subsec.mainresults}
Simple as our estimators are to describe and implement, they can be shown to be near ``optimal'' in the strong sense we now describe. We adopt the conventional statistical decision theoretic framework~\cite{Wald1950statistical}. Regarding the task of estimating functional $F(P)$, the $L_2$ risk of an arbitrary estimator $\hat{F}$ is defined as
\begin{equation}
\bE_P \left( F(P) - \hat{F} \right)^2,
\end{equation}
where the expectation is taken with respect to the distribution $P$ that generates the observations used by $\hat{F}$. Apparently, the $L_2$ risk is a function of both the \emph{unknown} distribution $P$ and the estimator $\hat{F}$, and our goal is to minimize this risk. Since $P$ is unknown, we cannot directly minimize it, but if we want to do well no matter what the true distribution $P$ is, we may want to adopt the minimax criterion~\cite{Wald1950statistical}\cite{Lehmann--Casella1998theory}, and try to minimize the \emph{maximum risk}
\begin{equation}\label{eqn.maximumriskdef}
\sup_{P\in \mathcal{M}_S}  \bE_P \left( F(P) - \hat{F} \right)^2,
\end{equation}
where $\mathcal{M}_S$ denotes the set of all discrete distributions with support size $S$. The estimator that minimizes the maximum risk above is called the \emph{minimax} estimator, and the corresponding risk is called the \emph{minimax} risk. The exact computation of the minimax risk and the minimax estimator for general $F(P)$ seems intractable. Although the maximum risk in (\ref{eqn.maximumriskdef}) is a convex function of $\hat{F}$ (supremum of convex functions is convex), minimizing this function involves computation of the objective function via $\sup_{P\in \mathcal{M}_S}$, which is a non-convex optimization problem. Moreover, even if we can compute it exactly, the minimax estimator will surely depend on the support size $S$, which is unknown to the statistician in many applications.

Hence, we slightly relax the requirement, and seek \emph{minimax rate-optimal} estimators $\hat{F}^*$ with maximum (worst-case) risk equal to the minimax risk up to a multiplicative constant. In other words, we want to design estimator $\hat{F}^*$ such that there exist two universal positive constants $0<C_1 \leq C_2<\infty$ that do not depend on the problem configuration (such as the support size $S$ and sample size $n$), for which
\begin{align}
& C_1 \cdot \inf_{\hat{F}}\sup_{P\in \mathcal{M}_S}  \bE_P \left( F(P) - \hat{F} \right)^2  \leq
\sup_{P\in \mathcal{M}_S}  \bE_P \left( F(P) - \hat{F}^* \right)^2 \nonumber \\ 
 & \quad \quad \leq C_2 \cdot \inf_{\hat{F}}\sup_{P\in \mathcal{M}_S}  \bE_P \left( F(P) - \hat{F} \right)^2.
\end{align}

As it turns out, it is possible to construct estimators $\hat{F}^*$ for a wide class of functionals, which do not rely on the knowledge of support size $S$. A brief description of the constructions is given in Section~\ref{subsec.ourestimators}. We find it intriguing that our estimators, which are minimax rate-optimal, are intimately connected to the problem of best (minimax) polynomial approximation, which is a convex optimization problem. In some sense, we have transformed the difficult-to-solve minimax and convex problem of minimizing the maximum risk in (\ref{eqn.maximumriskdef}) into another efficently solvable minimax and convex problem of minimizing the maximum deviation of a polynomial from a given function, which turns out to have been studied extensively in approximation theory for more than a century.

To ease the presentation, we consider the ``Poissonized'' observation model \cite[Pg. 508]{LeCam1986asymptotic}, since we can show that the minimax risks under the Multinomial model and Poisson model are essentially the same (cf. Lemma~\ref{lemma.poissonmultinomial}). Moreover, adopting the Poisson model significantly reduces the length of the proofs, and we emphasize that similar analysis can also go through for Multinomial settings, with more nuanced analysis. In the Poisson setting, we first draw a Poisson random number $N \sim \mathsf{Poi}(n)$, and then conduct the sampling $N$ times. Consequently, the observed number of occurrences of each symbol are independent~\cite[Thm. 5.6]{Mitzenmacher--Upfal2005probability}.

We have the following characterization of the minimax risk for entropy estimation.

\begin{theorem}\label{thm.entropy.risk}
Suppose $n \gtrsim \frac{S}{\ln S}$. Then the minimax risk of estimating entropy $H(P)$ satisfies
\begin{equation} \label{eqn.entropy.risk}
\inf_{\hat{H}}\sup_{P \in \mathcal{M}_S} E_P\left( \hat{H} - H(P) \right)^2 \asymp \frac{S^2}{(n \ln n)^2} + \frac{(\ln S)^2}{n}.
\end{equation}
Our estimator achieves this bound without knowledge of the support size $S$ under the Poisson model.
\end{theorem}

The following is an immediate consequence of Theorem~\ref{thm.entropy.risk}.

\begin{corollary}\label{cor.entropy.risk}
For our entropy estimator, the maximum $L_2$ risk vanishes provided $n \gg  \frac{S}{\ln S}$. Moreover, if $n \lesssim \frac{S}{\ln S}$, then the maximum risk of any estimator for entropy is bounded from zero.
\end{corollary}

It was first shown in \cite{Valiant--Valiant2011} that one must have $n \gg \frac{S}{\ln S}$ for consistently estimating the entropy. However, the entropy estimators based on linear programming proposed in Valiant and Valiant~\cite{Valiant--Valiant2011, Valiant--Valiant2013estimating} have not been shown to achieve the minimax risk. Another estimator proposed by Valiant and Valiant~\cite{Valiant--Valiant2011power} has only been shown to achieve the minimax risk in the restrictive regime of $\frac{S}{\ln S} \lesssim n \lesssim \frac{S^{1.03}}{\ln S}$. Wu and Yang~\cite{Wu--Yang2014minimax} independently applied the idea of best polynomial approximation to entropy estimation, and obtained its minimax $L_2$ rates. The minimax lower bound part of Theorem~\ref{thm.entropy.risk} follows from Wu and Yang~\cite{Wu--Yang2014minimax}. We also remark that, unlike the estimator we propose, the estimator in Wu and Yang~\cite{Wu--Yang2014minimax} relies on knowledge of the support size $S$, which generally may not be known.

For the functional $F_\alpha(P),0<\alpha<1$, we have the following.
\begin{theorem}\label{thm.Falpha.risk}
Suppose $n \gtrsim \frac{S^{1/\alpha}}{\ln S}$ when we estimate $F_\alpha(P),0<\alpha<1$. Then we have the following characterizations of the minimax risk.
\begin{enumerate}
\item $0<\alpha\leq 1/2$. If we also have $\ln n \lesssim \ln S$, then
\begin{equation}\label{eqn.Falpha.risk}
\inf_{\hat{F}_\alpha} \sup_{P \in \mathcal{M}_S} \bE_P \left(\hat{F}_\alpha - F_\alpha(P) \right)^2 \asymp \frac{S^2}{(n \ln n)^{2\alpha}}.
\end{equation}
\item $1/2 <\alpha<1$.
\begin{equation}
\inf_{\hat{F}_\alpha} \sup_{P \in \mathcal{M}_S} \bE_P \left(\hat{F}_\alpha - F_\alpha(P) \right)^2 \asymp \frac{S^2}{(n \ln n)^{2\alpha}}  +  \frac{S^{2-2\alpha}}{n}.
\end{equation}
\end{enumerate}
Our estimators $\hat{F}_\alpha$ achieve this bound without knowledge of the support size $S$ under the Poisson model.
\end{theorem}

One immediate corollary of Theorem~\ref{thm.Falpha.risk} is the following.
\begin{corollary}\label{cor.Falpha.risk}
For our estimators of $F_\alpha$, the maximum $L_2$ risk vanishes provided $n \gg  \frac{S^{1/\alpha}}{\ln S},0<\alpha<1$. Moreover, if $n \lesssim \frac{S^{1/\alpha}}{\ln S}$, then the maximum risk of any estimator for $F_\alpha$ is bounded from zero.
\end{corollary}

The minimax lower bound we present in Theorem~\ref{thm.Falpha.risk}\footnote{In a previous version of the manuscript, there is a $\sqrt{\ln S}$ gap between our minimax lower bound and the achievability in Theorem~\ref{thm.Falpha.risk}. Partially inspired by Wu and Yang~\cite{Wu--Yang2014minimax}, we modified the proof by using an argument similar to that in the lower bound proof of \cite{Wu--Yang2014minimax}, thereby closing the gap.} significantly improves on Paninski's lower bound in \cite{Paninski2004}, which states that if $n \lesssim S^{1/\alpha -1}$, then the maximum $L_2$ risk of any estimator for $F_\alpha(P), 0<\alpha<1$, is bounded from zero.

The next two theorems correspond to estimation of $F_\alpha(P)$, $\alpha>1$.

\begin{theorem}\label{th_2}
Suppose $1<\alpha<\frac{3}{2}$. Under the Poissonized model, our estimator $\hat{F}_\alpha$ satisfies
    \begin{align}
  \limsup_{n\to \infty} \,(n\ln n)^{2(\alpha-1)} \cdot \sup_S \sup_{P \in \mathcal{M}_S}\mathbb{E}_P \left(\hat{F}_\alpha-F_\alpha(P)\right)^2 <\infty.
\end{align}
\end{theorem}

In other words, our estimator $\hat{F}_\alpha, 1<\alpha<3/2$ achieves an $L_2$ convergence rate of $(n \ln n)^{-2(\alpha-1)}$ regardless of the support size. This also turns out to be the minimax rate, as shown by the following result.

\begin{theorem}\label{th_1}
Suppose $1<\alpha<\frac{3}{2}$. There exists a universal constant $c_0>0$ such that if $S = c_0 n \ln n$ then
  \begin{align}
  \liminf_{n\to \infty} \,(n\ln n)^{2(\alpha-1)} \cdot \inf_{\hat{F}}\sup_{P\in\mathcal{M}_{S}}\mathbb{E}_P\left(\hat{F}-F_\alpha(P)\right)^2 > 0,
\end{align}
where the infimum is taken over all possible estimators $\hat{F}$.
\end{theorem}

Table~\ref{table.summary} summarizes the minimax $L_2$ rates and the $L_2$ convergence rates of the MLE in estimating $F_\alpha(P),\alpha>0$ and $H(P)$. When the $L_2$ rates have two terms, the first and second terms represent respectively the contributions of the bias and the variance. When there is a single term, only the dominant term is retained. Conditions for these results are presented in parentheses.

\begin{table*}[t]
\begin{center}
\hspace*{-0.5cm}
    \begin{tabular}{| l | l | l |}
    \hline
    & Minimax $L_2$ rates & $L_2$ rates of MLE\\ \hline
$H(P)$    & $\frac{S^2}{(n \ln n)^2} + \frac{\ln^2 S}{n} \quad \left( n \gtrsim \frac{S}{\ln S} \right)$ (Thm.~\ref{thm.entropy.risk},\cite{Wu--Yang2014minimax}) & $\frac{S^2}{n^2} + \frac{\ln^2 S}{n}\quad \left( n\gtrsim S \right)$ \cite{Jiao--Venkat--Weissman2014MLE}  \\ \hline

$F_\alpha(P), 0<\alpha\leq \frac{1}{2}$ &  $\frac{S^2}{(n \ln n)^{2\alpha}} \quad \left( n \gtrsim S^{1/\alpha}/\ln S, \ln n \lesssim \ln S \right)$ (Thm.~\ref{thm.Falpha.risk})  & $\frac{S^2}{n^{2\alpha}}\quad \left( n \gtrsim S^{1/\alpha} \right)$ \cite{Jiao--Venkat--Weissman2014MLE}     \\
    \hline

$F_\alpha(P), \frac{1}{2}<\alpha<1$ &    $\frac{S^2}{(n \ln n)^{2\alpha}} + \frac{S^{2-2\alpha}}{n}\quad \left( n \gtrsim S^{1/\alpha}/\ln S \right)$ (Thm.~\ref{thm.Falpha.risk})  & $\frac{S^2}{n^{2\alpha}} + \frac{S^{2-2\alpha}}{n}\quad \left( n \gtrsim S^{1/\alpha} \right)$ \cite{Jiao--Venkat--Weissman2014MLE}  \\ \hline

$F_\alpha(P), 1< \alpha<\frac{3}{2}$ & $(n \ln n)^{-2(\alpha-1)}\quad \left(S \gtrsim n\ln n \right)$ (Thm.~\ref{th_2},\ref{th_1})  & $n^{-2(\alpha-1)}\quad \left(S \gtrsim n \right)$ \cite{Jiao--Venkat--Weissman2014MLE}  \\ \hline

$F_\alpha(P), \alpha\geq \frac{3}{2}$ & $n^{-1}$ \cite{Jiao--Venkat--Weissman2014MLE}  & $n^{-1}$  \\ \hline
    \end{tabular}

    \caption{Summary of results in this paper and the companion~\cite{Jiao--Venkat--Weissman2014MLE}}
     \label{table.summary}
\end{center}
\end{table*}

From a sample complexity perspective (i.e. how should the number of samples $n$ scale with the support size $S$ to achieve consistent estimation), Table~\ref{table.summary} implies the results in Table~\ref{table.samplecomplexity}.

\begin{table}[h!]
\begin{center}
\hspace*{-0.5cm}
    \begin{tabular}{| l | l | l |}
    \hline
    & MLE & Minimax rate-optimal\\ \hline
$H(P)$    &  $n \gg S$ & $n \gg S/\ln S$ \\ \hline

$F_\alpha(P), 0<\alpha<1$ & $n \gg S^{1/\alpha}$ & $n\gg S^{1/\alpha}/\ln S$    \\
    \hline

$F_\alpha(P), \alpha>1$ &  $n \gg 1$ & $n \gg 1$\\ \hline
    \end{tabular}
    \caption{The number of samples needed to achieve consistent estimation}
    \label{table.samplecomplexity}
\end{center}
\end{table}

Our work (including the companion paper~\cite{Jiao--Venkat--Weissman2014MLE}) is the first to obtain the minimax rates, minimax rate-optimal estimators, and the maximum risk of MLE for estimating $F_\alpha(P), 0<\alpha<3/2$, and entropy $H(P)$ in the most comprehensive regime of $(S,n)$ pairs. Evident from Table~\ref{table.summary} is the fact that the MLE cannot achieve the minimax rates for estimation of $H(P)$, and $F_\alpha(P)$ when $0 < \alpha < 3/2$. In these cases, our estimators have performance with $n$ samples essentially the same as the MLE with $n\ln n$ samples, and it is the best possible. In other words, the minimax rate-optimal schemes \emph{enlarge} the ``effective sample size'' from $n$ to $n\ln n$. Furthermore, all the improvements we have are in the bias, which is the dominating factor in the risk. This observation suggests a simple way to obtain the minimax $L_2$ rates from the $L_2$ rates of the MLE. One need merely find the bias term in the expression of MLE $L_2$ rates, and replace the term $n$ by $n\ln n$. This simple rule is intimately connected to the rationale behind the construction and analysis of our estimators, on which we elaborate in Section~\ref{sec.motivation}.

We also note that Table~\ref{table.samplecomplexity} is a \emph{``lossy compression''} of Table~\ref{table.summary}. Indeed, it did not reflect the important improvement of our estimator over the MLE in estimating $F_\alpha(P), 1<\alpha<3/2$. However, it is more transparent about the increase of difficulty in estimation when we decrease $\alpha$. Indeed, when $\alpha\to 0^+$, the sample complexity in estimating $F_\alpha(P)$, $S^{1/\alpha}/\ln S$ becomes super-polynomial in $S$, which implies that the problem has become extremely challenging. Indeed, the limiting case of $\alpha = 0$ corresponds to estimating the support size of a discrete distribution, which has long been known impossible to do consistently without additional assumptions~\cite{Efron--Thisted1976}\cite{Bunge--Fitzpatrick1993estimating}.

\subsection{Discussion of main results}\label{subsec.discussionmainresults}

Within the scope of estimating entropy of discrete distributions from i.i.d.\ samples, the reader should be aware of different problem formulations, so as not to be confused by seemingly contradictory results. The problem we consider in this paper is to estimate the entropy $H(P)$ for \emph{all} possible distributions $P$ supported on $S$ elements, a setting for which we obtain Theorem~\ref{thm.entropy.risk}. However, one may impose some additional structure on the distribution $P$, and thus restrict attention to smaller uncertainty sets of distributions. One may expect different answers depending on the size and nature of the uncertainty sets. One of the popular alternative settings is to assume that the distribution $P$ comes from a \emph{uniform} distribution with unknown support size $S$. Note that it is a great simplification of the problem, indeed, there is only one parameter $S$ to estimate. Correspondingly, we only need $n \gg \sqrt{S}$ samples to consistently estimate the support size, and the entropy under this setting~\cite{Santhanam--Orlitsky--Viswanathan2007new}, which is much smaller than the required $n \gg \frac{S}{\ln S}$ samples in our setting, cf. Theorem~\ref{thm.entropy.risk}.

Some readers may be concerned that the minimax decision theoretic framework we adopt is too pessimistic. In some sense, it characterizes the worst-case performance over all possible distributions $P \in \mathcal{M}_S$, and it would be disappointing if our estimator fails to behave reasonably for distributions lying in a strict subset of $\mathcal{M}_S$ not including the worst case distribution. Regarding this question, Brown~\cite{Brown1994minimaxity,Brown2000essay} argued that the minimax idea has been an essential foundation for advances in many areas of statistical research, including general asymptotic theory and methodology, hierarchical models, robust estimation, optimal design, and nonparametric function analysis. Second, the statistics community in general uses the \emph{adaptive} estimation framework to alleviate the pessimism of minimaxity~\cite{Cai2012minimax}. Specifically, one specifies a nested sequence of subsets of $\mathcal{M}_S$, and tries to construct an estimator that achieves simultaneously the minimax rates over each of the subsets without knowing the subset to which the active parameter $P$ actually belongs. It was shown recently in another related paper~\cite{Han--Jiao--Weissman2015adaptive} that along the nested subsets $\mathcal{M}_S(H) = \{P: H(P) \leq H, P \in \mathcal{M}_S\}$, our estimator (without knowing $H$ nor $S$) simultaneously achieves the minimax rates over $P\in \mathcal{M}_S(H)$ for all $H \leq \ln S$. Most surprisingly, the maximum risk of our estimator over $\mathcal{M}_S(H)$ for every $S$ and $H$ with $n$ samples is still essentially that of the MLE with $n\ln n$ samples, further reinforcing the effectiveness of our estimator.

It is instructive to consider our results in the context of the intriguing connections and differences between three important problems in information theory: entropy estimation, estimating a discrete distribution under relative entropy loss, and minimax redundancy in compressing i.i.d. sources. Table~\ref{table.comparethree} summarizes the known results.

\begin{table*}[t]
\begin{center}
    \begin{tabular}{| l | l | l | l |}
    \hline
    & entropy estimation & estimation of distribution   & compression with blocklength $n$\\ \hline
   $S\textrm{ fixed}$  & $\text{MSE}\sim \frac{\mathsf{Var}(-\ln P(X))}{n} $ \cite{LeCam1986asymptotic} &  $ \inf_{\hat{P}}\sup_{P }\bE D(P_X \| \hat{P}_X) \sim \frac{S-1}{2n}       $ \cite{Chentsov1972statisticheskie,Braess--Dette2004asymptotic} & $\inf_{Q} \sup_{P } \frac{1}{n}D(P_{X^n}\|Q_{X^n}) \sim \frac{S-1}{2n}\ln n$  \cite{Rissanen1986stochastic} \\ \hline
 large $S$ &
   $n \gg S/\ln S$ \cite{Valiant--Valiant2011} & $n \gg S$ \cite{Paninski2004variational} & $n \gg S$ \cite{Orlitsky--Santhanam2004speaking,Szpankowski--Weinberger2012minimax} \\
    \hline
    \end{tabular}

   \caption{Comparison of difficulties in entropy estimation, estimation of distribution, and data compression under classical asymptotics and high dimensional asymptotics}
    \label{table.comparethree}
\end{center}
\end{table*}

Table~\ref{table.comparethree} conveys several important messages. First, in the asymptotic regime, there is a logarithmic factor between the redundancy of the compression and distribution estimation problems. Indeed, since compression requires use of a coding distribution $Q$ that does not depend on the data, the redundancy of compression will definitely be larger than the risk under relative entropy in estimating the distribution. However, in the large alphabet setting, the problems are equally difficult - the phase transition of vanishing risk for both compression and distribution estimation happen when $n$ is linear in the support size $S$.

Second, the large alphabet setting shows that estimation of entropy is considerably easier than both estimating the corresponding distribution, and compression. It is somewhat surprising and enlightening, since there has been a well-received tradition to apply data compression techniques to estimate entropy, even beyond the information theory community, e.g. \cite{Song--Qu--Blumm--Barabsi2010limits,Takaguchi--Nakamura--Sato--Yano--Masuda2011predictability}, whereas one of the implications of Table~\ref{table.comparethree} is that the approach of entropy estimation via compression can be highly sub-optimal.

%
%

If we plot the phase transitions of $\ln n/\ln S$ for estimating $F_\alpha(P)$ using $F_\alpha(P_n)$ with respect to $\alpha$, we obtain Figure~\ref{fig.falphaphase}.

\begin{center}
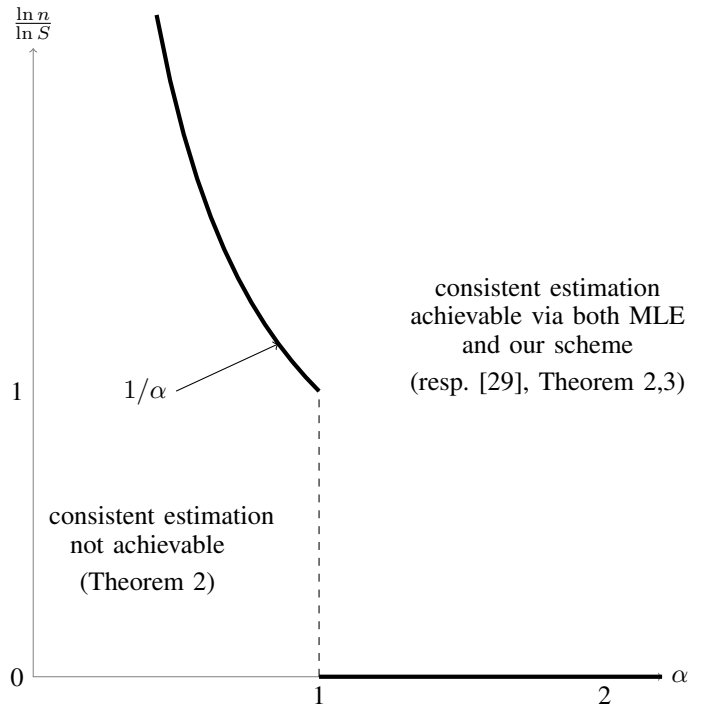

  \centering
\begin{tikzpicture}[xscale=3.8,yscale=3.8]
\draw [<->, help lines] (0.6,2.8) -- (0.6,0.6) -- (2.8,0.6);
\draw  [ultra thick](1.0316, 2.9171) -- (1.0789, 2.6879) -- (1.1263, 2.5000 ) -- (1.1737, 2.3431 )-- (1.2211, 2.2102) -- (1.2684,2.0961 ) -- (1.3158,1.9971 ) -- (1.3632,1.9103 ) -- (1.4105,1.8338 ) -- (1.4579,1.7656 ) -- (1.5053, 1.7047) -- (1.5526,1.6497 ) -- (1.6000, 1.6);
\draw [dashed] (1.6, 0.6) -- (1.6,1.6);
\draw [ultra thick](1.6, 0.6) -- (2.8, 0.6);
\node [above] at (1.05, 1.1) {consistent estimation};
\node [above] at (1,1) {not achievable};
\node [below] at (1,1) {(Theorem~\ref{thm.Falpha.risk})};
\node [left] at (0.6, 1.6) {1};
\node [left] at (0.6,0.6) {0};
\node [below] at (1.6, 0.6) {1};
\node [below] at (2.6,0.6) {2};
\node [right] at (2.8,0.6) {$\alpha$};
\node [above] at (0.6,2.8) {$\frac{\ln n}{\ln S}$};
\draw [->] (1.1,1.6) -- (1.4579,1.7656 );
\node [left] at (1.1, 1.6) {$1/\alpha$};
\node [above] at (2.4, 1.9) {consistent estimation};
\node [above] at (2.4, 1.8) {achievable via both MLE};
\node [above] at (2.4, 1.7) {and our scheme};
\node [below] at (2.4, 1.7) {(resp.~\cite{Jiao--Venkat--Weissman2014MLE}, Theorem~\ref{thm.Falpha.risk},\ref{th_2})};
\end{tikzpicture}
\captionof{figure}{For any fixed point above the thick curve, consistent estimation of $F_\alpha(P)$ is achieved using MLE $F_\alpha(P_n)$ \cite{Jiao--Venkat--Weissman2014MLE}. Our estimator slightly improves over MLE to achieve the optimal $n \asymp S^{1/\alpha}/\ln S$ sample complexity when $0<\alpha<1$. For the regime $0<\alpha<1$ below the thick curve, Theorem~\ref{thm.Falpha.risk} shows that no estimator can have vanishing maximum $L_2$ risk. }
\label{fig.falphaphase}
\end{center}

We observe a sharp phase transition at $\alpha=1$, as the sample size requirement shifts from $n \gg S^{\frac{1}{\alpha}}/\ln S$ to  $n \gg 1$, depending on whether $\alpha$ is in the left or right neighborhood of 1, respectively. Hence, $\alpha=1$ is a critical point in that consistent estimation requires a number of measurements super-linear or constant in the size of the alphabet according to whether $\alpha < 1$ or $\alpha > 1$.

Combining Table~\ref{table.comparethree} and Figure~\ref{fig.falphaphase} leads to the interesting observation that, in high dimensional asymptotics, estimating a functional of a distribution could be easier (e.g. $H(P), F_\alpha(P), \alpha>1$) or harder (e.g. $F_\alpha(P), 0<\alpha<1$) than estimating the distribution itself. This observation taps into another interesting interpretation of the functional $F_\alpha(P)$. In information theory, the random variable $\imath(X) = \ln \frac{1}{P(X)}$ is known as the \emph{information density}, and plays important roles in characterizing higher order fundamental limits of coding problems \cite{Polyanskiy--Poor--Verdu2010channel,Kontoyiannis--Verdu2013optimal}. The functional $F_\alpha(P)$ can be interpreted as the moment generating function for random variable $\imath(X)$ as
\begin{equation}
F_\alpha(P) = \bE_P \left[e^{(1-\alpha) \imath(X)}\right].
\end{equation}
It is shown in Valiant and Valiant \cite{Valiant--Valiant2011} that the distribution of $\imath(X)$ can be estimated using $n \gg S/\ln S$ samples. Since moment generating functions can determine the distribution under some conditions, it is indeed plausible to see that the problem of estimating $F_\alpha(P)$, or the moment generating function of $\imath(X)$, is either easier or harder than estimating the distribution of $\imath(X)$ itself for various values of $\alpha$.

We now briefly shift our focus towards estimation of R\'enyi entropy $H_\alpha(P)$, which is closely related to the functional $F_\alpha(P)$ via $H_\alpha(P) = \frac{\ln F_\alpha(P)}{1-\alpha}$. Acharya et al.\cite{Acharya--Orlitsky--Suresh--Tyagi2014complexity} considered the estimation of $H_\alpha(P)$, and demonstrated that the sample complexity for estimating $H_\alpha(P)$ may exhibit a different behavior than that of estimating $F_\alpha(P)$ for certain values of $\alpha$. It was also shown in \cite{Acharya--Orlitsky--Suresh--Tyagi2014complexity} that for $\alpha>1, \alpha \in \mathbb{Z}^{+}$, the sample complexity is $n\asymp S^{1-1/\alpha}$, which can be achieved by a bias-corrected MLE. Second, for non-integer $\alpha>1$, \cite{Acharya--Orlitsky--Suresh--Tyagi2014complexity} showed that the sample complexity for estimating $H_\alpha(P)$ is between $S^{1-\eta},\forall \eta>0$ and $S$, and that it suffices to take $n \gg S$ samples for the MLE to be consistent. By a partial application of results from the present paper, \cite{Acharya--Orlitsky--Suresh--Tyagi2014complexity} also showed that for $0<\alpha<1$, the sample complexity for estimating $H_\alpha(P)$ is between $S^{1/\alpha-\eta},\forall \eta>0$ and $S^{1/\alpha}/\ln S$.  However, certain questions remain unanswered. For example, it was not clear, for $\alpha>1, \alpha\notin \mathbb{Z}^+$, whether the MLE indeed requires $n \gg S$ samples, and whether there exist estimators that can consistently estimate $H_\alpha(P)$ with $n \ll S$ samples. We provide partial answers to these questions below by focusing on the case when $1 < \alpha < 3/2$. First, we show  in Theorem~\ref{th_3} that simply plugging in the novel estimator $\hat{F}_\alpha$ from Theorem~\ref{th_2} to the definition of $H_\alpha(P)$ results in an estimator that needs at most $S/\ln S$ samples when $1<\alpha<\frac{3}{2}$.

\begin{theorem}\label{th_3}
For any $\alpha\in(1,3/2)$ and any $\delta>0,\epsilon\in(0,1)$, there exists a constant $c=c_\alpha(\delta,\epsilon)>0$ such that,
  \begin{align}
  \limsup_{S\to\infty}\sup_{P\in\mathcal{M}_S, n\geq \frac{c S}{\ln S}} \bP\left(\left|\frac{\ln \hat{F}_\alpha}{1-\alpha}-H_\alpha(P)\right|\ge \delta\right) \le\epsilon,
\end{align}
where $\hat{F}_\alpha$ is the estimator from Theorem~\ref{th_2}.
\end{theorem}

In words, with high probability $(\ln \hat{F}_\alpha)/(1-\alpha)$ is close to the R\'enyi entropy provided $n \gtrsim S/\ln S$. In contrast, the MLE requires $n \gtrsim S$ samples for estimating $H_\alpha(P), 1<\alpha<\frac{3}{2}$, as is implied by the following theorem.

\begin{theorem}\label{cor_3}
  For any $\alpha\in(1,3/2)$ and any constant $c>0$, there exist some $\delta=\delta_\alpha(c)>0$ such that the MLE $H_\alpha(P_n)$ satisfies
    \begin{align}
  \liminf_{n\to\infty} \inf_{S \geq n/c} \sup_{P\in\mathcal{M}_S} \mathbb{P}\left(\left|H_\alpha(P_n)-H_\alpha(P)\right|\ge \delta\right) =1,
\end{align}
where $P_n$ is the MLE of $P$.
\end{theorem}

To conclude this discussion, we conjecture that plugging in our minimax rate-optimal estimators for $F_\alpha(P)$ into the definition of $H_\alpha(P)$ results in minimax rate-optimal estimators for $H_\alpha(P)$ for all $\alpha>0$.

\section{Motivation, methodology, and related work}\label{sec.motivation}

\subsection{Motivation}\label{subsec.motivation}

Existing theory proves inadequate for addressing the problem of estimating functionals of probability distributions. A natural estimator for functionals of the form~(\ref{eqn.generalf}) is the maximum likelihood estimator (MLE), or plug-in estimator, which simply evaluates $F({P}_n)$, where ${P}_n$ is the empirical distribution of the data. How well does the MLE perform? Interestingly, if $f \in C^1(0,1]$ and we focus on $n$ i.i.d.\  observations from a distribution with support size $S$, then the problem of estimating $F(P)$ becomes a classical problem when $S$ is fixed, and the number of observations $n\to \infty$. This maximum likelihood estimator is \emph{asymptotically efficient} \cite[Thm. 8.11, Lemma 8.14]{Vandervaart2000} in the sense of the H\'ajek convolution theorem \cite{Hajek1970characterization} and the H\'ajek--Le Cam local asymptotic minimax theorem \cite{Hajek1972local}. It is therefore not surprising to encounter the following quote from the introduction of Wyner and Foster \cite{Wyner--Foster2003lower} who considered entropy estimation:
\begin{center}
\parbox{.45\textwidth}{~~\emph{  ``The plug-in estimate is universal and optimal not only for finite alphabet i.i.d. sources but also for finite alphabet, finite memory sources. On the other hand, practically as well as theoretically, these problems are of little interest.
"}}
\end{center}

In light of this, is it fair to say that the entropy estimation problem is solved in the finite alphabet setting? It was observed in Paninski~\cite{Paninski2003} that the maximum of $\mathsf{Var}(-\ln P(X))$ over distributions with support size $S$ is of order $(\ln S)^2$ (a tight bound is also given by Lemma~\ref{lemma.varentropy} in the appendix). Since classical asymptotics (with the Delta method \cite[Chap. 3]{Vandervaart2000}) show that
\begin{equation}\label{eqn.approximationmle}
\bE_P (H(P_n) - H(P))^2 \sim \frac{\mathsf{Var}(-\ln P(X))}{n},\quad n \gg 1,
\end{equation}
a naive interpretation of~(\ref{eqn.approximationmle}) might be that it suffices to take $n \gg (\ln S)^2$ samples to guarantee the consistency of $H(P_n)$. Such an interpretation could however be very misleading. It was already observed in Paninski \cite{Paninski2003} that if $n  \lesssim S^{1-\delta},\delta>0$, then the maximum  $L_2$ risk of any entropy estimator would be unbounded as $S$ grows.

This apparent discrepancy shows that (\ref{eqn.approximationmle}) is not valid when $S$ might be growing with $n$, and it is of utmost importance to obtain risk bounds for estimators of entropy and other functionals of distributions in the latter regime. Indeed, in the modern era of high dimensional statistics, we often encounter situations where the support size is comparable to, or much larger than the number of observations. For example, half of the words in the collected works by Shakespeare appeared only once~\cite{Efron--Thisted1976}.

It was shown in the companion paper~\cite{Jiao--Venkat--Weissman2014MLE} that for $n \gtrsim S$, the maximum risk of the MLE $H(P_n)$ can be written as
\begin{equation}\label{eqn.mlemaximumriskentropy}
\sup_{P \in \mathcal{M}_S} \bE_P (H(P_n) - H(P))^2 \asymp \frac{S^2}{n^2} + \frac{(\ln S)^2}{n},
\end{equation}
where the first term corresponds to the squared bias (defined as $(\bE_P H(P_n) - H(P))^2$), and the second term corresponds to the variance (defined as $\bE_P (H(P_n) - \bE_P H(P_n))^2$). Then we can understand this mystery: when we fix $S$ and let $n \to \infty$, the variance dominates and we get the expression in~(\ref{eqn.approximationmle}). However, when $n$ and $S$ may grow together, the bias term will not vanish unless $n \gg S$. Thus we conclude that in the large alphabet setting of entropy estimation, it is the bias that dominates the risk, and we have to reduce bias to improve the estimation accuracy.

The fact that the bias dominates the risk in entropy estimation in the large alphabet setting has been known, see~\cite{Miller1955,Paninski2003}. However, a general recipe to overcome the bias has defied many attempts. We briefly review some of the approaches in the literature.

One of the earliest investigations on reducing the bias of MLE in entropy estimation is due to Miller~\cite{Miller1955}, who showed that, for any fixed distribution $P$ supported on $S$ elements, for each symbol $i, 1\leq i\leq S$, we have
\begin{equation}\label{eqn.miller}
\bE_P \left( -P_n(i) \ln P_n(i) \right) = -p_i \ln p_i - \frac{1-p_i}{2n} +  O\left(\frac{1}{n^2}\right).
\end{equation}
Summing up both sides shows that $\bE H(P_n)$ is nearly $H(P) - \frac{S-1}{2n}$ up to the error term $O\left(\frac{1}{n^2}\right)$. Hence, the so-called Miller--Madow bias-corrected entropy estimator is defined by $H(P_n) + \frac{S-1}{2n}$, whose maximum risk in estimating entropy was shown to be bounded from zero if $n \lesssim S$ by~\cite{Paninski2003}\cite{Jiao--Venkat--Weissman2014MLE}. Thus, it still requires $n \gg S$ samples to consistently estimate the entropy, which is the same as MLE.

Another popular approach in estimating entropy is based on using the Dirichlet prior smoothing. Dirichlet smoothing may carry two different meanings in terms of entropy estimation: 
\begin{itemize}
\item \cite{Schurmann--Grassberger1996entropy,Schober2013some} One first obtain a Bayes estimate for the discrete distribution $P$, which we denote by $\hat{P}_B$, and then plugs it in the entropy functional to obtain the entropy estimate $H(\hat{P}_B)$.
\item \cite{Wolpert--Wolf1995,Holste--Grosse--Herzel1998bayes} One calculates the Bayes estimate for entropy $H(P)$ under Dirichlet prior for squared error. The estimator is the conditional expectation $\bE[H(P)|\mathbf{X}]$, where $X$ represents the samples, and the distribution $P$ follows a certain Dirichlet prior. 
\end{itemize}
It was shown in~\cite{Han--Jiao--Weissman2015Bayes} that both approaches require at least $n\gg S$ to be consistent. 

The jackknife~\cite{Miller1974jackknife} is another popular technique in reducing the bias. However, it was shown by Paninski~\cite{Paninski2003} that the jackknifed MLE also requires $n\gg S$ samples to consistently estimate entropy.

Given the fact that $n \gg \frac{S}{\ln S}$ is both necessary and sufficient, the approaches mentioned above are far from optimal. We note that the problem of estimating entropy of discrete distributions from i.i.d. observations on large alphabets has been investigated in various disciplines by many authors, with many approaches difficult to analyze theoretically. Among them we mention the Miller--Madow bias-corrected estimator and its variants \cite{Miller1955,Carlton1969bias,Grassberger1988finite}, the jackknified estimator \cite{Zahl1977jackknifing}, the shrinkage estimator \cite{Hausser--Strimmer2009entropy}, the Bayes estimator under various priors \cite{Wolpert--Wolf1995, Nemenman--Shafee--Bialek2002entropy}, the coverage adjusted estimator \cite{Chao--Shen2003nonparametric}, the Best Upper Bound (BUB) estimator \cite{Paninski2003}, the B-Splines estimator \cite{Daub--Steuer--Selbig--Kloska2004}, and~\cite{Grassberger2008entropy}\cite{Nemenman2011coincidences}\cite{Nemenman--Bialek--vanSteveninck2004entropy}\cite{Vinck--Battaglia--Balakirsky--Vinck--Pennartz2012estimation} etc.

In what follows, we explain in detail our step by step approach to this problem, and arrive at minimax rate-optimal estimators for the entire range of functional estimation problems considered.

\subsection{How did we come up with our scheme?}\label{subsec.howdidwe}

Existing literature implied that it is possible to come up with consistent entropy estimators that require sublinear $n \ll S$ samples. The earliest indication to this effect appeared in Paninski \cite{Paninski2004}, but only an existential proof based on the Stone--Weierstrass theorem was provided. It was therefore a breakthrough when Valiant and Valiant \cite{Valiant--Valiant2011} introduced the first explicit entropy estimator requiring a sublinear number of samples.  They~\cite{Valiant--Valiant2011} showed that $n \gg S / \ln S$ samples are both necessary and sufficient to consistently estimate the entropy of a discrete distribution. However, the entropy estimators based on linear programming proposed in Valiant and Valiant~\cite{Valiant--Valiant2011, Valiant--Valiant2013estimating} have not been shown to achieve the minimax rate. Another estimator proposed by Valiant and Valiant~\cite{Valiant--Valiant2011power} has only been shown to achieve the minimax rate in the restrictive regime of $\frac{S}{\ln S} \lesssim n \lesssim \frac{S^{1.03}}{\ln S}$. Moreover, the scheme of \cite{Valiant--Valiant2011} can only be applied to functionals that are Lipschitz continuous with respect to a Wasserstein metric, which can be roughly understood as those functionals that are equally ``smooth'' or ``smoother'' than entropy. Notably, this does not include the functional $F_\alpha, \alpha < 1$ and other interesting nonsmooth functionals of distributions. Also, it is not clear whether these techniques generally lead to minimax rate-optimal estimators. Readers are referred to Valiant's thesis \cite{Valiant2012algorithmic} for more details.

Conceivably, there is a fundamental connection between the smoothness of a functional, and the hardness of estimating it. The ideal solution to this problem would be systematic and capture this trade-off for nearly every functional. Such a comprehensive view of functional estimation has yet to be realized. George P\'olya \cite{Polya2014solve} commented that ``the more general problem may be easier to solve than the special problem''. This motivated our present work, in which we provide a general framework and procedure for minimax estimation of functionals with non-asymptotic performance guarantees. To make things transparent, let us now start from scratch and demonstrate how our solution has a natural construction.

Suppose we would like to propose a general method to construct minimax rate-optimal estimators for functionals of the form~(\ref{eqn.generalf}). What are the \emph{prerequisites} that any method must satisfy? Based on our analysis above, the following criteria appear natural:
\begin{enumerate}
\item \emph{Asymptotic efficiency}. As modern asymptotic statistics~\cite{Vandervaart2000} tells us, if the function $f(p)$ in (\ref{eqn.generalf}) is differentiable on $(0,1]$, then the MLE $F(P_n)$ is asymptotically efficient. In other words, no matter how we adjust the MLE in finite sample settings, we have to ensure that when the number of samples $n$ go to infinity while $S$ remains fixed, our estimator is very similar to the MLE.
\item \emph{Bias reduction}. As our analysis of MLE~\cite{Jiao--Venkat--Weissman2014MLE} indicates, the MLE usually has large bias and small variance in high dimensions. Hence, the general method has to reduce bias in finite samples.
\end{enumerate}

Let us attempt to understand an estimator's bias more carefully. In the simplest setting, consider a Binomial random variable $X\sim \mathsf{B}(n,p)$, and suppose we wish to estimate the scalar $f(p)$ based on $X$. Denote by $g(X)$ an arbitrary estimator for $f(p)$. The bias of $g(X)$ can be written as
\begin{align}
\mathsf{Bias}_p(g(X)) & = \bE_p g(X) - f(p) \nonumber \\
& = \left(\sum_{j = 0}^n g(j) \binom{n}{j} p^j (1-p)^{n-j}\right) - f(p).\label{eqn.biasgeneralfwithg}
\end{align}

Equation~(\ref{eqn.biasgeneralfwithg}) conveys two important messages. First, the only form of $f(p)$ that can be estimated without bias is polynomials with order no more than $n$. Indeed, if $g(X)$ is an unbiased estimator for $f(p)$, then $\mathsf{Bias}_p (g(X)) =0$, for all $p\in [0,1]$. Thus, $f(p)$ is a polynomial of $p$ with order no more than $n$ because $\bE_p g(X)$ is such a polynomial. Conversely, any polynomial of $p$ whose order is no more than $n$ can be estimated without bias using $X$. Indeed, for all $0\leq r\leq n$,
\begin{equation}\label{eqn.unbiasedpolynomial}
\bE_p \left[ \frac{X(X-1)\cdot\cdots\cdot(X-r+1)}{n(n-1)\cdot \cdots \cdot (n-r+1)} \right] = p^r.
\end{equation}

Second, the bias $\mathsf{Bias}_p(g(X))$ as a function of $p$ corresponds to a polynomial approximation error. In other words, it is the difference between a function $f(p)$ and a polynomial $\bE_p g(X)$. This viewpoint was first proposed by Paninski~\cite{Paninski2003}, who made the important connection between the analysis of bias and approximation theory. Starting from the seminal work of Chebyshev, a central problem in approximation theory~\cite{Devore--Lorentz1993}, polynomial approximation, is targeted at designing polynomials that approximate any continuous function as well as possible. It then appears natural to choose the coefficients $\{g(j)\}_{j = 0}^n$ in a way that the resulting polynomial $\bE_p g(X)$ approximates the function $f(p)$ optimally. 

One may initially be tempted to use the Taylor series to approximate $f(p)$. However, a more careful inspection indicates that the Taylor polynomial is inappropriate for approximating general continuous functions. Even setting aside questions of convergence, Taylor polynomials are not defined for functions that are not infinitely differentiable. Even for functions that are analytic (such as $e^x, x\in [-1,1]$), one can show that truncating the Taylor series up to order $n$ results in maximum error $\sim \frac{1}{(n+1)!}$ on $[-1,1]$, but there exists a polynomial with order $n$ whose maximum approximation error is asymptotically $\frac{1}{2^n(n+1)!}$~\cite{Devore--Lorentz1993}, which is the so called \emph{best approximation polynomial}. The best polynomial approximation is targeted at computing the polynomial that minimizes the maximum deviation of the polynomial from the function $f(p)$. It is known that for any continuous function on a compact interval, there exists a unique best approximation polynomial for any order. Adopting this rationale, we may try to solve the following problem:
\begin{equation}\label{eqn.polyapproximationgbias}
g^* = \argmin_g \max_{p \in [0,1]} |\mathsf{Bias}_p(g(X))|,
\end{equation}
where we seek $g^*$ minimizing the maximum value of $|\mathsf{Bias}_p(g(X))|$. It gives us the best uniform control of the bias since we do not know $p$ a priori.

Applying advanced tools from approximation theory, Paninski~\cite{Paninski2003} tried the idea mentioned above, which unfortunately did not result in improved estimators. It turns out that this idea, while improving significantly in the bias, results in a blowing up of the variance term. Indeed, the squared bias of the estimator designed above can be shown to be $S^2/n^{4}$. Taking $n \gg \sqrt{S}$, the bias term will vanish, but the variance term will diverge, because Paninski~\cite{Paninski2003} already showed that if $n \lesssim S^{1-\delta}$, for any $\delta>0$, the maximum $L_2$ risk of any estimators for entropy will be bounded from zero.

In fact, there is a simpler way to understand why the global scale polynomial approximation idea of the form~(\ref{eqn.polyapproximationgbias}) does not work. It is destined to fail because it violates the first prerequisite of any general method to improve MLE in functional estimation. Indeed, this scheme does not behave like the MLE even if $n \gg S$.

This observation leads us to combine some core ideas that finally constitute our scheme. First, one needs to use approximation theory to reduce bias. Second, one cannot do approximation on a global scale (such as $p\in [0,1]$), but can only approximate the function $f(p)$ locally. Fortunately, the measure concentration phenomenon allows us to do approximation locally. For example, upon observing $\hat{p} = X/n, X\sim \mathsf{B}(n,p)$, we have $\mathsf{Var}(\hat{p}) = \frac{p(1-p)}{n}$, which vanishes as $n\to \infty$. Finally, where should we approximate? Intuitively, the bias is mainly due to the set of points where the function $f(p)$ changes abruptly. For $f(p) = -p\ln p$ or $p^\alpha,\alpha>0$, the most ``nonsmooth'' point is $p = 0$.

To sum up, we need to approximate locally around the ``nonsmooth'' points to reduce bias. Natural as this statement may seem, there are some parameters to be carefully specified. For this subsection we only consider $f(p) = -p\ln p$ or $p^\alpha,\alpha>0$. We detail the construction of our scheme by posing the following natural questions:
\begin{enumerate}
\item If we approximate function $f(p)$ in interval $[0,\Delta_n]$, how should we choose $\Delta_n$?
\item If we use a polynomial with order $K_n$ to approximate $f(p)$ in $[0,\Delta_n]$, how should we choose $K_n$? What should we do after obtaining the polynomial?
\item What should we do in interval $p\in [\Delta_n,1]$?
\end{enumerate}

Let us now answer these questions in the order in which they were asked. The value $\Delta_n$ should always be chosen to be the smallest number such that we can \emph{localize} the parameter $p$. In other words, suppose we observe $X\sim \mathsf{B}(n,p)$. Then, $\Delta_n$ should be chosen to ensure that if $X/n \leq c \Delta_n$, $c>0$ is a constant, then $p \in  [0,\Delta_n]$ with high probability. Similarly, if $X/n \geq C \Delta_n$, $C>0$ is a constant, then $p \in [\Delta_n,1]$ with high probability. It turns out that $\Delta_n \asymp \frac{\ln n}{n}$ fulfills this goal (cf. Lemma~\ref{lemma.poissontail}).

Regarding the second question, the value $K_n$ should always be chosen to be the largest number such that the increased variance does not exceed the bias. Indeed, if we use order $n$ approximation, then we essentially go back to the idea~(\ref{eqn.polyapproximationgbias}) that increases the variance too much such that the resulting estimator does not have vanishing risk. It turns out for $H(P)$ and $F_\alpha(P)$, $K_n \asymp \ln n$ is the correct order for which we need to conduct the best polynomial approximation.  Suppose we have obtained the best polynomial approximation of order $K_n$ for function $f(p)$ over regime $[0,\Delta_n]$. Noting that any polynomial of order no more than $n$ can be estimated without bias using the estimator in~(\ref{eqn.unbiasedpolynomial}), we use the corresponding unbiased estimator to estimate this $K_n$-order polynomial, thereby  ensuring that the bias of this estimator when $p\in [0,\Delta_n]$ is exactly the polynomial approximation error in approximating $f(p)$ over $[0,\Delta_n]$.

The third question refers to the scheme in the ``smooth'' regime. Interestingly, it was already observed in 1969 by Carlton~\cite{Carlton1969bias} that Miller's bias correction formula~(\ref{eqn.miller}) should \emph{only} be applied when $p_i \gg \frac{1}{n}$. In other words, Miller's formula~(\ref{eqn.miller}) is relatively accurate when $p_i \gg \frac{1}{n}$. In our case, since we have already chosen $\Delta_n \asymp \frac{\ln n}{n}$, in the smooth regime we have $p_i \gtrsim \frac{\ln n}{n}$. We use the first order bias-correction in this regime inspired by Miller, whose rationale is the following.

For Binomial random variable $X \sim \mathsf{B}(n,p)$, denote the empirical frequency by $\hat{p} = \frac{X}{n}$. Then it follows from Taylor's theorem that
\begin{align}
\bE f(\hat{p}) - f(p) & = \frac{1}{2} f''(p) \mathsf{Var}_p(\hat{p}) + O \left(\frac{1}{n^2}\right)  \\
& = \frac{f''(p) p(1-p)}{2n} + O \left(\frac{1}{n^2}\right),
\end{align}
where $f''(p)$ is the second derivative of $f(p)$. We define the first order bias-corrected estimator of $f(\hat{p})$ by
\begin{equation}\label{eqn.biascorrectiongeneralf}
f^c(\hat{p}) = f(\hat{p}) - \frac{f''(\hat{p}) \hat{p}(1-\hat{p})}{2n}.
\end{equation}

Taking $f(p) = -p\ln p$, $f^c(\hat{p}) = -\hat{p} \ln \hat{p} +\frac{1-\hat{p}}{2n}$, which is exactly the Miller--Madow bias corrected entropy estimator. Taking $f(p) = p^\alpha$, we have the corresponding bias corrected estimator
\begin{equation}
f^c(\hat{p}) = \hat{p}^\alpha \left( 1 + \frac{\alpha(1-\alpha)(1-\hat{p})}{2n \hat{p}} \right).
\end{equation}

Figure~\ref{fig.approximationfigure} demonstrates the estimators for $H(P)$ and $F_\alpha(P)$ pictorially, where $\hat{p}_i = P_n(i)$ is the empirical frequency of $i$-th symbol. An important observation is that our estimator naturally satisfies the first prerequisite of any improved method for functional estimation as discussed above. Indeed, as $n\to \infty$, all the observations will fall in the ``smooth'' regime, and in the smooth regime our estimators are very similar to the MLE, which naturally implies that they are also asymptotically efficient in the sense of H\'ajek and Le Cam~\cite{Vandervaart2000}.

\begin{center}
\centering
\begin{tikzpicture}[xscale=6,yscale=5]
\draw [<->, help lines] (0.6,1.8) -- (0.6,0.6) -- (1.8,0.6);
\draw (0.600000000000000,	0.600000000000000) -- (0.620000000000000,	0.741421356237310)--(0.640000000000000,0.800000000000000)--(0.660000000000000,	0.844948974278318)--(0.680000000000000,0.882842712474619)--(0.700000000000000,0.916227766016838)-- (0.705263157894737,	0.924442842261525) --(0.757894736842105,	0.997359707119513)--(0.810526315789474,	1.05883146774112)--(0.863157894736842,1.11298917604258)--(0.915789473684211,	1.16195148694902)--(0.968421052631579,1.20697697866688)--(1.02105263157895,	1.24888568452305)--(1.07368421052632,1.28824720161169)--(1.12631578947368,	1.32547625011001)--(1.17894736842105,1.36088591025268)--(1.23157894736842,	1.39471941423903)--(1.28421052631579,1.42717019186851)--(1.33684210526316,	1.45839507527895)--(1.38947368421053,1.48852331663864)--(1.44210526315789,	1.51766293548225)--(1.49473684210526,1.54590530292692)--(1.54736842105263,	1.57332852678458)--(1.60000000000000,1.60000000000000);

\node [left] at (0.6,0.6) {0};
\node [below] at (1.6,0.6) {1};
\draw [dashed] (1.1, 0.6) -- (1.1,1.8);
\node [below] at (0.85, 1.7) {unbiased estimate};
\node [below] at (0.85, 1.6) {of best polynomial};
\node [below] at (0.85, 1.5) {approximation of};
\node [below] at (0.85, 1.4) {order $\ln n$};
\node [below] at (1.1,0.6) {$\frac{\ln n}{n}$};
\node [below] at (0.85,0.8) {``nonsmooth''};
\node [below] at (1.3,0.8) {``smooth''};
\node [above] at (1.5, 1) {$f(\hat{p}_i) - \frac{f''(\hat{p}_i)\hat{p}_i(1-\hat{p}_i)}{2n}$};
\node [right] at (1.8,0.6) {$p_i$};
\node [above] at (0.6,1.8) {$f(p_i)$};
\end{tikzpicture}

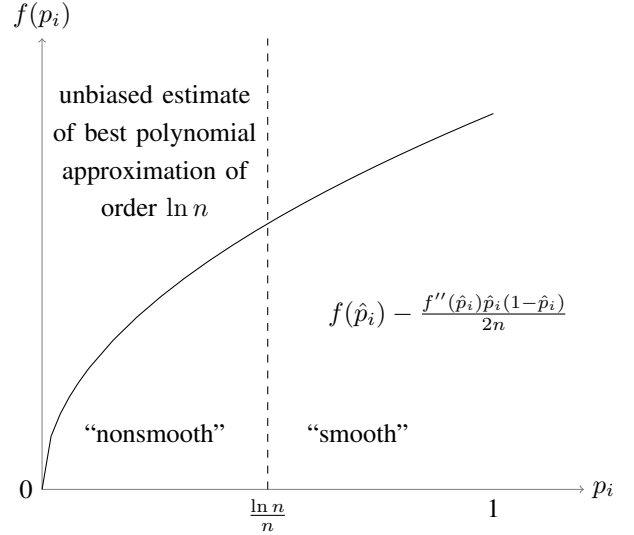
\captionof{figure}{Pictorial explanation of our estimators.}
\label{fig.approximationfigure}
\end{center}

The idea of approximation in the context of estimation has appeared before. Nemirovski~\cite{Ibragimov--Nemirovskii--Khasminskii1987some} pioneered the use of approximation theory in functional estimation in the Gaussian white noise model (see Nemirovski~\cite{Nemirovski2000topics} for a comprehensive treatment). Later, Lepski, Nemirovski, and Spokoiny~\cite{Lepski--Nemirovski--Spokoiny1999estimation} considered estimating the $L_r,r\geq 1$ norm of a regression function, and utilized trigonometric approximation. Cai and Low~\cite{Cai--Low2011} used best polynomial approximation to estimate the $\ell_1$ norm of a Gaussian mean.

We conclude this subsection by comparing any minimax rate-optimal estimator with our estimator. If we consider entropy, Theorem~\ref{thm.entropy.risk} demonstrates that when $n$ is not too large, the risk is dominated by the first term $\frac{S^2}{(n\ln n)^2}$, which corresponds to the squared bias of our estimator in the ``nonsmooth'' regime. Further, it is shown in Wu and Yang~\cite{Wu--Yang2014minimax} that in the worst case, the risk contributed by the ``nonsmooth'' regime (i.e. $[0,\frac{\ln n}{n}]$) is at least of order $\frac{S^2}{(n\ln n)^2}$. These observations together imply that the gist of any successful scheme should contribute squared bias nearly $\frac{S^2}{(n\ln n)^2}$. However, the bias always corresponds to a polynomial approximation error, and in the interval $[0,\frac{\ln n}{n}]$, it roughly corresponds to a polynomial with order $\ln n$. The squared bias $\frac{S^2}{(n\ln n)^2}$ corresponds to a polynomial whose error in approximating $f(p) = -p\ln p$ in $[0,\frac{\ln n}{n}]$ is nearly the same as the best approximation polynomial with order $\ln n$. The theory of strong uniqueness in approximation theory~\cite{Kroo--Pinkus2010strong} states that any polynomial whose approximation property is close to the best approximation must be close to the best approximation polynomial. Thus, we conclude that any successful scheme must inherently conduct near-best polynomial approximation in the ``nonsmooth'' regime, which is what we do in our scheme. Similar arguments also explain $F_\alpha(P)$ and Theorem~\ref{thm.Falpha.risk}, \ref{th_2}, and~\ref{th_1}.

\subsection{Related work}

The problem of estimating functionals of parameters is one that has been studied extensively in such fields as statistics, information theory, computer science, physics, neuroscience, psychology, and ecology, to name a few. Different communities have focused on different aspects of this general problem, and some seemingly different problems can be recast as functional estimation ones. Below we review some of the core ideas in various communities.

\subsubsection{Statistics}

Consider a sequence of independent and identically distributed (i.i.d.) random variables $Z_1,Z_2,\ldots,Z_n$ taking values in $\mathcal{Z}$, $Z_i \sim P_\theta, \theta \in \Theta \subset \mathbb{R}^p$. We would like to obtain a good estimate of the functional $\varphi(\theta)$. In general, this problem differs from that of seeking a good estimate of the parameter $\theta$. The most natural and ambitious aim towards this problem is to seek the ``optimal'' estimator given exactly $n$ samples, which falls in the realm of finite sample theory in statistics~\cite{Lehmann--Casella1998theory}. There is no consensus on what criterion best evaluates how ``good'' an estimator is in a finite sample sense. Over the years, various criteria for goodness have been proposed and analyzed, including the uniform minimum variance unbiased estimator (UMVUE), the minimum risk equivariant estimator (MRE), and the minimax estimator, among others. However, generally it is difficult to obtain estimators for functionals satisfying any of the finite sample optimality criteria mentioned above. Further, even if we obtained an estimator $\hat{\theta}_n$ that is or is close to being optimal for the parameter $\theta$ under a finite sample criterion, the plug-in approach $\varphi(\hat{\theta}_n)$ need not result in an optimal estimator for $\varphi(\theta)$ under some finite sample criterion.

In light of these shortcomings, there seems to be a perception that estimation under finite sample optimality criteria is not amenable to a general mathematical theory \cite{Ibragimov--Hasminskii1981}. Classical asymptotic theory is usually the refuge. The beautiful theory of H\'ajek and Le Cam \cite{Hajek1970characterization,Hajek1972local,LeCam1986asymptotic} showed that, under mild conditions, there exist systematic methods to construct an \emph{asymptotically efficient} estimator $\hat{\theta}_n$ for the finite dimensional parameter $\theta$, where if $\varphi(\theta)$ is differentiable at $\theta$, $\varphi(\hat{\theta}_n)$ is also asymptotically efficient for estimating $\varphi(\theta)$ \cite[Lemma 8.14]{Vandervaart2000}. Furthermore, it is also known that if the functional is non-differentiable, then it is nearly impossible to get an elegant mathematical theory~\cite{Hirano--Porter2012impossibility}.

The question of estimating functionals of finite dimensional parameters being satisfactorily answered under classical asymptotics, functional estimation in various nonparametric settings has been a strong area of focus since. There are several profound contributions in this area, of which we only mention a few. The most developed theory deals with linear functionals, for example, see \cite{Levit1974optimality,Levit1975efficiency,Koshevnik--Levit1977non,Ibragimov--Hasminskii1988estimation, Ibragimov--Hasminskii1981, Ibragimov--Khasminskii1988estimation, Donoho--Liu1991geometrizing2,Donoho--Liu1991geometrizing3, Donoho1994statistical, Goldenshluger--Pereverzev2000adaptive, Klemela--Tsybakov2001sharp,  Cai--Low2003note, Cai--Low2004minimax,Cai--Low2005adaptive,Cai--Low2005adaptivedifferentmeasure, Juditsky--Nemirovski2009nonparametric,Johannes--Schenk2012adaptive}. Another well studied situation deals with the case of ``smooth'' functionals, see \cite{Levit1978asymptotically, Ibragimov--Hasminskii1978nonparametric,Ibragimov--Nemirovskii--Khasminskii1987some, Hall--Marron1987estimation, Bickel--Ritov1988estimating, Donoho--Nussbaum1990minimax, Fan1991estimation, Efromovich--Low1996bickel,Birge--Massart1995,Goldstein--Khasminskii1996efficient,Laurent--Massart2000adaptive, Cai--Low2005nonquadratic, Cai--Low2006optimal}, among others. Estimation of non-smooth functionals is an extremely difficult problem, and is still largely open~\cite{Lepski--Nemirovski--Spokoiny1999estimation}. In particular, the problem of estimating differential entropy $\int -f \ln f$, where $f$ is a density, remains fertile ground for research, cf. \cite{Hall--Morton1993estimation, Tsybakov--vandermeulen1996root, Beirlant--Dudewicz--Gyorfi--Meulen1997, Hero2002--Ma--Michel--Gorman2002, Victor2002binless, Kraskov2004, Paninski--Yajima2008undersmoothed, Mnatsakanov--Misra--Li--Harner2008k, Wang--Kulkarni--Verdu2009now,Bouzebda--Elhattab2011uniform,Liu--Wasserman--Lafferty2012exponential}. Similar situations are also true for estimating the entropy of a discrete distribution supported on a countably infinite alphabet, cf. \cite{Antos--Kontoyiannis2001convergence,Wyner--Foster2003lower,Vu--Yu--Kass2007coverage,Zhang2012entropy,Zhang2013asymptotic}.


\subsubsection{Information theory}

In the information theory community, following the seminal work of Shannon \cite{Shannon1951prediction}, the focus has been on estimating entropy rates of general stationary ergodic processes with fixed (usually small) support (alphabet) sizes. Outside of the favored \emph{binary} alphabet, printed English contributed the other interesting example of support size $27$ (including the ``space''). Cover and King \cite{Cover--King1978convergent} gave an overview of the entropy rate estimation literature until 1978. Soon after the appearance of universal data compression algorithms proposed by Ziv and Lempel \cite{Ziv--Lempel1977universal,Ziv--Lempel1978compression}, the information theory community started applying these ideas in entropy rate estimation, e.g. Wyner and Ziv \cite{Wyner--Ziv1989}, and Kontoyiannis \emph{et al.} \cite{Kontoyiannis--Algoet--Suhov--Wyner1998nonparametric}. Verd\'u \cite{Verdu2005} provides an overview of universal estimation of information measures until 2005. Jiao \emph{et al.} \cite{Jiao--Permuter--Zhao--Kim--Weissman2013} constructed a general framework for applying data compression algorithms to establish near-optimal estimators for information rates, with a focus on directed information.

\subsubsection{Computer science, physics, neuroscience, psychology, ecology, etc}

Much of the efforts in computer science, physics, neuroscience, psychology, ecology, and related fields have focused on some special functionals of particular interest. For example, the problem of estimating Shannon entropy $H(P)$ from a finite alphabet source with i.i.d. observations has been investigated extensively. Section~\ref{subsec.motivation} and~\ref{subsec.howdidwe} summarize some of the efforts.

\subsubsection{Modern era: high dimensions and non-asymptotics}

The current era of ``big data'' abounds with applications in which we no longer operate in the asymptotic regime of large sample sizes. This sample scarcity regime necessitates going beyond classical asymptotic analysis and considering finitely many samples in high dimensions. Indeed, the recent successes of finite-blocklength analysis in information theory~\cite{Polyanskiy--Poor--Verdu2010channel}, and compressed sensing in statistics~\cite{Candes--Tao2006near, Donoho2006compressed} have demonstrated the benefit of carefully analyzing practical sample sizes. There are also ample recent examples in statistics approaching classical questions from a high dimensional perspective, cf. \cite{He--Shal2000parameters,Serdobolskii2007multiparametric,Boucheron--Gassiat2009bernstein,Spokoiny2012parametric,Spokoiny2013bernstein}. The machine learning community has the tradition of favoring non-asymptotic analysis, and usually pose the question of the sample complexity for achieving $\epsilon$ accuracy with $1-\delta$ probability, cf. \cite{Valiant1984theory,Vapnik2006estimation,Vapnik1998statistical}. The information theoretic counterpart of high dimensional statistics might be the large alphabet setting, with exciting recent advances (cf. \cite{Orlitsky--Santhanam--Zhang2004universal, Orlitsky--Santhanam2004speaking,Wagner--Viswanath--Kulkarni2011probability,Ohannessian--Tan--Dahleh2011canonical,Szpankowski--Weinberger2012minimax,Yang--Barron2013large}).

With the above as context, our work revisits the framework of functional estimation for finite dimensional models, with a focus on high dimensional and non-asymptotic analysis.

\subsection{General methodology for functional estimation}

\subsubsection{Review: general methods of estimation}\label{subsubsec.reviewmle}

We begin by reviewing the existing general approaches to estimation. Maximum likelihood is the most widely used statistical estimation technique, which emerged in modern form 90 years ago in a series of remarkable papers by Fisher \cite{Fisher1922mathematical, Fisher1925theory, Fisher1934two}. As evidence of its ubiquity, the Google Scholar
search query ``Maximum Likelihood Estimation'' yields approximately $2,570,000$ articles, patents and books. Indeed, in his response to Berkson \cite{Berkson1980minimum} in 1980, Efron explains the popularity of maximum likelihood:

\begin{center}
\parbox{0.45\textwidth}{~~\emph{  ``The appeal of maximum likelihood stems from its universal applicability, good mathematical properties, by which I refer to the standard asymptotic and exponential family results, and generally good track record as a tool in applied statistics, a record accumulated over fifty years of heavy usage.
"}}
\end{center}

Over the years, the following folk theorem seems to have been tacitly accepted by applied scientists:

\begin{theorem}[``Folk Theorem'']\label{thm.folk}
For a finite dimensional parametric estimation problem, it is ``good'' to employ the MLE.
\end{theorem}

From the perspective of mathematical statistics, however, maximum likelihood is by no means sacrosanct. As early as in 1930, in his letters to Fisher, Hotelling raised the possibility of the MLE performing poorly \cite{Stigler2007epic}. Subsequently, various examples showing that the performance of the MLE can be significantly improved upon have been proposed in the literature, cf. Le Cam~\cite{LeCam1990maximum} for an excellent overview. However, as Stigler~\cite[Sec. 12]{Stigler2007epic} discussed in his 2007 survey, while these early examples created a flurry of excitement, for the most part they were not seen as debilitating to the fundamental theory. Perhaps because these examples did not provide a systematic methodology for improving the MLE.

In 1956, Stein~\cite{Stein1956inadmissibility} observed that in the Gaussian location model $X \sim \mathcal{N}(\theta, I_p)$ (where $I_p$ is the $p\times p$ identity matrix), the MLE for $\theta$, $\hat{\theta}^{\textrm{MLE}} = X$ is inadmissible \cite[Chap. 1]{Lehmann--Casella1998theory} when $p\geq 3$. Later, James and Stein~\cite{James--Stein1961estimation} showed that an estimator that appropriately shrinks the MLE towards zero achieves uniformly lower $L_2$ risk compared to the risk of the MLE. The \emph{shrinkage} idea underlying the James--Stein estimator has proven extremely fruitful for statistical methodology, and has motivated further milestone developments in statistics, such as wavelet shrinkage~\cite{Donoho--Johnstone1994ideal}, and compressed sensing~\cite{Candes--Tao2006near,Donoho2006compressed}.

One interpretation of the shrinkage idea is that, when one desires to estimate a high dimensional parameter, the MLE may have a relatively small bias compared to the variance. Shrinking the MLE introduces an additional bias, but reduces the overall risk by reducing the variance substantially. A natural question now arises: what about situations wherein the bias is the dominating term? Does there exist an analogous methodology for improving over the performance of the MLE in such scenarios? A precedent to this line of questioning can be found in the 1981 Wald Memorial Lecture by Efron~\cite{Efron1982maximum} entitled ``Maximum Likelihood and Decision Theory'':

\begin{center}
\parbox{0.45\textwidth}{~~\emph{  ``
\ldots the MLE can be non-optimal if the statistician has one specific estimation problem in mind. Arbitrarily bad counterexamples, along the line of estimating $e^{\theta}$ from $X \sim \mathcal{N}(\theta,1)$, are easy to construct. Nevertheless the MLE has a good reputation, acquired over 60 years of heavy use, for producing reasonable point estimates. Useful general improvements on the MLE, such as robust estimation, and Stein estimation, are all the more impressive for their rarity.
"}}
\end{center}

For the aforementioned example, Efron~\cite{Efron1982maximum} argued that the reason the MLE $e^X$ may not be a good estimate for $e^\theta$, is that it has a large \emph{bias}. In particular, the statistician may prefer the uniform minimum variance unbiased estimator (UMVUE), $e^{X-\frac{1}{2}}$ to estimate $e^\theta$. As we discussed in the presentation of our main results, the bias is usually the dominating term in estimation of functionals of high-dimensional parameters. Notably, the two general improvements of the MLE, namely robust estimation and shrinkage estimation, are not designed to handle functional estimation problems such as the one presented by Efron. Also, as Efron himself observed, the statistician cannot always rely on the UMVUE to save the day, since these are generally very hard to compute, and may not always exist \cite[Remark C, Sec. 7]{Efron1982maximum}. Thus, there is a need to address, both in scope and methodology, the improvement over the MLE for problems where the  bias is the leading term. Such a solution could be considered the dual of the idea of shrinkage, since the trade-off between bias and variance is now reversed, i.e., one might want to sacrifice the variance to reduce the bias.

\subsubsection{Approximation: dual of shrinkage}

Our main results in this paper imply that Theorem~\ref{thm.folk} is far from true in high-dimensional non-asymptotic settings. Now, we aim to abstract our scheme in estimating functionals of type~(\ref{eqn.generalf}), and distill a general methodology for estimating functionals of parameters of any finite dimensional parametric families.

Consider estimating $G(\theta)$ of a parameter $\theta \in \Theta \subset \mathbb{R}^p$ for an experiment $\{P_\theta: \theta \in \Theta\}$, with a consistent estimator $\hat{\theta}_n$ for $\theta$, where $n$ is the number of observations. Suppose the functional $G(\theta)$ is analytic\footnote{A function $f$ is analytic at a point $x_0$ if and only if its Taylor series about $x_0$ converges to $f$ in some neighborhood of $x_0$.} everywhere except at $\theta \in \Theta_0$. A natural estimator for $G(\theta)$ is $G(\hat{\theta}_n)$, and we know from classical asymptotics \cite[Lemma 8.14]{Vandervaart2000} that if the model satisfies the benign LAN (Local Asymptotic Normality) condition~\cite{Vandervaart2000} and $\hat{\theta}_n$ is asymptotically efficient for $\theta$, then $G(\hat{\theta}_n)$ is also asymptotically efficient for $G(\theta)$ for $\theta \notin \Theta_0$. Note that this general framework naturally encompasses the family of probability functionals as a special case. To see this, let $\Theta$ be the $S$-dimensional probability simplex, where $S$ denotes the support size. For functionals of the form (\ref{eqn.generalf}), if $f$ is analytic on $(0,1]$, it is clear that $\Theta_0$ denotes the boundary of the probability simplex. One natural candidate for $\hat{\theta}_n$ is the empirical distribution, which is an unbiased estimator for any $\theta \in \Theta$.

We propose to conduct the following two-step procedure in estimating $G(\theta)$.

\begin{enumerate}

\item \textbf{Classify Regime}: Compute $\hat{\theta}_n$, and declare that we are operating in the ``nonsmooth'' regime if $\hat{\theta}_n$ is ``close'' enough to $\Theta_0$. Note that $G(\theta)$ is not analytic at any $\theta \in \Theta_0$. Otherwise declare we are in the ``smooth'' regime;
\item {\bf Estimate}:
\begin{itemize}
\item If $\hat{\theta}_n$ falls in the ``smooth'' regime, use an estimator ``similar'' to $G(\hat{\theta}_n)$ to estimate $G(\theta)$;
\item If $\hat{\theta}_n$ falls in the ``nonsmooth'' regime, replace the functional $G(\theta)$ in the ``nonsmooth'' regime by an approximation $G_{\text{appr}}(\theta)$ (another functional) which can be estimated without bias, then apply an unbiased estimator for the functional $G_{\text{appr}}(\theta)$.
\end{itemize}
\end{enumerate}

\subsubsection{Details of ``Approximation''}

While this general recipe appears clean in its description, there are several problem-dependent features that one needs to design carefully -- namely
\begin{question}\label{question.general1}
How to determine ``nonsmooth'' regime? What is the size of it?
\end{question}

\begin{question}\label{question.general2}
What approximation should we choose to approximate $G(\theta)$ in the ``nonsmooth'' regime?
\end{question}

\begin{question}\label{question.general3}
What does `` `similar' to $G(\hat{\theta}_n)$'' mean precisely? What exactly do we do in the ``smooth'' regime?
\end{question}

The careful reader may have realized that Questions~\ref{question.general1},\ref{question.general2}, and~\ref{question.general3} resemble the questions we asked in Section~\ref{subsec.howdidwe}. Answers to these questions draw on additional problems we investigated~\cite{Jiao--Han--Weissman2015divergence} beyond these in the present paper. 

\begin{enumerate}
\item Question~\ref{question.general1}

We should always choose the ``nonsmooth'' regime to be the smallest regime such that we can still \emph{localize} the parameter $\theta$. In other words, when we observe $\hat{\theta}_n$ in the ``nonsmooth'' regime, we should be able to infer with high probability that $\theta$ is also in the ``nonsmooth'' regime. Similarly, we should also be able to localize the parameter in the ``smooth'' regime. A concrete case would be the following. Say we observe $X \sim \mathsf{B}(n,p)$, and we would like to estimate a functional which is not analytic at $p_0 = 0.2$. How should we define the ``nonsmooth'' regime? Noting that $\mathsf{Var}(X/n) = \frac{p(1-p)}{n}$, it turns out we can set the ``nonsmooth'' regime to be $\left[ p_0 -\sqrt{\frac{p_0(1-p_0)\ln n}{n}},p_0 +\sqrt{\frac{p_0(1-p_0)\ln n}{n}} \right ]$ (cf. Lemma~\ref{lemma.poissontail}).

\item Question~\ref{question.general2}

We should always choose an approximation $G_{\text{appr}}(\theta)$ that can be estimated without bias. This requirement leads us to the general theory of unbiased estimation, which was pioneered by Halmos~\cite{Halmos1946theory} and Kolmogorov~\cite{Kolmogorov1950unbiased}. For a comprehensive survey the readers are referred to the monograph by Voinov and Nikulin~\cite{Voinov--Nikulin1996unbiased}.

There is a delicate trade-off: the approximation $G_{\text{appr}}(\theta)$ should be estimated without bias, but also should approximate the functional $G(\theta)$ well, and at the same time not incur too much additional variance. These three requirements yield a highly non-trivial interplay between approximation theory and statistics, of which our understanding is as yet incomplete.

For functionals in (\ref{eqn.generalf}), the separability of each $p_i$ essentially reduces the problem from multivariate to univariate, for which best polynomial approximation plays an important role in the optimal solution. Similar stories are true for the Gaussian setting, e.g. estimating $\sum_{i = 1}^S f(\mu_i)$ where $\mu \in \mathbb{R}^S$ is the mean of a normal vector. Modern approximation theory provides mature machinery of polynomial approximation in one dimension, with various profound results developed over the last century. The best approximation error rate $E_n[f]_A$:
\begin{equation}
E_n[f]_A = \min_{P \in \mathsf{poly}_n} \max_{x\in A}|f(x) - P(x)|,
\end{equation}
where $\mathsf{poly}_n$ is the collection of polynomials with order at most $n$ on $A$, is a crucial object in approximation theory as well as our general methodology. Quantifying $E_n[f]_A$ and obtaining the polynomial that achieves it turned out to be extremely challenging. Remez \cite{Remez1934determination} in 1934 proposed an efficient algorithm for computing the best polynomial approximation, and it was recently implemented and highly optimized in Matlab by the Chebfun team \cite{Trefethen--chebfunv5,Pachon--Trefethen2009barycentric}. Regarding the theoretical understanding of $E_n[f]_A$, de la Vall\'ee-Poussin, Bernstein, Ibragimov, Markov, Kolmogorov and others have made significant contributions, and it is still an active research area. Among others, Bernstein \cite{Bernstein1937,Bernstein1938} and Ibragimov \cite{Ibragimov1946} showed various exact limiting results for some important classes of functions like $|x|^p$ and $|x|^m \ln |x|^n$. For example, we have
\begin{theorem}\label{thm.bernsteinxp}\cite{Bernstein1938}
The following limit exists for all $p>0$:
\begin{equation}
\lim_{n\to \infty} n^p E_n [|x|^p]_{[-1,1]} = \mu(p),
\end{equation}
where $\mu(p)$ is a constant bounded as
\begin{equation}
\frac{\Gamma(p)}{\pi} \Bigg | \sin \frac{\pi p}{2} \Bigg | \left( 1- \frac{1}{p-1} \right) \le \mu(p) \le \frac{\Gamma(p)}{\pi} \Bigg | \sin \frac{\pi p}{2} \Bigg |,
\end{equation}
where $\Gamma(\cdot)$ denotes the Gamma function.
\end{theorem}

Regarding bounds on $E_n[f]$ for any finite $n$, Korneichuk \cite[Chap. 6]{Korneichuk1991} provides a comprehensive study. For a comprehensive treatment of modern approximation theory, DeVore and Lorentz \cite{Devore--Lorentz1993}, Ditzian and Totik \cite{Ditzian--Totik1987} provide excellent references. For the most up-to-date review of polynomial approximation, we refer the readers to Bustamante~\cite{Bustamante2011algebraic}.

We emphasize that the discussions above refer to approximation in dimension one. The general multivariate case is extremely complicated. Rice~\cite{Rice1963tchebycheff} wrote:
\begin{center}
\vspace{5pt}
\parbox{0.45\textwidth}{~~\emph{  ``The theory of Chebyshev approximation (a.k.a. best approximation) for functions
of one real variable has been understood for some time and is quite elegant. For about fifty years attempts have been made to generalize this theory to functions of several variables. These attempts have failed because of the lack
of uniqueness of best approximations to functions of more than one variable.
"}}
\vspace{5pt}
\end{center}

Another related paper~\cite{Jiao--Han--Weissman2015divergence} showed that the non-uniqueness can cause serious trouble: some polynomial that can achieve the best approximation error rate cannot be used in our general methodology in functional estimation. What if we relax the requirement of computing the best approximation in multivariate case, and only want to analyze the best approximation rate (i.e., the best approximation error up to a multiplicative constant)? That turns out also to be extremely difficult. Ditzian and Totik~\cite[Chap. 12]{Ditzian--Totik1987} obtained the error rate estimate on simple polytopes\footnote{A simple polytope in $\mathbb{R}^d$ is a polytope such that each vertex has $d$ edges. }, balls, and spheres, and it remained open until Totik~\cite{Totik2013polynomial} generalized the results to general polytopes.  For results in balls and spheres, the readers are referred to Dai and Xu~\cite{Dai--Xu2013approximation}. We still know little about regimes other than polytopes, balls, and spheres.

We do not know whether in general polynomial approximation can achieve the minimax rates in general settings. Probably other approximation bases need be chosen for certain problems.

\item Question~\ref{question.general3}

Note that we have assumed $G(\theta)$ is analytic in the ``smooth'' regime. For various statistical models (like Gaussian and Poisson), any analytic functional admits unbiased estimators. We propose to use Taylor series bias correction~\cite{Withers1987} in the ``smooth'' regime, where the order of the Taylor series may vary between problems.
\end{enumerate}

\subsection{Remaining content}

The rest of the paper is organized as follows. Section~\ref{sec.achievability} details the construction of our estimators $\hat{H}$ and $\hat{F}_\alpha$ and their analysis. In Section~\ref{sec.lowerbound} we present our general approach for proving minimax lower bounds and apply it to establish Theorems~\ref{thm.Falpha.risk} and~\ref{th_1}. Section~\ref{sec.experiments} presents a few experiments demonstrating the practical advantages of our estimators in entropy estimation, mutual information estimation, entropy rate estimation, and learning graphical models. Complete proofs of the remaining theorems and lemmas are provided in the appendices.

\section{Estimator construction and analysis}\label{sec.achievability}

Throughout our analysis, we utilize the Poisson sampling model, which is equivalent to having a $S$-dimensional random vector $\mathbf{Z}$ such that each component $Z_i$ in $\mathbf{Z}$ has distribution $\mathsf{Poi}(np_i)$, and all coordinates of $\mathbf{Z}$ are independent. For simplicity of analysis, we conduct the classical ``splitting'' operation \cite{Tsybakov2013aggregation} on the Poisson random vector $\mathbf{Z}$, and obtain two independent identically distributed random vectors $\mathbf{X} = [X_1,X_2,\ldots,X_S]^T, \mathbf{Y} = [Y_1,Y_2,\ldots,Y_S]^T$, such that each component $X_i$ in $\mathbf{X}$ has distribution $\spo(np_i/2)$, and all coordinates in $\mathbf{X}$ are independent. For each coordinate $i$, the splitting process generates a random variable $T_i$ such that $T_i|\mathbf{Z} \sim \mathsf{B}(Z_i, 1/2)$, and assign $X_i = T_i, Y_i = Z_i - T_i$. All the random variables $\{T_i:1\leq i\leq S\}$ are conditionally independent given our observation $\mathbf{Z}$.

For simplicity, we re-define $n/2$ as $n$, and denote
\begin{equation}
\hat{p}_{i,1} = \frac{X_i}{n}, \hat{p}_{i,2} = \frac{Y_i}{n}, \Delta = \frac{c_1 \ln n}{n}, K = c_2 \ln n, t = \frac{\Delta}{4},
\end{equation}
where $c_1,c_2$ are positive constants to be specified later. For simplicity we assume $K$ is always a non-negative integer. Note that $\Delta,K,t$ are functions of $n$, where we omit the subscript $n$ for brevity. We remark that the ``splitting'' operation is used to simplify the analysis, and is not performed in the experiments. We also note that for random variable $X$ such that $nX \sim \spo(np)$,
\begin{equation}
\bE \prod_{r = 0}^{k-1} \left( X - \frac{r}{n} \right) = p^k,
\end{equation}
for any $k \in \mathbb{N}_+$. For a proof of this fact we refer the readers to Withers~\cite[Example 2.8]{Withers1987}.
\subsection{Estimator construction}

Our estimator $\hat{F}_\alpha, \alpha>0$, is constructed as follows.
\begin{equation}
\hat{F}_\alpha \triangleq \sum_{i = 1}^S \left[  L_\alpha(\hat{p}_{i,1}) \mathbbm{1}(\hat{p}_{i,2} \leq 2\Delta) + U_\alpha(\hat{p}_{i,1}) \mathbbm{1}(\hat{p}_{i,2}>2\Delta) \right],  \label{eqn.falphaconstruct}
\end{equation}
where
\begin{align}
S_{K,\alpha}(x) & \triangleq  \sum_{k = 1}^{K} g_{k,\alpha} (4\Delta)^{-k + \alpha} \prod_{r = 0}^{k-1} (x-r/n) \label{eqn.SKalpha} \\
L_\alpha(x) & \triangleq \min \left\{ S_{K,\alpha}(x) , 1 \right\} \label{eqn.Lalphadef} \\
U_\alpha(x) & \triangleq I_n(x)\left( 1 + \frac{\alpha(1-\alpha)}{2n x} \right) x^\alpha. \label{eqn.Ualphadef}
\end{align}

We explain each equation in detail as follows.

\begin{enumerate}
\item Equation~(\ref{eqn.falphaconstruct}):

Note that $\hat{p}_{i,1}$ and $\hat{p}_{i,2}$ are i.i.d. random variables such that $n \hat{p}_{i,1} \sim \mathsf{Poi}(n p_i)$. We use $\hat{p}_{i,2}$ to determine whether we are operating in the ``nonsmooth'' regime or not. If $\hat{p}_{i,2} \leq 2\Delta$, we declare we are in the ``nonsmooth'' regime, and plug in $\hat{p}_{i,1}$ into function $L_\alpha(\cdot)$. If $\hat{p}_{i,2}>2\Delta$, we declare we are in the ``smooth'' regime, and plug in $\hat{p}_{i,1}$ into $U_\alpha(\cdot)$.

\item Equation~(\ref{eqn.SKalpha}):

The coefficients $g_{k,\alpha}, 0\leq k \leq K$ are coefficients of the best polynomial approximation of $x^\alpha$ over $[0,1]$ up to degree $K$, i.e.,
\begin{equation}
\sum_{k = 0}^{K} g_{k,\alpha} x^k = \argmin_{y(x)\in \mathsf{poly}_{K}} \sup_{x\in [0,1]} |y(x) - x^\alpha|,
\end{equation}
where $\mathsf{poly}_{K}$ denotes the set of algebraic polynomials up to order $K$. Note that in general $g_{k,\alpha}$ depends on $K$, which we do not make explicit for brevity. Lemma~\ref{lemma.lowpartbiasvariance} shows that for $nX \sim \spo(np)$,
\begin{equation}
\bE S_{K,\alpha}(X) = \sum_{k =1 }^{K} g_{k,\alpha} (4\Delta)^{-k + \alpha} p^k.
\end{equation}
Thus, we can understand $S_{K,\alpha}(X), nX\sim \spo(np)$ as a random variable whose expectation is nearly \footnote{Note that we have removed the constant term from the best polynomial approximation. It is to ensure that we assign zero to symbols we do not see. } the best approximation of function $x^\alpha$ over $[0,4\Delta]$.

\item Equation~(\ref{eqn.Lalphadef}):

Any reasonable estimator for $p_i^\alpha$ should be upper bounded by the value one. We cut off $S_{K,\alpha}(x)$ by upper bound $1$, and define the function $L_\alpha(x)$, which means ``lower part''.

\item Equation~(\ref{eqn.Ualphadef}):

The function $U_\alpha(x)$ (standing for ``upper part'') is nothing but a product of an interpolation function $I_n(x)$\footnote{The usage of the interpolation function was partially inspired by Valiant and Valiant~\cite{Valiant--Valiant2011power}. } and the bias-corrected MLE. The careful reader may note that the bias-corrected MLE is not exactly the same as what we used before in (\ref{eqn.biascorrectiongeneralf}). It is because here we are using the Poisson model instead of the Multinomial model. In the Poisson model, the bias correction formula should be modified to
\begin{equation}\label{eqn.biascorrectiongeneralfpoisson}
f^c(\hat{p}) = f(\hat{p}) - \frac{f''(\hat{p}) \hat{p}}{2n}.
\end{equation}

The interpolation function $I_n(x)$ is designed to make $U_\alpha(x)$ a smooth function on $[0,1]$. Indeed, when $0<\alpha<1$, were it not for the interpolation function, $U_\alpha(x)$ would be unbounded for $x$ close to zero. Note that $L_\alpha(x)$ and $U_\alpha(x)$ are dependent on $n$. We omit this dependence in notation for brevity. The interpolation function $I_n(x)$ is defined as follows:
\begin{equation}
I_n(x) = \begin{cases} 0 & x \leq t \\
g\left( x-t; t \right) & t < x < 2t \\ 1 & x \geq  2t\end{cases}
\end{equation}

The following lemma characterizes the properties of the function $g(x;a)$ appearing in the definition of $I_n(x)$. In particular, it shows that $I_n(x)$ is four times continuously differentiable. 
\begin{lemma}\label{lemma.gxa}
For the function $g(x;a)$ on $[0,a]$ defined as follows,
\begin{align}
& g(x;a) \nonumber \\
& \quad \triangleq 126 \left( \frac{x}{a} \right)^5 - 420 \left(\frac{x}{a}\right)^6 \nonumber \\
& \quad \quad + 540 \left( \frac{x}{a} \right)^7-315 \left( \frac{x}{a} \right)^8 + 70 \left( \frac{x}{a} \right)^9 ,
\end{align}
we have the following properties:
\begin{align}
g(0;a) = 0,\quad & g^{(i)}(0;a) = 0, 1\leq i\leq 4 \\
g(a;a) = 1, \quad & g^{(i)}(a;a) = 0, 1\leq i\leq 4
\end{align}
\end{lemma}

The function $g(x;1)$ is depicted in Figure~\ref{fig.gx1}.
\begin{center}
\vspace*{15pt}
  \centering
  \centerline{\includegraphics[width=0.4\textwidth]{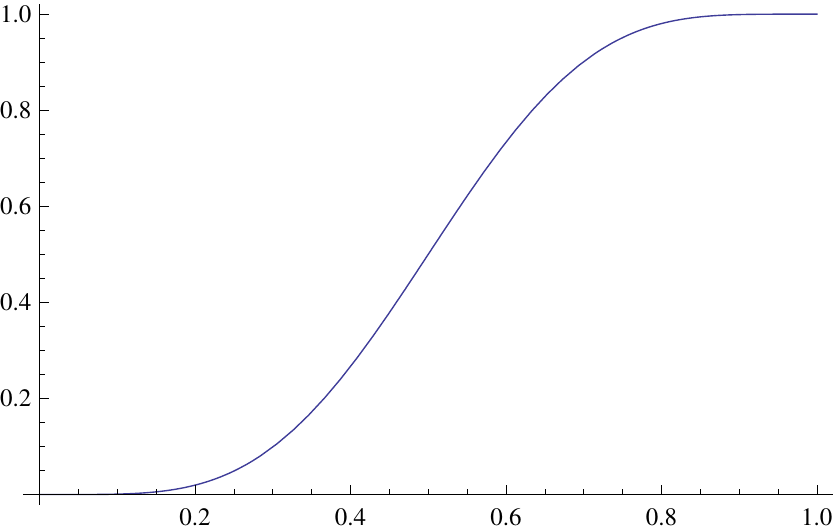}}
  \vspace{-0.1cm}
\captionof{figure}{The function $g(x;1)$ over interval $[0,1]$. }
\label{fig.gx1}
\end{center}

\end{enumerate}

Similarly, we define our estimator for entropy $H(P)$ as
\begin{equation}
\hat{H} \triangleq \sum_{i = 1}^S \left[  L_H(\hat{p}_{i,1}) \mathbbm{1}(\hat{p}_{i,2} \leq 2\Delta) + U_H(\hat{p}_{i,1}) \mathbbm{1}(\hat{p}_{i,2}>2\Delta) \right],
\end{equation}
where
\begin{align}
S_{K,H}(x) & \triangleq  \sum_{k = 1}^{K} g_{k,H} (4\Delta)^{-k + 1} \prod_{r = 0}^{k-1} (x-r/n) \\
L_H(x) & \triangleq \min \left\{ S_{K,H}(x) , 1 \right\} \\
U_H(x) & \triangleq I_n(x)\left( -x \ln x + \frac{1}{2n}\right).
\end{align}

The coefficients $\{g_{k,H}\}_{1\leq k\leq K}$ are defined as follows. We first define
\begin{equation}
\sum_{k = 0}^K r_{k,H} x^k = \argmin_{y(x) \in \mathsf{poly}_{K}} \sup_{x\in [0,1]} |y(x)- (-x \ln x)|
\end{equation}
and then define
\begin{equation}\label{eqn.gkhdefine}
g_{k,H} = r_{k,H}, 2\leq k \leq K, g_{1,H} = r_{1,H} - \ln (4\Delta).
\end{equation}

Lemma~\ref{lemma.approsmallentropy} shows that for $nX \sim \spo(np)$,
\begin{equation}
\bE S_{K,H}(X) = \sum_{k = 1}^{K} g_{k,H} (4\Delta)^{-k + 1} p^k
\end{equation}
is a near-best polynomial approximation for $-p\ln p$ on $[0,4\Delta]$.

Figure~\ref{fig.comparepluginandbest} is designed to provide pictorial explanation of our scheme in the ``nonsmooth'' regime for entropy estimation. The curve $-p\ln p$ is the functional we want to estimate, and the horizontal axis represents the possible values of $p$. We take $n = 100$, and use a $3$-order best polynomial approximation for function $-p\ln p$ over regime $[0, \frac{\ln n}{n} = 0.0461]$. It is evident from the curve that the expectation $\bE [L_H(\hat{p})]$ is very close to the true function $-p\ln p$, but the expectation of the MLE $\bE[-\hat{p}\ln \hat{p}]$, and the function $L_H(\cdot)$ itself are both far from $-p\ln p$. It demonstrates the nuanced nature of our scheme: we first construct a polynomial (equal to $\bE [L_H(\hat{p})]$) that approximates the function $-p\ln p$ very well, then we design the function $L_H(\cdot)$ to make sure its expectation is the good approximation. Plugging $\hat{p}$ in $L_H(\cdot)$ may seem less sensible than plugging $\hat{p}$ into $-p\ln p$ at first glance, but Figure~\ref{fig.comparepluginandbest} vividly demonstrates that in fact plugging $\hat{p}$ into $L_H(\cdot)$ is far more accurate in estimating $-p\ln p$.

\begin{center}
  \centering
  \centerline{\includegraphics[width=0.45\textwidth]{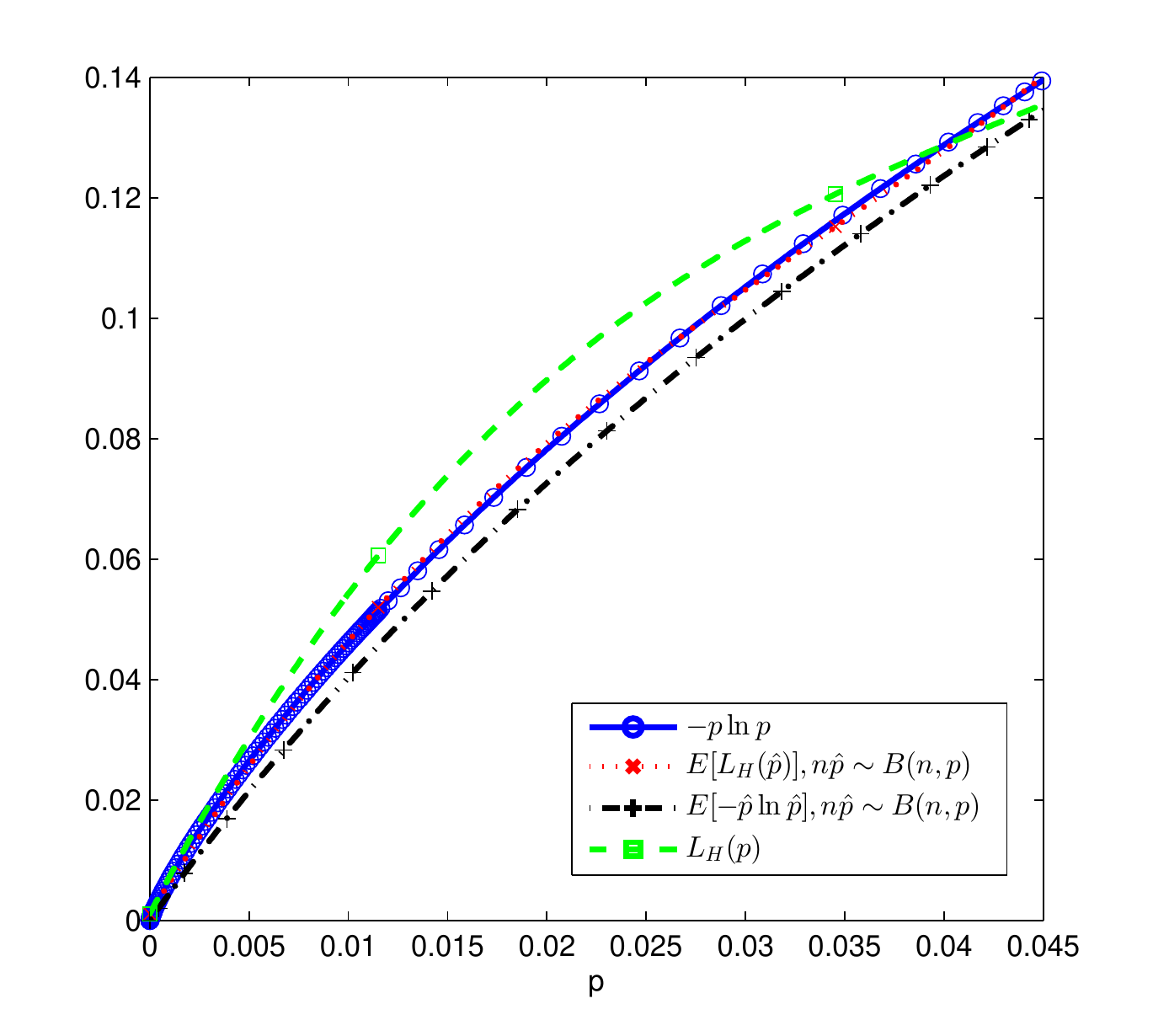}}
  \vspace{-0.1cm}
\captionof{figure}{Comparison of our scheme and MLE}
\label{fig.comparepluginandbest}
\end{center}

\subsection{Estimator analysis}

We demonstrate our analysis techniques via the proof of Theorem~\ref{thm.Falpha.risk} and~\ref{th_2}, and note that similar techniques allow us to establish Theorem~\ref{thm.entropy.risk}.

The next two lemmas show that the estimators $U_\alpha(x), U_H(x)$ have desirable bias and variance properties when the true probability $p$ is not too small.

\begin{lemma}\label{lemma.largebias}
Suppose $nX \sim \spo(np), p \geq  \Delta, c_1 \ln n\geq 1$. For $0<\alpha<3/2$, we have
\begin{align}
\left| \bE U_\alpha(X) - p^\alpha \right| & \leq \frac{17}{n^\alpha(c_1 \ln n)^{2-\alpha}} + \frac{8310 (1+\alpha)}{\alpha(2-\alpha)}p^\alpha n^{-c_1/8}.
\end{align}
For $0<\alpha\leq 1/2$,
\begin{align}
\mathsf{Var}(U_\alpha(X))& \leq   \frac{24}{n^{2\alpha}(c_1\ln n)^{1-2\alpha}}  + \frac{576}{\alpha}p^{2\alpha} n^{-c_1/8} \nonumber \\
& \quad  + \frac{28800}{\alpha^2} p^{2\alpha}n^{-c_1/4}. 
\end{align}
For $1/2<\alpha<1$,
\begin{align}
\mathsf{Var}(U_\alpha(X))& \leq \frac{14p^{2\alpha-1}}{n} + \frac{576}{\alpha}p^{2\alpha}n^{-c_1/8} \nonumber \\
& \quad + \frac{28800}{\alpha^2}p^{2\alpha}n^{-c_1/4} + \frac{8}{n^{2\alpha}(c_1\ln n)^{2-2\alpha}}.
\end{align}
For $1<\alpha<3/2$, 
\begin{align}
\mathsf{Var}(U_\alpha(X))& \leq \frac{202p}{n} + \frac{8}{n^2} + \frac{28800}{\alpha^2}p^{2\alpha} n^{-c_1/4} \nonumber \\
& \quad + \frac{120}{\alpha}p^{2\alpha} n^{-c_1/8}. 
\end{align}
\end{lemma}

\begin{lemma}\label{lemma.entropyupper}
If $nX \sim \spo(np), p\geq \Delta$,
\begin{align}
\left| \bE U_H(X) +p \ln p \right| & \leq \frac{3}{c_1 n \ln n} + \frac{2}{3 (c_1 \ln n)^2 n} \nonumber \\
& \quad + 8024\left( p \ln(1/p) + 2p\right)n^{-c_1/8} 
\end{align}
\begin{align}
\mathsf{Var}(U_H(X)) &  \leq  2p(\ln p - \ln 2)^2/n \nonumber \\
& \quad + 54 p^2 \left| 2 (\ln p)^2 -2 \ln p + 3\right|n^{-c_1/8} \nonumber \\
& \quad + \left( \frac{1}{n} + 60 \left( p\ln(1/p) + 2p\right)n^{-c_1/8}\right)^2 \nonumber \\
& \quad + 2 \left( p\ln(1/p) + \frac{1}{2n}\right) \times 
\nonumber \\
& \qquad \Big( \frac{1}{n} + 60 \left( p\ln(1/p) + 2p\right)n^{-c_1/8}\Big).
\end{align}
\end{lemma}

The following two lemmas characterize the performance of $S_{K,\alpha}(X)$ and $S_{K,H}(X), nX\sim \spo(np)$ when $p$ is not too large.

\begin{lemma}\label{lemma.lowpartbiasvariance}
If $nX \sim \spo(np), p \leq 4 \Delta, \alpha>0$, we have
\begin{equation}
| \bE S_{K,\alpha}(X) - p^\alpha |  \leq \frac{c_3}{(n \ln n)^\alpha},
\end{equation}
and for $n$ large enough, we can take $c_3 = \frac{2 \mu(2\alpha) c_1^\alpha}{c_2^{2\alpha}}$, where $c_3$ is the constant appearing in Lemma~\ref{lemma.approsmall}. If we also have $c_2 \leq 4c_1$, then
\begin{equation}
\bE S_{K,\alpha}^2(X)  \leq n^{8c_2 \ln 2} \frac{(4c_1 \ln n)^{2+2\alpha}}{n^{2\alpha}}.
\end{equation}

For the entropy, if $p \leq 4 \Delta$, we have
\begin{equation}
| \bE S_{K,H}(X) +p \ln p |  \leq \frac{C}{n \ln n}.
\end{equation}
When $n$ is large enough, $C$ can be taken to be $ \frac{4c_1 \nu_1(2)}{c_2^2}$, which is given in Lemma~\ref{lemma.approsmallentropy}. If we also have $c_2 \leq 4c_1$, then
\begin{equation}
\bE S_{K,H}^2(X)  \leq n^{8c_2 \ln 2} \frac{(4c_1 \ln n)^{4}}{n^{2}}.
\end{equation}
\end{lemma}

\begin{lemma}\label{lem_L_small}
  If $nX\sim \mathsf{Poi}(np), p\le \frac{1}{n\ln n}, 1<\alpha<3/2$, then for $c_2\le 4c_1$,
\begin{align}
  |\mathbb{E}S_{K,\alpha}(X)-p^\alpha| &\le D_1\left(\frac{4c_1}{c_2^2n\ln n}\right)^{\alpha-1}p\\
  \mathbb{E}S_{K,\alpha}^2(X)&\le n^{10c_2\ln 2}\frac{(4c_1\ln n)^{2\alpha+2}p}{n^{2\alpha-1}},
\end{align}
where $D_1$ is a universal positive constant appearing in Lemma \ref{lemma.nonasympxa}.
\end{lemma}


With the machinery established in Lemma~\ref{lemma.largebias},~\ref{lemma.entropyupper},~\ref{lemma.lowpartbiasvariance}, and~\ref{lem_L_small}, we are now ready to bound the bias and variance of each summand in our estimators. Define,
\begin{equation}
\xi = \xi(X,Y) = L_\alpha(X) \mathbbm{1}(Y \leq 2\Delta) + U_\alpha(X) \mathbbm{1}(Y>2\Delta),
\end{equation}
where $n X \stackrel{D}{=} n Y \sim \spo(np)$, and $X$ is independent of $Y$. Apparently, we have
\begin{equation}
\hat{F}_\alpha = \sum_{i = 1}^S \xi(\hat{p}_{i,1},\hat{p}_{i,2}),
\end{equation}
and each of the $S$ summands are independent. Hence, it suffices to analyze the bias and variance of $\xi(X,Y)$ thoroughly for all values of $p$ in order to obtain a risk bound for $\hat{F}_\alpha$. We break this into three different regimes. In the first case when $p \leq \Delta$, we shall show that the estimator essentially behaves like $L_\alpha(X)$, which is a good estimator when $p$ is small. In the second case when $\Delta \leq p \leq 4\Delta$, we show that our estimator uses either $L_\alpha(X)$ or $U_\alpha(X)$, which are both good estimators in this case. In the last case $p\geq 4\Delta$, we show that our estimator behaves essentially like $U_\alpha(X)$, which has good properties when $p$ is not too small.

We denote $B(\xi) \triangleq \bE \xi(X,Y) - p^\alpha$ as the bias of $\xi$.
\begin{lemma}\label{lemma.preparemainf}
Suppose $0<\alpha<1$, $0<c_1 = 16(\alpha + \delta), 0<8c_2 \ln 2 = \epsilon < \alpha, \delta>0$. Then,
\begin{enumerate}
\item when $p\leq \Delta $,
\begin{align}
|B(\xi)| & \lesssim \frac{1}{(n \ln n)^\alpha} ,\\
\mathsf{Var}(\xi) & \lesssim \frac{(\ln n)^{2+2\alpha}}{n^{2\alpha - \epsilon}}.
\end{align}
\item when $\Delta < p \leq 4\Delta$,
\begin{align}
|B(\xi)| & \lesssim \frac{1}{(n \ln n)^\alpha}, \\
\mathsf{Var}(\xi) & \lesssim  \begin{cases} \frac{(\ln n)^{2+2\alpha}}{n^{2\alpha - \epsilon}}  & 0<\alpha\leq 1/2, \\ \frac{(\ln n)^{2+2\alpha}}{n^{2\alpha - \epsilon}} + \frac{p^{2\alpha-1}}{n} & 1/2 < \alpha <1.  \end{cases}
\end{align}
\item when $p>4\Delta$,
\begin{align}
|B(\xi)| & \lesssim  \frac{1}{n^\alpha(\ln n)^{2-\alpha}},\\
\mathsf{Var}(\xi) & \lesssim  \begin{cases} \frac{1}{n^{2\alpha}(\ln n)^{1-2\alpha}} & 0<\alpha\leq 1/2 ,\\ \frac{1}{n^{2\alpha}(\ln n)^{1-2\alpha}} + \frac{p^{2\alpha-1}}{n} & 1/2 < \alpha <1. \end{cases}
\end{align}
\end{enumerate}
\end{lemma}

Now the result of Theorem~\ref{thm.Falpha.risk} follows easily from Lemma~\ref{lemma.preparemainf}. We have
\begin{align}
|\mathsf{Bias}(\hat{F}_\alpha)| & \leq \sum_{i = 1}^S |B(\xi(\hat{p}_{i,1},\hat{p}_{i,2}))| \\
& \lesssim \sum_{i = 1}^S \frac{1}{(n \ln n)^\alpha} \\
& \lesssim \frac{S}{(n \ln n)^\alpha},
\end{align}
and
\begin{align}
\mathsf{Var}(\hat{F}_\alpha) & = \sum_{i = 1}^S \mathsf{Var}(\xi(\hat{p}_{i,1},\hat{p}_{i,2})) \\
& \lesssim \sum_{i = 1}^S \begin{cases} \frac{(\ln n)^{2+2\alpha}}{n^{2\alpha - \epsilon}} & 0<\alpha\leq 1/2 \\ \frac{(\ln n)^{2+2\alpha}}{n^{2\alpha - \epsilon}} + \frac{p_i^{2\alpha-1}}{n} & 1/2<\alpha<1 \end{cases} \\
& \lesssim \begin{cases} \frac{S(\ln n)^{2+2\alpha}}{n^{2\alpha - \epsilon}} & 0<\alpha\leq 1/2 \\ \frac{S(\ln n)^{2+2\alpha}}{n^{2\alpha - \epsilon}} + \sum_{i = 1}^S \frac{p_i^{2\alpha-1}}{n} & 1/2<\alpha<1 \end{cases} \\
& \lesssim \begin{cases} \frac{S(\ln n)^{2+2\alpha}}{n^{2\alpha - \epsilon}} & 0<\alpha\leq 1/2 \\ \frac{S(\ln n)^{2+2\alpha}}{n^{2\alpha - \epsilon}} + \frac{S^{2-2\alpha}}{n} & 1/2<\alpha<1 \end{cases}.
\end{align}

Here we have used the fact that
\begin{equation}
\sup_{P \in \mathcal{M}_S} \sum_{i = 1}^S p_i^{2\alpha-1} = S (1/S)^{2\alpha-1} = S^{2-2\alpha},
\end{equation}
since $x^{2\alpha-1}$ is a concave function when $1/2<\alpha<1$.

Combining the bias and variance bounds, we have
\begin{align}
& \sup_{P \in \mathcal{M}_S} \bE \left(\hat{F}_\alpha - F_\alpha \right)^2 \nonumber\\
& \quad = \left( \mathsf{Bias}(\hat{F}_\alpha) \right)^2 +\mathsf{Var}(\hat{F}_\alpha ) \\
& \quad \lesssim   \begin{cases}\frac{S^2}{(n \ln n)^{2\alpha}} +  \frac{S(\ln n)^{2+2\alpha}}{n^{2\alpha - \epsilon}} & 0<\alpha\leq 1/2 \\\frac{S^2}{(n \ln n)^{2\alpha}} + \frac{S(\ln n)^{2+2\alpha}}{n^{2\alpha - \epsilon}} +  \frac{S^{2-2\alpha}}{n} & 1/2<\alpha<1 \end{cases}
\end{align}
where $\epsilon>0$ is a constant that is arbitrarily small. Note that when $1/2<\alpha<1$, we can remove the middle term in the risk bound, since when $\ln n \lesssim \ln S$, the first term dominates, otherwise the third term dominates.

The proof of Theorem~\ref{thm.entropy.risk} is essentially the same as that for Theorem~\ref{thm.Falpha.risk}, with the only differences being replacing Lemma~\ref{lemma.largebias} with Lemma~\ref{lemma.entropyupper}, applying the entropy part of Lemma~\ref{lemma.lowpartbiasvariance} and Lemma~\ref{lemma.varentropy}. The proof of Theorem~\ref{th_2} is slightly more involved, and we need to split the analysis into four different regimes.

\begin{lemma}\label{lem_indivi_bound}
Suppose $1<\alpha<3/2$. Setting $c_1=16(\alpha+\delta), 0<10c_2\ln 2=\epsilon<2\alpha-2,\delta>0$, we have the following bounds on $|B(\xi)|$ and $\mathsf{Var}(\xi)$.
\begin{enumerate}
    \item when $p\le \frac{1}{n\ln n}$,
    \begin{align}
      |B(\xi)| &\lesssim \frac{p}{(n\ln n)^{\alpha-1}},\\
      \mathsf{Var}(\xi) &\lesssim \frac{(\ln n)^{2\alpha+2}p}{n^{2\alpha-1-\epsilon}}.
    \end{align}
    \item when $\frac{1}{n\ln n} < p \le \Delta$,
    \begin{align}
      |B(\xi)| &\lesssim \frac{1}{(n\ln n)^\alpha},\\
      \mathsf{Var}(\xi) &\lesssim \frac{(\ln n)^{2\alpha+2}}{n^{2\alpha-\epsilon}}.
    \end{align}
    \item when $p>\Delta$,
    \begin{align}
      |B(\xi)| &\lesssim \frac{1}{n^\alpha(\ln n)^{2-\alpha}},\\
      \mathsf{Var}(\xi) &\lesssim \frac{1}{n^2} + \frac{p}{n}.
    \end{align}
\end{enumerate}
\end{lemma}

Now the result of Theorem \ref{th_2} follows easily from Lemma \ref{lem_indivi_bound}. First, the total bias can be bounded by
\begin{align}
  |\mathsf{Bias}(\hat{F}_\alpha)| &\le \sum_{i=1}^S |B(\xi(\hat{p}_{i,1},\hat{p}_{i,2}))|\\
  &= \sum_{i: p_i\le \frac{1}{n\ln n}} |B(\xi(\hat{p}_{i,1},\hat{p}_{i,2}))| \nonumber \\
  & \quad + \sum_{i: \frac{1}{n\ln n}<p_i\le \Delta} |B(\xi(\hat{p}_{i,1},\hat{p}_{i,2}))| \nonumber \\
  & \quad + \sum_{i: p_i> \Delta} |B(\xi(\hat{p}_{i,1},\hat{p}_{i,2}))|\\
  &\lesssim \sum_{i: p_i\le \frac{1}{n\ln n}}\frac{p}{(n\ln n)^{\alpha-1}} + \sum_{i: \frac{1}{n\ln n}<p_i\le \Delta} \frac{1}{(n\ln n)^\alpha} \nonumber \\
  & \quad + \sum_{i: p_i>\Delta}  \frac{1}{n^\alpha(\ln n)^{2-\alpha}}\\
  &\le \frac{1}{(n\ln n)^{\alpha-1}} + \frac{1}{(n\ln n)^{\alpha}}\cdot (n\ln n) \nonumber \\
  & \quad+ \frac{1}{n^\alpha(\ln n)^{2-\alpha}}\cdot \frac{n}{c_1\ln n}\\
  &\lesssim \frac{1}{(n\ln n)^{\alpha-1}}.
\end{align}

Second, the total variance is bounded by
\begin{align}
  \mathsf{Var}(\hat{F}_\alpha) &= \sum_{i=1}^S \mathsf{Var}(\xi(\hat{p}_{i,1},\hat{p}_{i,2}))\\
  &= \sum_{i: p_i\le \frac{1}{n\ln n}} \mathsf{Var}(\xi(\hat{p}_{i,1},\hat{p}_{i,2})) \nonumber \\
  & \quad + \sum_{i: \frac{1}{n\ln n}<p_i\le\Delta} \mathsf{Var}(\xi(\hat{p}_{i,1},\hat{p}_{i,2})) \nonumber \\
  &  + \sum_{i: p_i> \Delta} \mathsf{Var}(\xi(\hat{p}_{i,1},\hat{p}_{i,2}))\\
  &\lesssim \sum_{i: p_i\le \frac{1}{n\ln n}} \frac{(\ln n)^{2\alpha+2}p}{n^{2\alpha-1-\epsilon}} \nonumber \\
  & \quad + \sum_{i: \frac{1}{n\ln n}<p_i\le\Delta} \frac{(\ln n)^{2\alpha+2}}{n^{2\alpha-\epsilon}} \nonumber \\
  & \quad + \sum_{i: p_i> \Delta} \left(\frac{p}{n}+\frac{1}{n^2}\right)\\
  &\le \frac{(\ln n)^{2\alpha+2}}{n^{2\alpha-1-\epsilon}} + \frac{(\ln n)^{2\alpha+2}}{n^{2\alpha-\epsilon}} \cdot (n\ln n) \nonumber \\
  & \quad + \left(\frac{1}{n} + \frac{1}{n^2}\cdot \frac{n}{c_1\ln n}\right)\\
  &\lesssim \frac{(\ln n)^{2\alpha+3}}{n^{2\alpha-1-\epsilon}} + \frac{1}{n}\\
  &\lesssim \frac{1}{(n\ln n)^{2\alpha-2}}.
\end{align}

Combining the bias and variance bounds, we have
\begin{align}
  \sup_{P}\mathbb{E}_P \left(\hat{F}_\alpha-F_\alpha(P)\right)^2 & = \left(\mathsf{Bias}(\hat{F}_\alpha)\right)^2 +  \mathsf{Var}(\hat{F}_\alpha) \\
  & \lesssim \frac{1}{(n\ln n)^{2\alpha-2}},\quad 1<\alpha<\frac{3}{2},
\end{align}
which completes the proof of Theorem \ref{th_2}.

\section{Minimax lower bounds for estimating $F_\alpha(P), 0<\alpha<3/2$}\label{sec.lowerbound}

There are two main lemmas that we employ towards the proof of the minimax lower bounds in Theorem~\ref{thm.Falpha.risk} and~\ref{th_1}. The first lemma is the Le Cam two-point method. Suppose we observe a random vector ${\bf Z} \in (\mathcal{Z},\mathcal{A})$ which has distribution $P_\theta$ where $\theta \in \Theta$. Let $\theta_0$ and $\theta_1$ be two elements of $\Theta$. Let $\hat{T} = \hat{T}({\bf Z})$ be an arbitrary estimator of a function $T(\theta)$ based on $\bf Z$. Le Cam's two-point method gives the following general minimax lower bound.
\begin{lemma}\label{lem_twopoint}
  \cite[Sec. 2.4.2]{Tsybakov2008} Denoting the Kullback-Leibler divergence between $P$ and $Q$ by
  \begin{align}
    D(P\|Q) \triangleq \begin{cases}
      \int \ln\left(\frac{dP}{dQ}\right)dP, &\text{if }P\ll Q,\\
      +\infty, &\text{otherwise}.
    \end{cases}
  \end{align}
  we have
  \begin{align}
& \inf_{\hat{T}} \sup_{\theta \in \Theta} \bP_\theta\left( |\hat{T} - T(\theta)| \geq \frac{|T(\theta_1)-T(\theta_0)|}{2} \right) \nonumber \\
& \quad \geq
\frac{1}{4}\exp\left(-D\left(P_{\theta_1}\|P_{\theta_0}\right)\right).
  \end{align}
\end{lemma}

The second lemma is the so-called method of two fuzzy hypotheses presented in Tsybakov \cite{Tsybakov2008}. Suppose we observe a random vector ${\bf Z} \in (\mathcal{Z},\mathcal{A})$ which has distribution $P_\theta$ where $\theta \in \Theta$. Let $\sigma_0$ and $\sigma_1$ be two prior distributions supported on $\Theta$. Write $F_i$ for the marginal distribution of $\mathbf{Z}$ when the prior is $\sigma_i$ for $i = 0,1$. 
Let $\hat{T} = \hat{T}({\bf Z})$ be an arbitrary estimator of a function $T(\theta)$ based on $\bf Z$. We have the following general minimax lower bound.

\begin{lemma}\cite[Thm. 2.15]{Tsybakov2008} \label{lemma.tsybakov}
Given the setting above, suppose there exist $\zeta\in \mathbb{R}, s>0, 0\leq \beta_0,\beta_1 <1$ such that
\begin{align}
\sigma_0(\theta: T(\theta) \leq \zeta -s) & \geq 1-\beta_0 \\
\sigma_1(\theta: T(\theta) \geq \zeta + s) & \geq 1-\beta_1.
\end{align}
If $V(F_1,F_0) \leq \eta<1$, then
\begin{equation}
\inf_{\hat{T}} \sup_{\theta \in \Theta} \bP_\theta\left( |\hat{T} - T(\theta)| \geq s \right) \geq \frac{1-\eta - \beta_0 - \beta_1}{2},
\end{equation}
where $F_i,i = 0,1$ are the marginal distributions of $\mathbf{Z}$ when the priors are $\sigma_i,i = 0,1$, respectively.
\end{lemma}

Here $V(P,Q)$ is the total variation distance between two probability measures $P,Q$ on the measurable space $(\mathcal{Z},\mathcal{A})$. Concretely, we have
\begin{equation}
V(P,Q) \triangleq \sup_{A\in \mathcal{A}} | P(A) - Q(A) | = \frac{1}{2} \int |p-q| d\nu,
\end{equation}
where $p = \frac{dP}{d\nu}, q = \frac{dQ}{d\nu}$, and $\nu$ is a dominating measure so that $P \ll \nu, Q \ll \nu$.

\subsection{Minimax lower bound for Theorem~\ref{thm.Falpha.risk} ($F_\alpha:0<\alpha<1$)}

Note that the minimax lower bound in Theorem~\ref{thm.Falpha.risk} consists of two parts when $1/2<\alpha<1$. Hence, for $1/2<\alpha<1$, it suffices to first show that
\begin{equation}\label{eqn.falphalowerfirststep}
\inf_{\hat{F}_\alpha} \sup_{P\in \mathcal{M}_S} \bE_P \left( \hat{F}_\alpha - F_\alpha(P) \right)^2 \gtrsim \frac{S^{2-2\alpha}}{n},
\end{equation}
and then show
\begin{equation}\label{eqn.falphalowersecondstep}
\inf_{\hat{F}_\alpha} \sup_{P\in \mathcal{M}_S} \bE_P \left( \hat{F}_\alpha - F_\alpha(P) \right)^2 \gtrsim \frac{S^2}{(n\ln n)^{2\alpha}},
\end{equation}
in order to obtain the desired conclusion via the relation $\max\{a,b\}\geq \frac{a+b}{2}$.

Regarding (\ref{eqn.falphalowerfirststep}), we have the following theorem.
\begin{theorem}\label{th_lower_F}
  For $\frac{1}{2}\le\alpha<1$, we have
  \begin{align}
    \inf_{\hat{F}_\alpha}\sup_{P\in\mathcal{M}_S} \bE_P\left(\hat{F}_\alpha-F_\alpha(P)\right)^2 &\gtrsim \frac{S^{2-2\alpha}}{n},
  \end{align}
  where the infimum is taken over all possible estimators $\hat{F}_\alpha$.
\end{theorem}

\begin{proof}
Applying this lemma to our Poissonized model $n\hat{p}_i\sim\mathsf{Poi}(np_i),1\le i\le S$, we know that for $\theta_1=(p_1,p_2,\cdots,p_S), \theta_0=(q_1,q_2,\cdots,q_S)$,
\begin{align}
&  D\left(P_{\theta_1}\|P_{\theta_0}\right) \nonumber \\
&\quad= \sum_{i=1}^S D\left(\mathsf{Poi}(np_i)\|\mathsf{Poi}(nq_i)\right)\\
  &\quad= \sum_{i=1}^S \sum_{k=0}^\infty \bP\left(\mathsf{Poi}(np_i)=k\right)\cdot \left[k\ln\frac{p_i}{q_i}-n(p_i-q_i)\right] \\
  &\quad= \sum_{i=1}^S np_i\ln\frac{p_i}{q_i} - n\sum_{i=1}^S(p_i-q_i)\\
  &\quad = nD(\theta_1\|\theta_0),
\end{align}
then Markov's inequality yields
\begin{align}
&  \inf_{\hat{F}} \sup_{P \in \mathcal{M}_S} \bE_P\left( \hat{F}-F_\alpha(P)\right)^2 \\
&\quad \ge \frac{|F_\alpha(\theta_1)-F_\alpha(\theta_0)|^2}{4}\times \nonumber \\
& \qquad \quad
  \inf_{\hat{F}} \sup_{P \in \mathcal{M}_S} \bP\left( |\hat{F}-F_\alpha(P)| \ge \frac{|F_\alpha(\theta_1)-F_\alpha(\theta_0)|}{2}\right)\\
  &\quad \ge \frac{|F_\alpha(\theta_1)-F_\alpha(\theta_0)|^2}{16}\exp\left(-nD(\theta_1\|\theta_0)\right),
\end{align}
where we are operating under the Poissonized model.

Fix $\epsilon\in(0,1/2)$ to be specified later. Letting
\begin{align}
  \theta_1 & = \left(\frac{1}{2(S-1)},\cdots,\frac{1}{2(S-1)},\frac{1}{2}\right), \\
  \theta_0 & = \left(\frac{1+\epsilon}{2(S-1)},\cdots,\frac{1+\epsilon}{2(S-1)},\frac{1-\epsilon}{2}\right),
\end{align}
direct computation yields
\begin{align}
  D(\theta_1\|\theta_0) & = -\frac{1}{2}\ln(1+\epsilon) + \frac{1}{2}\ln\frac{1}{1-\epsilon} \\
  & = -\frac{1}{2}\ln(1-\epsilon^2)\le \epsilon^2,
\end{align}
and
\begin{align}
&  |F_\alpha(\theta_1)-F_\alpha(\theta_0)| \nonumber \\
 &\quad= [(1+\epsilon)^\alpha-1]\cdot(2(S-1))^{1-\alpha} - 2^{-\alpha}\left[1-\left(1-\epsilon\right)^\alpha\right]\\
  &\quad \ge \left(\alpha\epsilon-\frac{\alpha(1-\alpha)\epsilon^2}{2}\right)\cdot(2(S-1))^{1-\alpha} \nonumber \\
  & \qquad - 2^{-\alpha}\left(\alpha\epsilon+\frac{\alpha(1-\alpha)\epsilon^2}{2}\right)\\
  &\quad \ge \alpha\left((2(S-1))^{1-\alpha}-2^{-\alpha}\right)\epsilon \nonumber\\
  & \qquad -\frac{\alpha(1-\alpha)}{2}\left((2(S-1))^{1-\alpha}+2^{-\alpha}\right)\epsilon^2.
\end{align}

Hence, by choosing $\epsilon=n^{-\frac{1}{2}}$, we know that
\begin{align}
& \inf_{\hat{F}} \sup_{P \in \mathcal{M}_S} \bE_P\left( \hat{F}-F_\alpha(P)\right)^2 \nonumber \\
& \  \ge \frac{\alpha^2}{16en}\left[(2(S-1))^{1-\alpha}-2^{-\alpha} -\frac{1-\alpha}{2n}\left((2(S-1))^{1-\alpha}+2^{-\alpha}\right)\right]^2
\end{align}
under the Poissonized model. Applying Lemma \ref{lemma.poissonmultinomial}, we know that under the Multinomial model, the non-asymptotic minimax lower bound is
\begin{align}
 & \frac{\alpha^2}{32en}\left[(2(S-1))^{1-\alpha}-2^{-\alpha} -\frac{1-\alpha}{4n}\left((2(S-1))^{1-\alpha}+2^{-\alpha}\right)\right]^2 \nonumber \\
 & \quad - e^{-n/4}S^{2(1-\alpha)} \gtrsim \frac{S^{2-2\alpha}}{n}.
\end{align}
The proof is complete.
\end{proof}

Now we start the proof of (\ref{eqn.falphalowersecondstep}) in earnest. For $1/2<\alpha<1$, (\ref{eqn.falphalowersecondstep}) follows directly from (\ref{eqn.falphalowerfirststep}) if $\frac{S^{2-2\alpha}}{n}\gtrsim \frac{S^2}{(n\ln n)^{2\alpha}}$, or equivalently, $S\lesssim n^{1-\frac{1}{2\alpha}}\ln n$. Hence, we only need to consider the case where $S\gtrsim n^{1-\frac{1}{2\alpha}}\ln n$, which implies that $\ln S\gtrsim \ln n$. Since the condition $\ln S\gtrsim \ln n$ is also treated as an assumption in Theorem \ref{thm.Falpha.risk} for $0<\alpha\le 1/2$, we adopt it throughout the following proof.

We construct the two fuzzy hypotheses required by Lemma~\ref{lemma.tsybakov}. Similar construction was applied in proving minimax lower bounds in~\cite{Lepski--Nemirovski--Spokoiny1999estimation} and~\cite{Cai--Low2011}. 

\begin{lemma}\label{lemma.priorconstruct}
For any given positive integer $L>0$, there exist two probability measures $\nu_0^*$ and $\nu_1^*$ on $[0,1]$ that satisfy the following conditions:
\begin{enumerate}
\item $\int t^l \nu_1^*(dt) = \int t^l \nu_0^*(dt)$, for $l = 0,1,2,\ldots,L$;
\item $\int t^\alpha \nu^*_1(dt) - \int t^\alpha\nu^*_0(dt) = 2 E_L[x^\alpha]_{[0,1]}$,
\end{enumerate}
where $E_L[x^\alpha]_{[0,1]}$ is the distance in the uniform norm on $[0,1]$ from the function $f(x) = x^{\alpha}$ to the space $\mathsf{poly}_L$ of polynomials of no more than degree $L$.
\end{lemma}

The two probability measures $\nu_0^*$ and $\nu^*_1$ can be understood as the solution to the optimization problem of maximizing $\int t^\alpha \nu_1(dt) - \int t^\alpha\nu_0(dt)$, with the constraint that $\int t^l \nu_1(dt) = \int t^l \nu_0(dt)$, for $l = 0,1,2,\ldots,L$, $\mathsf{supp}(\nu_i) \subset [0,1],i = 0,1$. Wu and Yang~\cite{Wu--Yang2014minimax} gave an explicit construction of the measures $\nu^*_0$ and $\nu^*_1$ from the solution of the best polynomial approximation problem for general functions on an interval. In some sense, the two probability measures $\nu_0^*$ and $\nu_1^*$ are chosen to be those that differ the most in terms of the expectations of the functions we care about (here is $x^\alpha$), with the same moments up to a certain order. Hence, they are difficult to distinguish via samples, but the corresponding functional values are maximally apart from each other. 

According to Lemma~\ref{lemma.nonasympxa}, we have
\begin{equation}
\lim_{L\to \infty} L^{2\alpha}E_L[x^\alpha]_{[0,1]} = \frac{\mu(2\alpha)}{2^{2\alpha}},
\end{equation}

Since we have assumed $n \gtrsim \frac{S^{1/\alpha}}{\ln S}$ and $\ln S\gtrsim \ln n$, we represent
\begin{equation}
n = c \frac{S^{1/\alpha}}{\ln S},\quad 1\lesssim c\lesssim n^{1-\delta},
\end{equation}
for some constant $\delta\in(0,1)$, which implies that
\begin{equation}
S \sim  \left( \frac{\alpha}{c} \right)^\alpha n^\alpha (\ln n)^{\alpha}.
\end{equation}
Note that
\begin{equation}
\frac{1}{c^{2\alpha}} \asymp \frac{S^2}{(n\ln n)^{2\alpha}}.
\end{equation}

Define
\begin{equation}
M = d_1 \frac{\ln n}{n}, \quad L =d_2 \ln n , \quad S' = S-1,
\end{equation}
where $d_1,d_2$ are positive constants (not depending on $n$) that will be determined later. Without loss of generality we assume that $d_2 \ln n$ is always a positive integer.

For a given integer $L$, let $\nu^*_0$ and $\nu^*_1$ be the two probability measures possessing the properties given in Lemma~\ref{lemma.priorconstruct}. Let $g(x) = Mx$ and let $\mu_i$ be the measures on $[0,1]$ defined by $\mu_i(A) = \nu^*_i(g^{-1}(A))$ for $i = 0,1$. It follows from Lemma~\ref{lemma.priorconstruct} that:
\begin{enumerate}
\item $\int t^l \mu_1(dt) = \int t^l \mu_0(dt)$, for $l = 0,1,2,\ldots,L$;
\item $\int t^{\alpha} \mu_1(dt) - \int t^{\alpha} \mu_0(dt) = 2M^\alpha E_L[x^\alpha]_{[0,1]}$.
\end{enumerate}

Let $\mu_1^{S'}$ and $\mu_0^{S'}$ be the product priors $\mu_i^{S'} = \prod_{j = 1}^{S'} \mu_i$. We assign these priors to the length-$S'$ vector $(p_1,p_2,\ldots,p_{S'})$. Under $\mu_0^{S'}$ or $\mu_1^{S'}$, we have almost surely
\begin{equation}
\sum_{i = 1}^{S'} p_i \leq S' M \sim d_1 \left( \frac{\alpha}{c} \right)^\alpha \frac{(\ln n)^{\alpha + 1}}{n^{1-\alpha}}   \ll 1,
\end{equation}
hence
\begin{equation}
p_S^\alpha \geq \left( 1- O\left( \frac{(\ln n)^{\alpha+1}}{n^{1-\alpha}}\right) \right)^\alpha \sim 1, \quad n\to \infty.
\end{equation}

We decompose $F_\alpha(P)$ as
\begin{equation}
F_\alpha(P) = \underline{F_\alpha}(P) + p_S^\alpha,
\end{equation}
where
\begin{equation}
 \underline{F_\alpha}(P) = \sum_{i =1}^{S'} p_i^\alpha.
\end{equation}

We argue that it suffices to show the minimax lower bound in Theorem~\ref{thm.Falpha.risk} holds when we replace $F_\alpha(P)$ by $\underline{F_\alpha}(P)$. Indeed, we just showed that
\begin{align}
\left( | F_\alpha(P) - \underline{F_\alpha}(P)| -1\right)^2 & \lesssim \frac{1}{c^{2\alpha}} \frac{(\ln n)^{2+2\alpha}}{n^{2-2\alpha}} \\
& \ll \frac{1}{c^{2\alpha}} \\
& \asymp \frac{S^2}{(n\ln n)^{2\alpha}}.
\end{align}
Hence, if there exists an estimator $\tilde{F}$ that violates the minimax lower bound for $F_\alpha(P)$ in Theorem~\ref{thm.Falpha.risk}, then $\tilde{F}-1$ will violates the same minimax lower for estimating $\underline{F_\alpha}(P)$, which will contradict what we show below.


For $Y|p \sim \mathsf{Poi}(np),p\sim \mu_0$, we denote the marginal distribution of $Y$ by $F_{0,M}(y)$, whose pmf can be computed as
\begin{equation}
F_{0,M}(y) = \int \frac{e^{-np} (np)^y}{y!} \mu_0(dp).
\end{equation}
We define $F_{1,M}(y)$ in a similar fashion.

\begin{lemma}\label{lemma.lowermeanvar}
The following bounds are true if $d_1 = 1, d_2 = 10e$:
\begin{align}
\bE_{\mu_1^{S'}} \underline{F_\alpha}(P) - \bE_{\mu_0^{S'}} \underline{F_\alpha}(P) & =  2 \left( \frac{\alpha}{c} \right)^\alpha \frac{\mu(2\alpha) d_1^\alpha}{(2 d_2)^{2\alpha}}(1 + o(1)) \\
&  \asymp \frac{1}{c^\alpha}, \\
\mathsf{Var}_{\mu_j^{S'}}(\underline{F_\alpha}(P)) & \leq \left( \frac{\alpha d_1^2}{c} \right)^\alpha \frac{(\ln n)^{3\alpha}}{n^\alpha} \\
& \stackrel{n\to \infty}{\longrightarrow} 0, \quad j = 0,1, \\
V(F_{1,M}, F_{0,M}) & = \frac{1}{2}\sum_{y = 0}^\infty |F_{1,M}(y) - F_{0,M}(y)| \\
&  \leq \frac{1}{n^6}.
\end{align}
\end{lemma}

Setting
\begin{align*}
\sigma_j & = \mu_j^{S'}, j = 0,1,\\
\theta & = (p_1,p_2,\ldots,p_{S-1}), \\
T(\theta) & = \underline{F_\alpha}(P), \\
 s & = \frac{1}{2}\left( \frac{\alpha}{c} \right)^\alpha \frac{\mu(2\alpha) d_1^\alpha}{(2 d_2)^{2\alpha}} ,\\
\zeta & =  \bE_{\mu_0^{S'}} \underline{F_\alpha}(P)+2s
\end{align*}
in Lemma~\ref{lemma.tsybakov}, it follows from Chebyshev's inequality and $c\lesssim n^{1-\delta}$ that
\begin{align}
\beta_0 & = \sigma_0(\underline{F_\alpha}(P) >\zeta -s) \\
&  = \sigma_0(\underline{F_\alpha}(P) - \bE_{\sigma_0} \underline{F_\alpha}(P) > s) \\
& \leq \frac{\mathsf{Var}_{\sigma_0}(\underline{F_\alpha}(P))}{s^2} \\
& \lesssim \left(\frac{c(\ln n)^3}{n}\right)^\alpha\to 0,
\end{align}
and
\begin{align}
\beta_1 & = \sigma_1(\underline{F_\alpha}(P) <\zeta + s) \\ & = \sigma_1(\underline{F_\alpha}(P) - \bE_{\sigma_1} \underline{F_\alpha}(P) < -s) \\
& \leq \frac{\mathsf{Var}_{\sigma_1}(\underline{F_\alpha}(P))}{s^2} \\
& \lesssim \left(\frac{c(\ln n)^3}{n}\right)^\alpha \to 0.
\end{align}

Also, it follows from the general fact that $V(\prod_{i = 1}^n P_i, \prod_{i = 1}^n Q_i ) \leq \sum_{i = 1}^n V(P_i,Q_i)$ (which follows easily from a coupling argument \cite{Lindvall2002lectures}) that
\begin{equation}
\eta \leq  \frac{S'}{n^6} = O(n^{-5}) \to 0,
\end{equation}

Applying Lemma~\ref{lemma.tsybakov}, we have
\begin{equation}
\inf_{\hat{F}} \sup_{P \in \mathcal{M}_S} \bP \left( |\hat{F} - \underline{F_\alpha}| \geq s \right) \geq \frac{1}{2},\quad n \to \infty.
\end{equation}

According to Markov's inequality, we have
\begin{equation}
\inf_{\hat{F}} \sup_{P \in \mathcal{M}_S} \bE \left( \hat{F} - \underline{F_\alpha} \right)^2 \geq \frac{1}{2}s^2 \asymp \frac{1}{c^{2\alpha}} \asymp \frac{S^2}{(n\ln n)^{2\alpha}}.
\end{equation}

\subsection{Minimax lower bound for Theorem~\ref{th_1} ($F_\alpha: 1<\alpha<3/2$)}\label{sec.lowerbound.>1}

First we assume that $S=n\ln n$. Similar to Lemma \ref{lemma.priorconstruct}, we construct two measures as follows for $\alpha\in(1,3/2)$.
\begin{lemma}\label{lem_measure}
  For any $0<\eta<1$ and positive integer $L>0$, there exist two probability measures $\nu_0$ and $\nu_1$ on $[\eta,1]$ such that
  \begin{enumerate}
    \item $\int t^{l} \nu_1(dt) = \int t^{l} \nu_0(dt)$, for all $l=0,1,2,\cdots,L$;
    \item $\int t^{\alpha-1} \nu_1(dt) - \int t^{\alpha-1} \nu_0(dt) = 2E_L[x^{\alpha-1}]_{[\eta,1]}$,
  \end{enumerate}
  where $E_L[x^\beta]_{[\eta,1]}$ is the distance in the uniform norm on $[\eta,1]$ from the function $f(x)=x^\beta$ to the space spanned by $\{1,x,\cdots,x^L\}$.
\end{lemma}

Based on Lemma \ref{lem_measure}, two new measures $\tilde{\nu}_0,\tilde{\nu}_1$ can be constructed as follows: for $i=0,1$, the restriction of $\tilde{\nu}_i$ on $[\eta,1]$ is absolutely continuous with respect to $\nu_i$, with the Radon-Nikodym derivative given by
\begin{align}\label{eq:derivative}
  \frac{d\tilde{\nu}_i}{d\nu_i}(t) = \frac{\eta}{t}\le 1, \qquad t\in[\eta,1],
\end{align}
and $\tilde{\nu}_i(\{0\})=1-\tilde{\nu}_i([\eta,1])\ge0$. Hence, $\tilde{\nu}_0,\tilde{\nu}_1$ are both probability measures on $[0,1]$, with the following properties
\begin{enumerate}
  \item $\int t^1 \tilde{\nu}_1(dt) = \int t^1 \tilde{\nu}_0(dt) = \eta$;
  \item $\int t^l \tilde{\nu}_1(dt) = \int t^l \tilde{\nu}_0(dt)$, for all $l=2,\cdots,L+1$;
  \item $\int t^{\alpha} \tilde{\nu}_1(dt) - \int t^{\alpha} \tilde{\nu}_0(dt) = 2\eta E_L[x^{\alpha-1}]_{[\eta,1]}$.
\end{enumerate}
The construction of measures $\tilde{\nu}_0,\tilde{\nu}_1$ are inspired by Wu and Yang~\cite{Wu--Yang2014minimax}.

The following lemma characterizes the properties of $E_L[x^\beta]_{[\eta,1]}$ using well-developed tools from approximation theory~\cite{Ditzian--Totik1987}. Similar results can be found in Wu and Yang~\cite{Wu--Yang2014minimax} in which they treated the logarithmic function.
\begin{lemma}\label{lem_conti}
  For $0<\beta<1/2$, there exists a universal positive constant $D$ such that
  \begin{align}
    \liminf_{L\to\infty}L^{2\beta}E_L[x^\beta]_{[(DL)^{-2},1]}>0.
  \end{align}
\end{lemma}

Define
\begin{align}
L = d_2\ln n,  \quad\eta = \frac{1}{(DL)^2}, \quad M = \frac{d_1}{S\eta} = \frac{d_1d_2^2D^2\ln n}{n},
\end{align}
with universal positive constants $d_1,d_2$ to be determined later. Without loss of generality we assume that $d_2\ln n$ is always a positive integer. By the choice of $\eta$ we know that
\begin{align}\label{eq:appro_err}
  \liminf_{n\rightarrow\infty} (\ln n)^{2(\alpha-1)}E_L[x^{\alpha-1}]_{[\eta,1]} >0.
\end{align}

Let $g(x)=Mx$ and let $\mu_i$ be the measures on $[0,M]$ defined by $\mu_i(A)=\tilde{\nu}_i(g^{-1}(A))$ for $i=0,1$. It then follows that
\begin{enumerate}
  \item $\int t^1 \mu_1(dt) = \int t^1 \mu_0(dt) = d_1/S$;
  \item $\int t^l \mu_1(dt) = \int t^l \mu_0(dt)$, for all $l=2,\cdots,L+1$;
  \item $\int t^{\alpha} \mu_1(dt) - \int t^{\alpha} \mu_0(dt) = 2\eta M^\alpha E_L[x^{\alpha-1}]_{[\eta,1]}$.
\end{enumerate}

Let $\mu_0^{S}$ and $\mu_1^{S}$ be product priors which we assign to the length-$S$ vector $P=(p_1,p_2,\cdots,p_{S})$. Note that $P$ may not be a probability distribution, we consider the set of \emph{approximate} probability vectors
\begin{align}
  \mathcal{M}_S(\gamma) \triangleq \left\{P: \left|\sum_{i=1}^S p_i-d_1\right|\le \frac{1}{(\ln n)^\gamma}\right\},
\end{align}
with universal constant $\gamma>0$ to be specified later, and further define the minimax risk under the Poissonized model for estimating $F_\alpha(P)$ with $P\in\mathcal{M}_S(\gamma)$ as
\begin{align}
  R_P(S,n,\gamma) &\triangleq \inf_{\hat{F}}\sup_{P\in \mathcal{M}_S(\gamma)} \mathbb{E}_P|\hat{F}-F_\alpha(P)|^2.
\end{align}

The equivalence of the minimax risk under the Multinomial model $R(S,n)$ (defined in (\ref{eq:minimaxrisk.multi})) and $R_P(S,n,\gamma)$ is established in the following lemma.
\begin{lemma}\label{lem_equiv}
  For any $S,n\in\mathbb{N},\gamma>0, 1<\alpha<3/2$, we have
  \begin{align}
        R(S,\frac{d_1n}{2})\ge \frac{1}{2d_1^{2\alpha}}R_P(S,n,\gamma) - \frac{1}{d_1^{2\alpha}}\exp(-\frac{d_1n}{8}) - \frac{4M^{2\alpha-2}}{d_1^{2\alpha}(\ln n)^{2\gamma}}.
  \end{align}
\end{lemma}

In light of Lemma \ref{lem_equiv}, it suffices to consider $R_P(S,n,\gamma)$ to give a lower bound of $R(S,n)$. Denote
\begin{align}
  \chi & \triangleq\mathbb{E}_{\mu_1^{S}}F_\alpha(P) - \mathbb{E}_{\mu_0^{S}}F_\alpha(P) \\
  & = 2\eta M^\alpha E_L[x^{\alpha-1}]_{[\eta,1]}\cdot S \\
  & =2d_1M^{\alpha-1}E_L[x^{\alpha-1}]_{[\eta,1]},
\end{align}
and
\begin{align}
  E_i \triangleq \mathcal{M}_{S}(\gamma)\bigcap \left\{P: |F_\alpha(P)-\mathbb{E}_{\mu_i^{S}}F_\alpha(P)|\le\frac{\chi}{4}\right\},\qquad i=0,1.
\end{align}
Applying Chebyshev's inequality and the union bound yields that
\begin{align}
  \mu_i^{S}[(E_i)^c]&\le \mu_i^{S}\left[\left|\sum_{j=1}^{S} p_j-d_1\right|> \frac{1}{(\ln n)^\gamma}\right] \nonumber \\
  & \quad + \mu_i^{S}\left[|F_\alpha(P)-\mathbb{E}_{\mu_i^{S}}F_\alpha(P)|>\frac{\chi}{4}\right]\\
  &\le (\ln n)^{2\gamma}\sum_{j=1}^{S} \mathsf{Var}_{\mu_i^{S}}[p_j] + \frac{16}{\chi^2}\mathsf{Var}_{\mu_i^{S}}[F_\alpha(P)]\\
  &\le (\ln n)^{2\gamma}SM^2 + \frac{16}{\chi^2}SM^{2\alpha}\\
  &= \frac{d_1^2d_2^4D^4(\ln n)^{2\gamma+3}}{cn} + \frac{4d_2^4D^4(\ln n)^3}{cn(E_L[x^{\alpha-1}]_{[\eta,1]})^2}\\
  &\to 0\text{ as }n\to\infty,\label{eq:appro}
\end{align}
where (\ref{eq:appro}) follows from (\ref{eq:appro_err}). Denote by $\pi_i$ the conditional distribution defined as
\begin{align}
  \pi_i(A) = \frac{\mu_i^S(E_i\cap A)}{\mu_i^S(E_i)}, \qquad i=0,1.
\end{align}
Now consider $\pi_0,\pi_1$ as two priors and $F_0,F_1$ as the corresponding marginal distributions. Setting
\begin{align}
  \zeta & = \mathbb{E}_{\mu_0^{S}}F_\alpha(P) + \frac{\chi}{2}, \\
   s &= \frac{\chi}{4},\\
   d_1 & =\frac{1}{(10eD)^2},\\
   d_2 & =10e, \\
   \gamma & =2\alpha,
\end{align}
we have $\beta_0=\beta_1=0$. The total variational distance is then upper bounded by
\begin{align}
  V(F_0,F_1) &\le V(F_0,G_0) + V(G_0,G_1) + V(G_1,F_1)\label{eqn.triangletv} \\
  &\le  \mu_0^{S}[(E_0)^c] + V(G_0,G_1) + \mu_1^{S}[(E_1)^c]\label{eqn.dataprocessingtv}\\
  &\le \mu_0^{S}[(E_0)^c] + \frac{S}{n^6} + \mu_1^{S}[(E_1)^c]\label{eq:variation}\\
  &\to 0,\label{eq:variation_2}
\end{align}
where $G_i$ is the marginal probability under prior $\mu_i^S$. Equation~(\ref{eqn.triangletv}) follows from the triangle inequality of the total variation distance, and (\ref{eqn.dataprocessingtv}) follows from the data processing inequality satisfied by the total variation distance. Equation (\ref{eq:variation}) is given by Lemma~\ref{lemma.lowermeanvar}, and (\ref{eq:variation_2}) follows from (\ref{eq:appro}). The idea of converting approximate priors $\mu_i^S$ into priors $\pi_i$ via conditioning comes from Wu and Yang~\cite{Wu--Yang2014minimax}.

It follows from Lemma \ref{lemma.tsybakov} and Markov's inequality that
\begin{align}
  R_P(S,n,\gamma) &\ge s^2\inf_{\hat{F}}\sup_{P\in\mathcal{M}_S(\gamma)}\mathbb{P}\left(|\hat{F}-F_\alpha(\theta)|\ge s\right) \\
  &\ge \frac{1-V(F_0,F_1)}{32}\chi^2\\
  &= \frac{d_1^2(1-V(F_0,F_1))}{8}\left(\frac{\ln n}{n}\right)^{2(\alpha-1)} \times \nonumber \\
  & \qquad (E_L[x^{\alpha-1}]_{[\eta,1]})^2.
\end{align}

Now we consider the scale $(S,n)=(m\ln m,\frac{d_1m}{2})$, and it follows from (\ref{eq:appro_err}) and Lemma \ref{lem_equiv} that for this scale,
\begin{align}
  &\liminf_{n\to\infty} (n\ln n)^{2(\alpha-1)}\cdot\inf_{\hat{F}}\sup_{P\in\mathcal{M}_S}\mathbb{E}_P|\hat{F}-F_\alpha(P)|^2 \\
  &=\liminf_{m\to\infty} \left(\frac{d_1m}{2}\ln (\frac{d_1m}{2})\right)^{2(\alpha-1)} \cdot R(S,\frac{d_1m}{2}) \\
  &\ge \left(\frac{d_1}{2}\right)^{2(\alpha-1)} \liminf_{m\to\infty} (m\ln m)^{2(\alpha-1)}\times \nonumber \\
 & \qquad \Big[\frac{1}{2d_1^{2\alpha}}R_P(S,m,\gamma)-\frac{1}{d_1^{2\alpha}}\exp\left(-\frac{d_1m}{8}\right)\nonumber \\
 & \qquad-\frac{4}{d_1^{2\alpha}(\ln m )^{2\gamma}}\left(\frac{\ln m}{m}\right)^{2\alpha-2}\Big]\\
  &\ge \liminf_{m\to \infty} \Big[\frac{1-V(F_0,F_1)}{2^{2\alpha+2}}
  \left((\ln m)^{2(\alpha-1)}E_L[x^{\alpha-1}]_{[\eta,1]}\right)^2 \nonumber \\
  & \qquad   -\frac{(m\ln m)^{2(\alpha-1)}}{2^{2(\alpha-1)}d_1^{2}}\exp\left(-\frac{d_1m}{8}\right)  -\frac{2^{4-2\alpha}}{d_1^{2}(\ln n)^4}\Big]\\
  &> 0.
\end{align}
Then the proof is completed by choosing any $c_0>2/d_1$ in Theorem \ref{th_1} by noticing that under this scale,
\begin{align}
  \lim_{n\to\infty}\frac{S}{n\ln n}=\lim_{m\to\infty} \frac{m\ln m}{(d_1m/2)\ln(d_1m/2)} = \frac{2}{d_1}.
\end{align}

\section{Experiments}\label{sec.experiments}

As mentioned in the Introduction, the implementation of our algorithm is extremely efficient and has linear complexity with respect to the sample size $n$, independent of the support size. The only overhead that deserves special mention is the computation of the best polynomial approximation, which is performed via the Remez algorithm \cite{Remez1934determination} \emph{offline} before obtaining any samples. The Chebfun team \cite{Trefethen--chebfunv5} provides a highly optimized implementation of the Remez algorithm in Matlab \cite{Pachon--Trefethen2009barycentric}. In numerical analysis, the convergence of an algorithm is called quadratic if
the error $e_m$ after the $m$-th computation satisfies $e_m \leq  C \alpha^{2^m}$ for some $C>0$ and $0<\alpha<1$. Under some assumptions about the function to approximate, one can prove \cite[Pg. 96]{Devore--Lorentz1993} the quadratic convergence of the Remez algorithm. Empirical experiments partially validate the efficiency of the Remez algorithm, which computes order $500$ best polynomial approximation for $-x \ln x, x\in [0,1]$ in a fraction of a second on a Thinkpad X220 laptop. Considering the fact that the order of approximation we conduct is logarithmic in $n$, in practice we do not need to perform this computation: we simply precompute the best polynomial approximation coefficients for various orders (e.g., up to order $200$) and store them in the software.

We emphasize that although the value of constants $c_1,c_2$ required in Lemma~\ref{lemma.preparemainf} lead to rather poor constants in the bias and variance bounds, the practical performance could be much better than what the theoretical bounds guarantee. It is due to the conservative nature of our worst-case (maximum $L_2$ risk) formalism and the fact that we use upper bounds in the analysis that, while being optimal up to a multiplicative constant, are not absolutely tightest for a fixed $S$ and $n$. Practically, experimentation shows that $c_1 \in [0.05,0.2],c_2 = 0.7$ results in very effective entropy estimation. In our experiments, we do not conduct ``splitting'' and lose half of the samples, and we evaluate our estimator on the multinomial rather than the Poisson sampling model required for the analysis. Moreover, we do not remove the constant term in the best polynomial approximation in experiments, since one can show they do not influence the achievability of the minimax rates, and it is easier to conduct best polynomial approximation with the constant term.


Given the extensive literature on entropy estimation, we demonstrate the efficacy of our methodology in functional estimation in estimating entropy. Specifically, we compare our estimator with the following estimators proposed in the literature:
\begin{enumerate}
  \item The MLE: the entropy of the empirical distribution $\hat{H}^{\mathsf{MLE}}=\sum_{i=1}^S -\hat{p}_i\ln \hat{p}_i$. It has been shown in~\cite{Paninski2003,Jiao--Venkat--Weissman2014MLE} that this approach cannot achieve the minimax rates.
  \item The Miller-Madow bias-corrected estimator \cite{Miller1955}: $\hat{H}^{\mathsf{MM}}=\hat{H}^{\mathsf{MLE}} + \frac{S-1}{2n}$. It has been shown~\cite{Paninski2003,Jiao--Venkat--Weissman2014MLE} that this approach cannot achieve the minimax rates.
  \item The Jackknifed MLE \cite{Miller1974jackknife}: $\hat{H}^{\mathsf{JK}}(\mathbf{Z}) = n\hat{H}^{\mathsf{MLE}}(\mathbf{Z}) - \frac{n-1}{n}\sum_{j=1}^n \hat{H}^{\mathsf{MLE}}(\mathbf{Z}^{-j})$, where $\mathbf{Z}^{-j}$ is the remaining sample by removing $j$-th observation. It has been shown in~\cite{Paninski2003} that this approach cannot achieve the minimax rates.
  \item The unseen estimator by Valiant and Valiant \cite{Valiant--Valiant2013estimating}: the estimator in \cite{Valiant--Valiant2011} is the first estimator shown to achieve the optimal sample complexity $n \asymp \frac{S}{\ln S}$ in entropy estimation. Recently, Valiant and Valiant \cite{Valiant--Valiant2013estimating} provided a modification of \cite{Valiant--Valiant2011} to estimate entropy, and demonstrated its superior empirical performance via comparison with various existing algorithms, even with the algorithm proposed in Valiant and Valiant \cite{Valiant--Valiant2011}. Hence, it is most informative to compare our algorithm with that of \cite{Valiant--Valiant2013estimating}. In our experiments, we downloaded and used the Matlab implementation of the estimator in \cite{Valiant--Valiant2013estimating}, with default parameters.
  \item The coverage adjusted estimator (CAE) \cite{Chao--Shen2003nonparametric,Vu--Yu--Kass2007coverage}: an estimator specifically designed to apply to settings in which there is a significant component of the distribution that is unseen. Defining $\hat{C} = 1-\frac{f_1}{n}$, where $f_1$ denotes the number of symbols that only appear once in the sample~\cite{Good1953population}, the CAE estimator is then given by
      \begin{align}
        \hat{H}^{\mathsf{CAE}} = \sum_{i=1}^S \frac{(-\hat{C}\hat{p}_i\ln (\hat{C}\hat{p}_i))}{1-(1-\hat{C}\hat{p}_i)^n}.
      \end{align}
  \item The best upper bound estimator (BUB) \cite{Paninski2003}: an estimator proposed by Paninski which minimizes the sum of upper bounds of the squared bias and the variance, which are given by approximation theory and the bounded-difference inequality \cite{Boucheron--Lugosi--Massart2013}, respectively. Note that this estimator requires the knowledge of $S$, and we give the true support size as its input. Hence we are comparing our estimator with the best-case performance of the BUB estimator.
  \item The shrinkage estimator \cite{Hausser--Strimmer2009entropy}: the plug-in estimator of the shrinkage estimate of the distribution, which is given by
     \begin{align}
       \hat{p}_i^s & = \lambda \cdot \frac{1}{S} + (1-\lambda) \cdot \hat{p}_i  \\
      \lambda & = \frac{1-\sum_{i=1}^S \hat{p}_i^2}{(n-1)(\sum_{i=1}^S \hat{p}_i^2-1/S)}.
     \end{align}
     Hence the distribution estimate is shrunk towards the uniform distribution, and $\hat{H}^{\mathsf{shrinkage}} = \sum_{i=1}^S -\hat{p}_i^s\ln \hat{p}_i^s$. We also give the true support size $S$ as its input.
  \item The Grassberger estimator \cite{Grassberger2008entropy}: this estimator aims to explicitly construct a function $f$ with a small bias $|\bE f(X)+p\ln p|$ for $X\sim \mathsf{B}(n,p)$. Define a sequence $\{G_n\}_{n=0}^\infty$ with
      \begin{align}
        G_n = \psi_0(n) + (-1)^n \int_0^1 \frac{x^{n-1}}{x+1}dx, \qquad n\ge 0
      \end{align}
      where $\psi_0(x)=\frac{d}{dx}\left[\ln\Gamma(x)\right]$ is the digamma function. Then the estimator is given by $\hat{H}^{\mathsf{Grassberger}} = \ln n - \sum_{i=1}^S \hat{p}_iG_{n\hat{p}_i}$.
  \item The Dirichlet-smoothed plug-in estimator \cite{Schober2013worst}: the plug-in estimator of the Bayes estimate of the distribution when the Dirichlet prior $\mathsf{Dir}(a)$ is imposed on $P\in\mathcal{M}_S$, i.e.,
      \begin{align}
        \hat{H}^{\mathsf{Dirichlet}} = \sum_{i=1}^S -\frac{n\hat{p}_i+a}{n+Sa}\ln\left(\frac{n\hat{p}_i+a}{n+Sa}\right).
      \end{align}
      We give the true support size $S$ as its input, and we set the parameter $a=\sqrt{n}/S$ to obtain a minimax estimator of the distribution $P$ under $\ell_2$ loss \cite[Example 5.4.5]{Lehmann--Casella1998theory}. It has been shown in~\cite{Han--Jiao--Weissman2015Bayes} that this approach cannot achieve the minimax rates.
  \item The Bayes estimator under Dirichlet prior \cite{Wolpert--Wolf1995}: this estimator also uses the Dirichlet prior $\mathsf{Dir}(a)$ on the distribution $P\in\mathcal{M}_S$, but then it is the Bayes estimator for $H(P)$ under this prior in lieu of the plug-in approach. Wolpert and Wolf \cite{Wolpert--Wolf1995} gives an explicit expression of this estimator:
      \begin{align}
        \hat{H}^{\mathsf{Bayes}} = \sum_{i=1}^S \frac{n\hat{p}_i+a}{n+Sa}\cdot \left(\psi_0(n+Sa+1) - \psi_0(n\hat{p}_i+a+1)\right).
      \end{align}
      It has been shown in~\cite{Han--Jiao--Weissman2015Bayes} that this approach cannot achieve the minimax rates. We feed the algorithm with true $S$ and set $a=\sqrt{n}/S$ in our experiments.
  \item The Nemenman--Shafee--Bialek estimator (NSB) \cite{Nemenman--Shafee--Bialek2002entropy,Nemenman2011coincidences}: instead of fixing some parameter $a$ in the Dirichlet prior, the NSB estimator uses an infinite Dirichlet mixture for averaging so that the prior distribution of $H(P)$ is near uniform. Then the NSB estimator is the Bayes estimator under this prior. There are some algorithmic stability issues due to the involvement of several numerical algorithms, including the Newton-Raphson iterative algorithm and numerical integration.
\end{enumerate}

\subsection{What do we want to test?}
We would like to clarify the aim of experimentation in evaluating estimators for functionals, such as entropy, compared to theoretical study. On the face of it, experimentation on simulated data seems to be the holy grail: indeed, in simulations we can compute the true expected squared error of certain estimator for certain distributions very accurately, and the smaller the risk is, the better the estimator is. However, a close inspection of this procedure demonstrates a severe limitation of the simulation approach: one can never simulate \emph{all} the possible distributions in a not-too-small subset of the space of discrete distributions with support size $S$. For example, suppose we have chosen $10$ distributions with support size $S$, and conducted experiments on entropy estimation. The only possible conclusion we may draw from these experiments is that the estimator that performs well on these $10$ distributions should be applied if the true distribution is one of the $10$ distributions. We cannot draw conclusions about other distributions with support size $S$ that were not tested, but we can never test all the distributions. The advantage of theory is to study the performance of estimators for all the possible distributions. On the practical side, if the statistician is being conservative, i.e., the statistician wants the scheme to perform well no matter what distribution might be the true distribution, then only theoretical study can resolve this issue, which is our starting point for this paper.

One might argue that in practice, one may possess some knowledge of the underlying distribution. For example, we may know that the distribution is exactly uniform, without knowledge of its support size. In that case, entropy estimation has been shown to be considerably easier~\cite{Santhanam--Orlitsky--Viswanathan2007new}: it is necessary and sufficient to take $n\gg \sqrt{S}$ samples to consistently estimate the entropy $\ln S$. Comparing with the optimal sample complexity $S/\ln S$ under the assumption that the estimator is required to estimate the entropy of \emph{any} discrete distribution with $S$ elements, the knowledge of uniformity reduces the difficulty of the problem considerably. We remark that if the statistician has convincing knowledge that the unknown distribution has certain structures (such as uniformity), then one should design schemes to exploit that prior knowledge. Recently, follow-up work~\cite{Han--Jiao--Weissman2015adaptive} showed that the estimator in the present paper is also \emph{adaptive}, i.e., in some sense it can automatically adjust itself to the unknown distribution to achieve higher accuracy. Recall Section~\ref{subsec.discussionmainresults} for more discussions on this feature.

The extensive comparative study we conduct below demonstrates that, not only does our entropy estimator enjoy strong theoretical guarantees, but also works well in practice for various distributions. Furthermore, experiments test the numerical stability, space, and time efficiency of the implementation, which are of crucial importance in practical applications.

\subsection{Convergence properties along $n =c \frac{S}{\ln S}$}

Since the optimal sample complexity in entropy estimation is $n\asymp S/\ln S$, we investigate the performance of various estimators along the scaling $n = c \frac{S}{\ln S},S\to \infty$.

In light of the proof of the lower bounds in Section \ref{sec.lowerbound}, we can construct two priors on $\mathcal{M}_S$ such that the entropies corresponding to the priors are quite different, but these two priors are hard to distinguish based on observed samples. Hence for any estimator, the arithmetic mean of the expected MSE based on distributions drawn from each prior should be lower bounded by the minimax risk. Moreover, for estimators which cannot achieve the minimax risk up to a multiplicative constant, this MSE will blow up eventually along $n = c\frac{S}{\ln S}$ as $S\to\infty$.

We choose $c = 10$, and sample $15$ points equally spaced in a logarithmic scale from $10$ to $10^6$ as candidates for support size $S$. For each support size $S$, we construct two product priors $\mu_0^S,\mu_1^S$ as in Section \ref{sec.lowerbound.>1} by replacing $x^{\alpha-1}$ with $-\ln x$ in Lemma \ref{lem_measure}. Wu and Yang~\cite{Wu--Yang2014minimax} gave an explicit construction of the priors and used them first in the proof of minimax lower bounds for entropy estimation. Then for $i=0,1$, in every sample we obtain a non-negative random vector from $\mu_i^S$, which is normalized into a probability distribution $P\in\mathcal{M}_S$. Then we take $n = 10S/\ln S$ samples from distribution $P$ and obtain an estimate of $H(P)$. We repeat all the preceding steps $20$ times by Monte Carlo experiments to obtain the empirical MSE under each prior, then we take the arithmetic mean to form the total empirical MSE. We remark that the resulting prior from this approach is nearly a least favorable prior~\cite{Wald1950statistical}, under which the Bayes risk is at least the same order of the minimax risk.

\begin{figure}[!htbp]
        \subfigure[]{
            \includegraphics[scale=0.6]{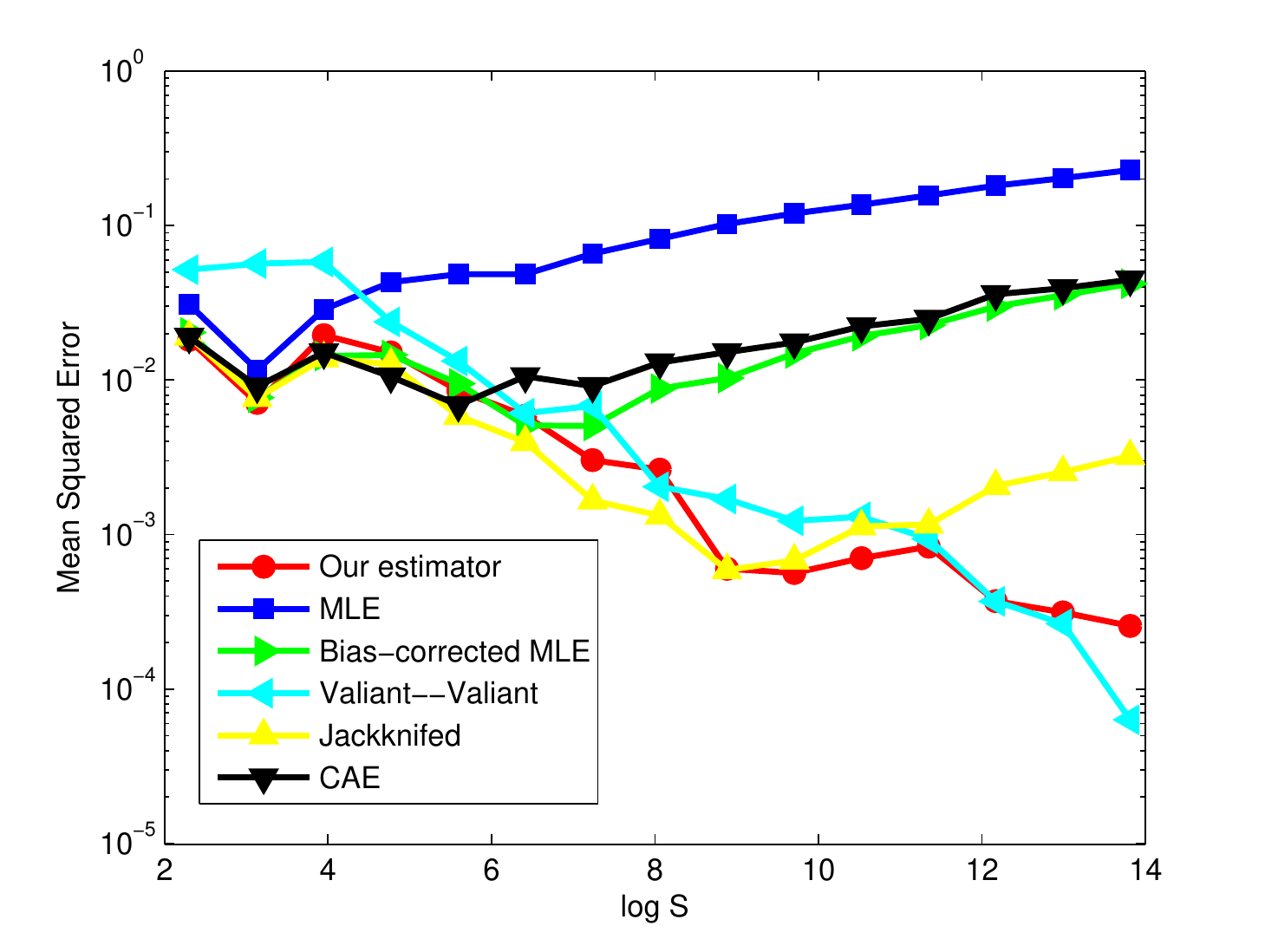}\\
            }
        \subfigure[]{
            \includegraphics[scale=0.6]{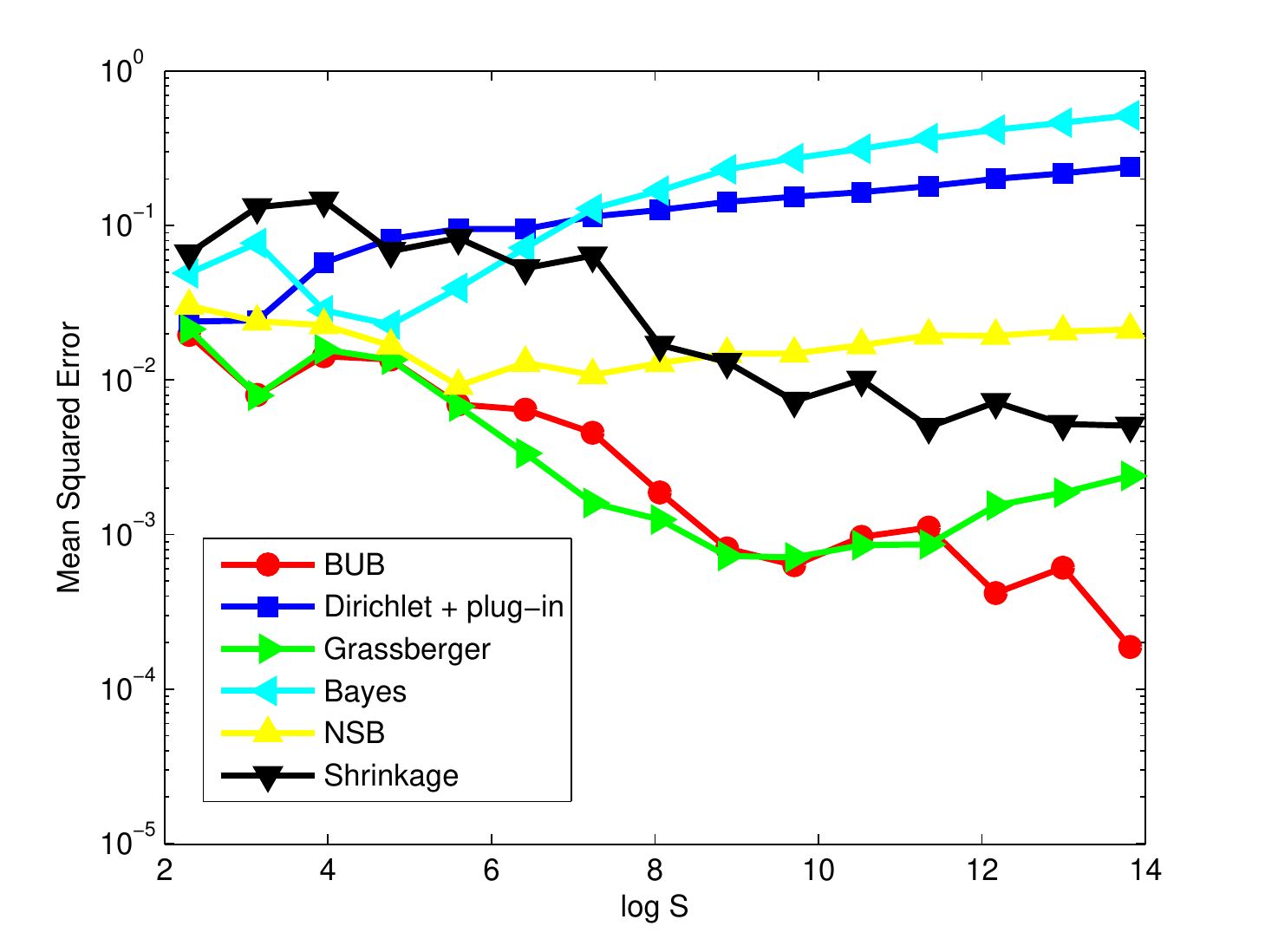}\\
            }
        \caption{The total empirical MSE of all 12 estimators along sequence $n = 10S/\ln S$, where $S$ is sampled equally spaced logarithmically from $10$ to $10^6$. The horizontal line is $\ln S$, and the vertical line is the MSE in logarithmic scale.}
        \label{Fig.Convergence}
\end{figure}

The experimental results are exhibited in Figure~\ref{Fig.Convergence}. We remark that the vertical line is the MSE in logarithmic scale, which means that a small positive slope represents exponential growth in MSE. Bearing this in mind, Figure \ref{Fig.Convergence} suggests that the following three estimators out of 12, namely our estimator, the estimator by Valiant and Valiant \cite{Valiant--Valiant2013estimating} and the best upper bound (BUB) estimator, achieve the minimax rate. Some other estimators have already been shown not to achieve the optimal sample complexity, e.g., the Miller-Madow bias-corrected MLE \cite{Paninski2003,Jiao--Venkat--Weissman2014MLE}, the jackknifed estimator \cite{Paninski2003}, and the Dirichlet-smoothed plug-in estimator as well as the Bayes estimator under Dirichlet prior~\cite{Han--Jiao--Weissman2015Bayes}. We remark that the shrinkage estimator only improves the MLE when the distribution is near-uniform, and this estimator performs poorly under the Zipf distribution (verified by our experiments), thus it attains neither the optimal minimum sample complexity nor the minimax rate.

Motivated by the preceding result, in our subsequent experiments we only consider our estimator, the estimator in \cite{Valiant--Valiant2013estimating} and the BUB estimator, as well as the MLE used as a benchmark. We also remark that all other estimators are still tested in our experiments, which consistently demonstrate the superior empirical performance of our estimator, the estimator in \cite{Valiant--Valiant2013estimating} and the BUB estimator over others.

\subsection{Estimation of entropy}
Now we examine the performance of our estimator, the estimator in \cite{Valiant--Valiant2013estimating}, the BUB estimator and the MLE in entropy estimation for various distributions. We remark that in theory, our estimator is the only one among these estimators which has been shown to achieve the minimax $L_2$ rates, while the estimator in \cite{Valiant--Valiant2013estimating} has only been shown order-optimal in terms of the sample complexity, and currently neither the maximum $L_2$ risk nor the sample complexity is known for the BUB estimator. Moreover, the BUB estimator requires an accurate upper bound for the support size $S$, and it was remarked in \cite{Vu--Yu--Kass2007coverage} that the performance of BUB degrades considerably when this bound is inaccurate.

We divide our experiments into two regimes.

\subsubsection{Data rich regime: $S \ll n$}

We first experiment in the regime $S \ll n$, which is an ``easy'' regime where even the MLE is known to perform very well. However, the estimator in \cite{Valiant--Valiant2013estimating} exhibits peculiar behavior. We sample $8$ points equally spaced in a logarithmic scale from $10^2$ to $10^3$ as candidates for support size $S$, and for each $S$ we conduct $20$ Monte Carlo simulations of estimation based on $n = 50S$ observations from certain distribution over an alphabet of size $S$. We consider two special distributions, i.e., the uniform distribution with $p_i=1/S$ and the Zipf distribution $p_i=i^{-\alpha}/\sum_{j=1}^S j^{-\alpha}$ with order $\alpha=1$ for $1\le i\le S$. The empirical root MSE is exhibited in Figure~\ref{fig.datarich}.

\begin{figure}[!htbp]
        \subfigure[Uniform Distribution]{
            \includegraphics[scale=0.6]{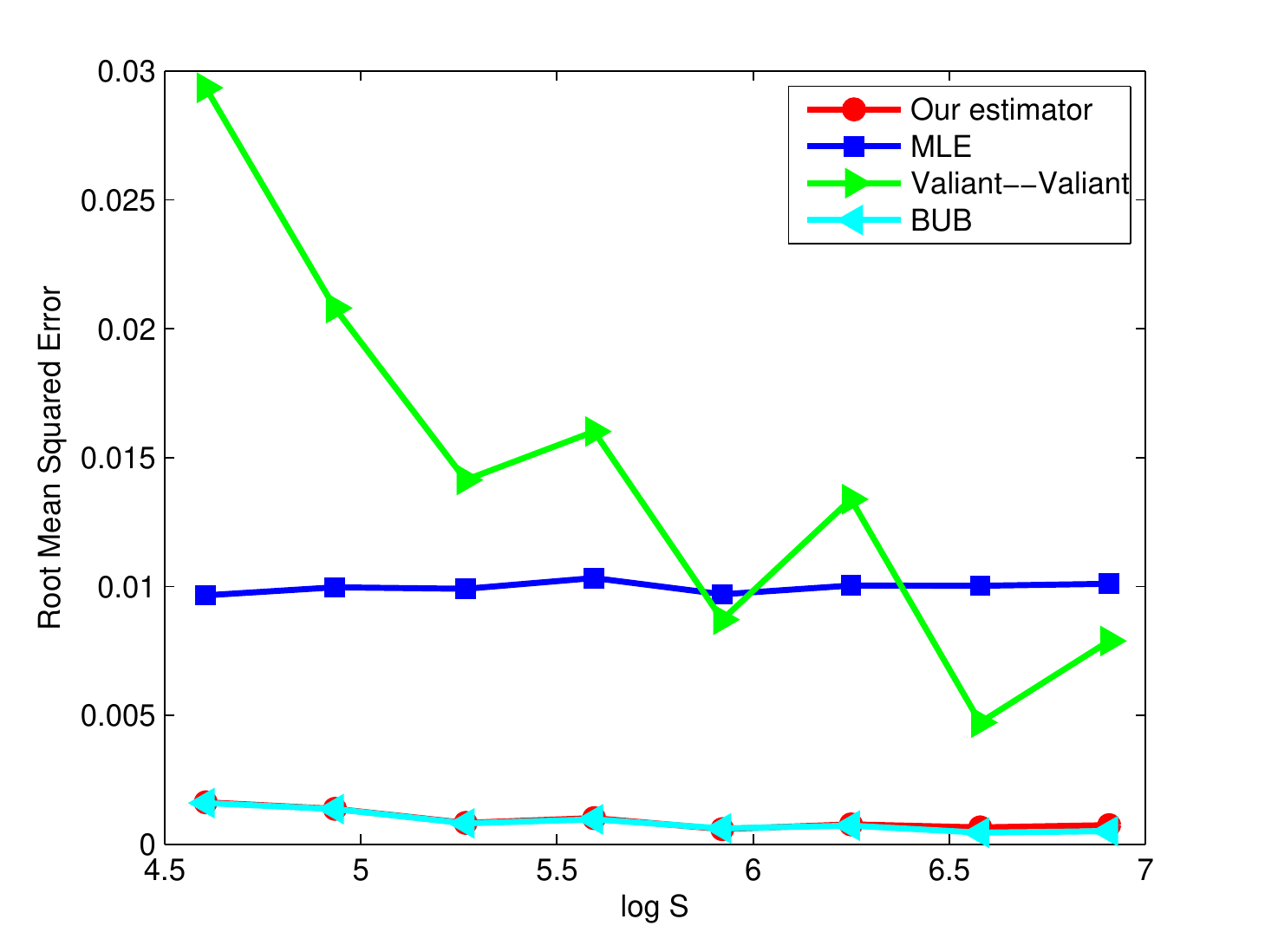}\\
            }
        \subfigure[Zipf Distribution]{
            \includegraphics[scale=0.6]{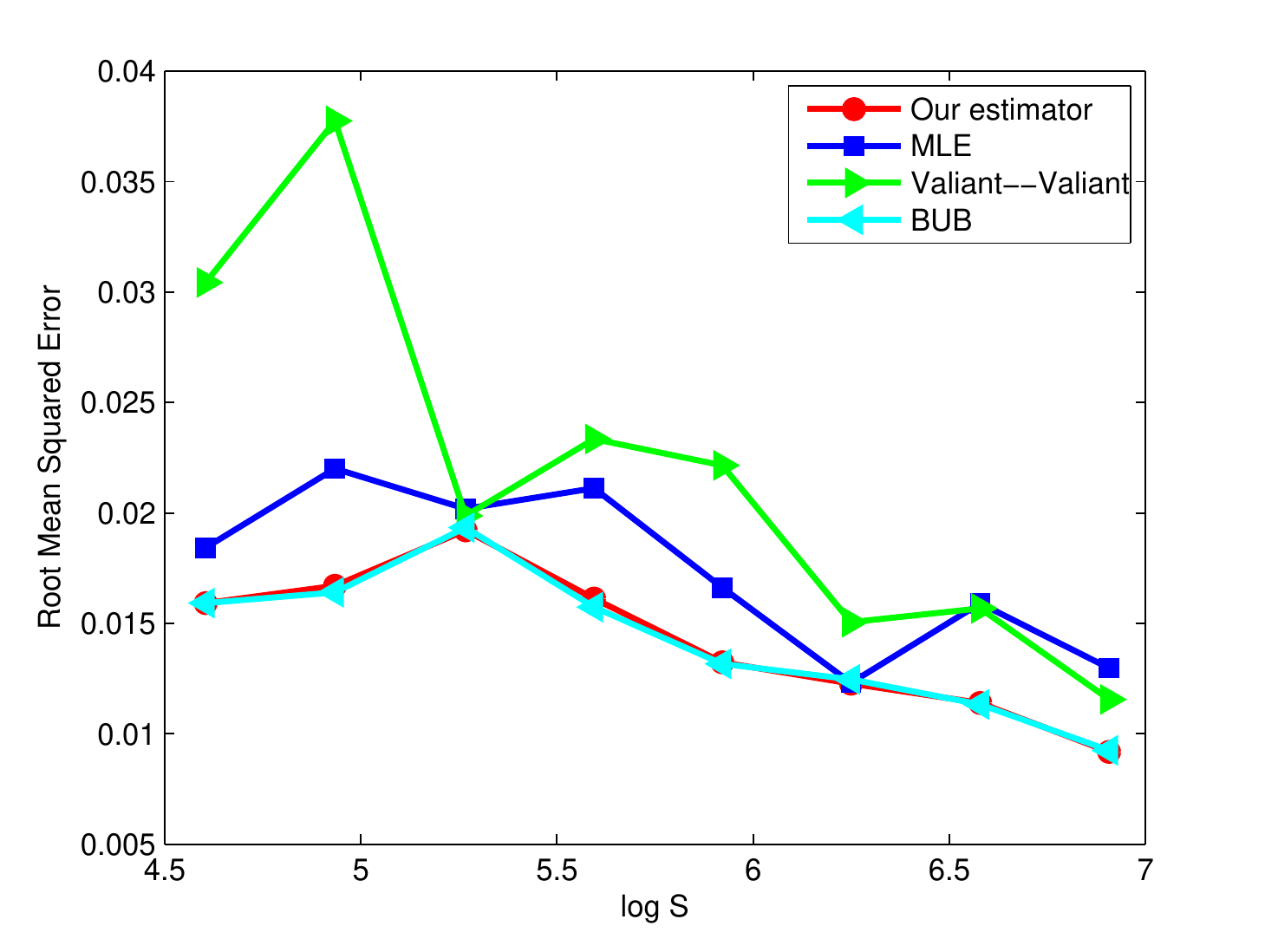}\\
            }
        \caption{The empirical MSE of our estimator, the MLE, the estimator in \cite{Valiant--Valiant2013estimating} and the BUB estimator for the uniform and Zipf distributions along sequence $n = 50S$, where $S$ is sampled equally spaced logarithmically from $10^2$ to $10^3$. The horizontal line is $\ln S$, and the vertical line is the root MSE.}
        \label{fig.datarich}
\end{figure}

It is quite clear that both our estimator and the BUB estimator perform quite well for these distributions in the data rich regime, but the MLE and the estimator in \cite{Valiant--Valiant2013estimating} return the entropy estimate which is far from the true entropy. A close inspection shows that the variance dominates the squared bias for our estimator and the BUB estimator, while there is a huge bias which constitutes the major part of the MSE for both the MLE and the estimator in \cite{Valiant--Valiant2013estimating}.

We remark that the estimator in \cite{Valiant--Valiant2013estimating} and the BUB estimator have substantially longer running time than ours in the data rich regime. For the uniform distribution, the total running time of our estimator in $160$ Monte Carlo simulations is $0.75$s, with a small overhead over the MLE which requires $0.47$s, whereas the one in \cite{Valiant--Valiant2013estimating} takes $186.72$s and the BUB estimator takes $32.97$s. Similar results hold for the Zipf distribution.

\subsubsection{Data sparse regime: $S\asymp n$ or $S \gg n$}

This is the regime where the conventional approaches such as MLE fail. We sample $8$ points equally spaced in a logarithmic scale from $10,000$ to $40,000$ as candidates for support size $S$, and for each $S$ we conduct $20$ Monte Carlo simulations of estimation based on $n = c\frac{S}{\ln S}$ observations from certain distribution over an alphabet of size $S$, where $c=5$ for the uniform distribution and $c=15$ for the Zipf distribution with $\alpha=1$. Note that when $S=20,000$, we have $n=10,098$ for $c=5$ and $n=30,293$ for $c=15$, so we are actually in the data sparse regime. The outputs of our estimator, the MLE, the estimator in \cite{Valiant--Valiant2013estimating} and the BUB estimator for both distributions are exhibited in Figure~\ref{fig.datasparse}.

\begin{figure}[!htbp]
        \subfigure[Uniform Distribution]{
            \includegraphics[scale=0.6]{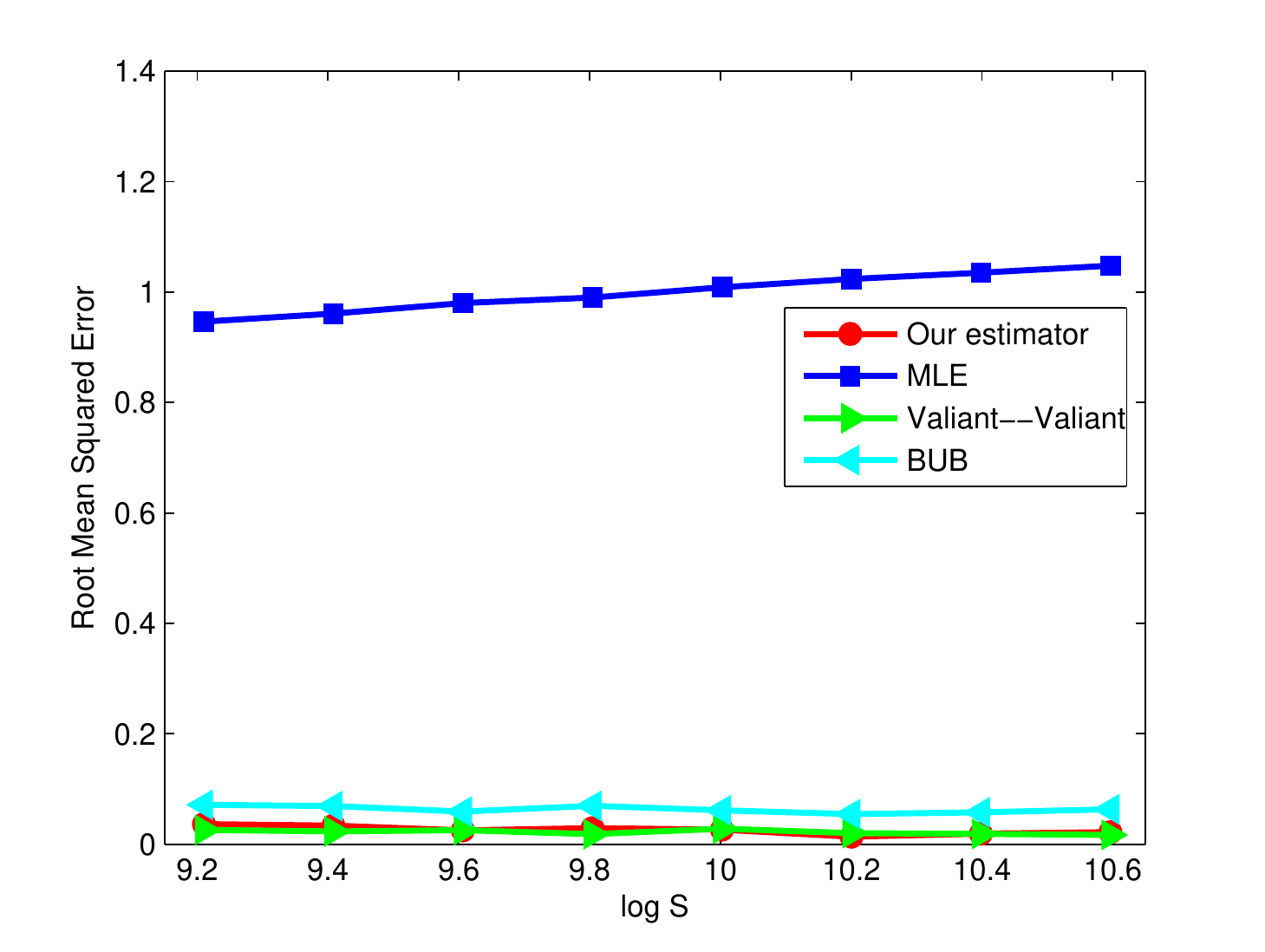}\\
            }
        \subfigure[Zipf Distribution]{
            \includegraphics[scale=0.6]{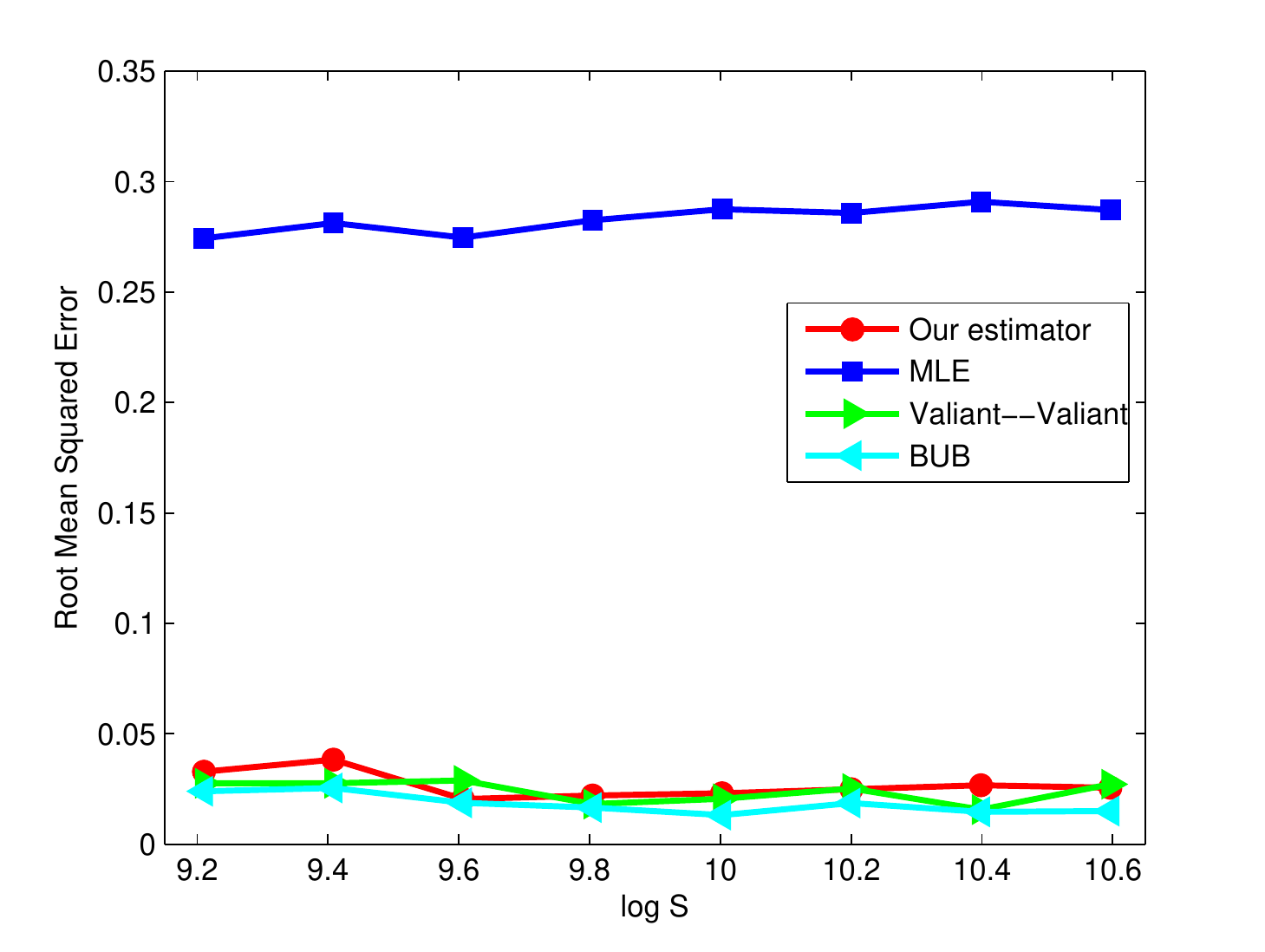}\\
            }
        \caption{The empirical root MSE of our estimator, the MLE, the estimator in \cite{Valiant--Valiant2013estimating} and the BUB estimator along sequence $n = c\frac{S}{\ln S}$, where $S$ is sampled equally spaced logarithmically from $10,000$ to $40,000$. The quantity $c=5$ for the uniform distribution and $c=15$ for the Zipf distribution. The horizontal line is $\ln S$, and the vertical line is the root MSE.}
        \label{fig.datasparse}
\end{figure}

Figure~\ref{fig.datasparse} shows that the MLE is far from the true entropy. Both our estimator and that of \cite{Valiant--Valiant2013estimating} perform quite well, but comparatively the BUB estimator exhibits a large bias for some distributions, e.g., the uniform distribution. Interestingly, with the same sample size $n = 10000$, the estimator in \cite{Valiant--Valiant2013estimating} and the BUB estimator run much faster than in the data rich regime, with a total running time $9.28$s and $2.35$s for the uniform distribution. However, it is still slower than our estimator, which takes $0.71$s, only $0.05$s longer compared with the MLE.

We have experimented with other distributions such as the Zipf with order $\alpha\neq 1$, the mixture of the uniform and Zipf distributions, as well as randomly generated distributions, with similar results. In summary, we observe that
\begin{enumerate}
\item the MLE usually concentrates at some point far away from the true functional value, particularly when the support size is comparable, or larger than the number of observations;
\item the estimator in \cite{Valiant--Valiant2013estimating} performs quite well in the data sparse regime $S\asymp n$ and $S \gg n$, but performs worse than the MLE in the data rich regime $S \ll n$, which is undesirable in applications such as mutual information estimation and situations where the support size $S$ is unknown;
\item the BUB estimator performs quite well in the data rich regime $S\ll n$, but performs worse than our estimator and the estimator in \cite{Valiant--Valiant2013estimating} for some distributions in the data sparse regime. Moreover, it requires the knowledge of $S$ and does not have theoretical guarantees on its worst-case performance thus far, which is undesirable and not convincing in practice;
\item our estimator has stable performance, linear complexity, and high accuracy.
\end{enumerate}

\subsection{Estimation of mutual information}

One functional of particular significance in various applications is the mutual information $I(X;Y)$, but it cannot be directly expressed in the form of~(\ref{eqn.generalf}). Indeed, we have
\begin{align}
I(X;Y) & = \sum_{x,y} P_{XY}(x,y) \ln \frac{P_{XY}(x,y)}{P_X(x)P_Y(y)} \\
& = \sum_{x,y} P_{XY}(x,y) \ln \frac{P_{XY}(x,y)}{ ( \sum_{y} P_{XY}(x,y) ) ( \sum_{x} P_{XY}(x,y) ) }.
\end{align}

However, one can easily show that if $X,Y$ both take values in alphabets of size $S$, then the sample complexity for estimating $I(X;Y)$ is $n \asymp S^2/\ln S$, rather than $n \asymp S^2$ required by the MLE~\cite{Jiao--Venkat--Han--Weissman2014beyond}. Applying our entropy estimator in the following way results in an essentially minimax (rate-optimal) mutual information estimator. We represent
\begin{equation}
I(X;Y) = H(X) + H(Y) - H(X,Y),
\end{equation}
where $H(X,Y)$ is the entropy associated with the joint distribution $P_{XY}$, and use our entropy estimator to estimate each term. As was exhibited in previous experiments, in the data rich regime, the BUB estimator is better than the MLE and the estimator in \cite{Valiant--Valiant2013estimating}, and in the data sparse regime, the worst-case performance of \cite{Valiant--Valiant2013estimating} is better than the BUB estimator and the MLE, and in both regimes our estimators are doing well uniformly. However, in mutual information estimation, the estimators of $H(X)$ and $H(Y)$ may be operating in the data rich regime, but that of $H(X,Y)$ in the data sparse regime. Conceivably, in this situation none of the MLE, \cite{Valiant--Valiant2013estimating} and the BUB estimator would perform well, but our estimator is expected to have good performance.

In order to investigate this intuition, we fix $n=2,500$ and sample $8$ points equally spaced in a logarithmic scale from $100$ to $200$ as candidates for support size $S$, and we generate two random variables $X,Y$ both with support size $S$ as follows. We first randomly generate two marginal distributions $P_X(i)$ and $P_Z(i), 1\leq i \leq S$, where for each $i$ we choose two independent random variables distributed as $\mathsf{Beta}(0.6,0.5)$ for $P_X(i)$ and $P_Z(i)$, and we normalize at the end to make them distributions. We pass $X$ through a transition channel to obtain $Y$, such that $Y = (X+Z) \bmod S$. Note that we are in the regime $S\ll n\ll S^2$. We conduct $20$ Monte Carlo simulations for each $S$, and the results are exhibited in Figure~\ref{fig.mutual}.

\begin{figure}[!htbp]
\centering
    \includegraphics[width = 0.45 \textwidth]{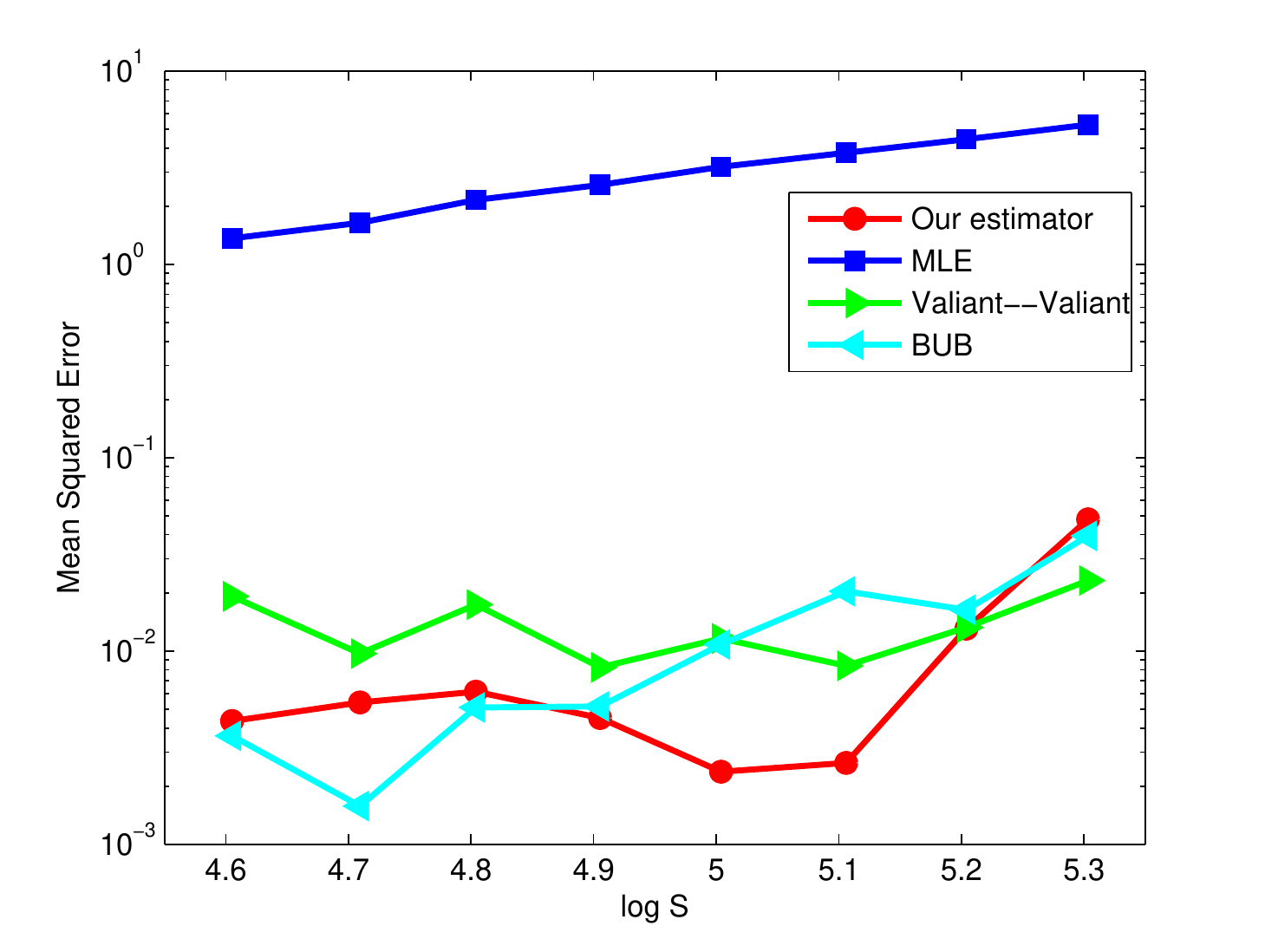}\\
    \caption{The empirical MSE of our estimator, the MLE, the estimator in \cite{Valiant--Valiant2013estimating} and the BUB estimator, where $n=2,500$ is fixed, and $S$ is sampled equally spaced logarithmically from $100$ to $200$. The horizontal line is $\ln S$, and the vertical line is the MSE in logarithmic scale. The goal is the estimate the mutual information $I(X;Y)$.}
    \label{fig.mutual}
\end{figure}

It is clear from Figure~\ref{fig.mutual} that the MLE deviates from the true mutual information significantly, but our estimator is quite accurate, performing comparably to the estimator in \cite{Valiant--Valiant2013estimating} and the BUB estimator for most support sizes $S$. Note that the true mutual information is only about 0.4, so it is a significant improvement for decreasing the MSE from 0.01 to its half. At the same time, the estimator in \cite{Valiant--Valiant2013estimating} and the BUB estimator have considerably longer running time than our estimator. It takes $128.13$s and $48.86$s for the estimator in \cite{Valiant--Valiant2013estimating} and the BUB estimator to complete the 160 simulations, respectively, whereas ours requires $1.98$s.

\subsection{Estimation of entropy rate}
Another functional of particular significance is the entropy rate $H=H(X_0|X_{-\infty}^{-1})$ of a stationary ergodic stochastic process $\{X_n\}_{n=-\infty}^\infty$, of fundamental importance in information theory~\cite{Shannon1948}. Consider a stationary ergodic Markov process $\{X_n\}_{n=0}^\infty$ with support size $S$ and memory length $D$, then the entropy rate $H$ can be expressed as
\begin{align}
  H = H(X_{D+1}|X_1^D) = H(X_1^{D+1}) - H(X_1^D)
\end{align}
where we have adopted the notation $X_m^n=(X_m,X_{m+1},\cdots,X_n)$ for $m\le n$. Hence, to estimate the entropy rate, it suffices to estimate the joint entropy $H(X_1^{D+1})$ and $H(X_1^D)$ separately, which can be accomplished by any entropy estimator $\hat{H}$. Specifically, to estimate $H(X_1^D)$, we can construct $n-D+1$ supersymbols $Y_i=X_i^{i+D-1}\overset{D}{=} Y$ for $1\le i\le n-D+1$ and then estimate $H(Y)$ using the estimator for entropy. The size of the alphabet in which $Y$ takes value is $S^D$, but we remark that the sample complexity $n\gg S^D/\ln (S^D)$ need not be optimal in this setting of a stationary ergodic stochastic process, i.e., there may exist some estimator $\hat{H}$ with a vanishing worst-case $L_2$ risk for this case even if $n\lesssim S^D/\ln (S^D)$. The reasons are twofold: the supersymbols $Y_i$ are no longer independent, and it satisfies the asymptotic equipartition property (AEP) by the Shannon-McMillan-Breiman theorem \cite{Algoet--Cover1988sandwich}. To see the role played by ergodicity, it has been shown in \cite{Han--Jiao--Weissman2014Distribution} that minimax estimation of discrete distributions with support size $S$ under $\ell_1$ loss requires $n\gg S$ samples, which turns out to be $S^D\ll n$, or equivalently, $D\ll \frac{\ln n}{\ln S}$, in the $D$-tuple distribution estimation in stochastic processes. However, \cite{Marton--Shields1994entropy} showed that it suffices to choose the memory length $D\le (1-\epsilon)\frac{\ln n}{H}$ for any $\epsilon>0$ in a stationary ergodic stochastic process, where $H$ is its entropy rate, to guarantee that the empirical joint distribution of $D$-tuple converges to the true joint distribution under $\ell_1$ loss. Noting that $H\le\ln S$ holds for any distribution, we conclude that the stochastic process is easier to handle compared with a single distribution with a large support size. Finally, we remark that minimax estimation of entropy rate for stationary ergodic processes in the finite sample setting remains largely open.

Now we examine the performance of our estimator, the MLE, the estimator in \cite{Valiant--Valiant2013estimating} and the BUB estimator in the estimation of entropy rate. We fix the memory length $D=4$, and choose the support size $S$ from $7$ to $14$. For each support size $S$, we construct a discrete distribution $P_Z$ where $P_Z(i)$ is independently drawn from $\mathsf{Beta}(0.6,0.5)$ for $1\le i\le S$ before normalization. Then for the sample size $n=1.5S^{D+1}/\ln(S^{D+1})$, we draw $n$ i.i.d. samples $Z_1,Z_2,\cdots,Z_n$ from $P_Z$ and then construct the stochastic process $\{X_k\}_{k=1}^n$ with memory length $D$ as follows: $X_k=Z_k$ for $1\le k\le D$, and $X_k=(Z_k+\sum_{j=k-D}^{k-1} X_j)\bmod S$ for $k>D$. We conduct $20$ Monte Carlo simulations for each $S$, with results exhibited in Figure~\ref{fig.entropyrate}

\begin{figure}[!htbp]
\centering
    \includegraphics[width = 0.45\textwidth]{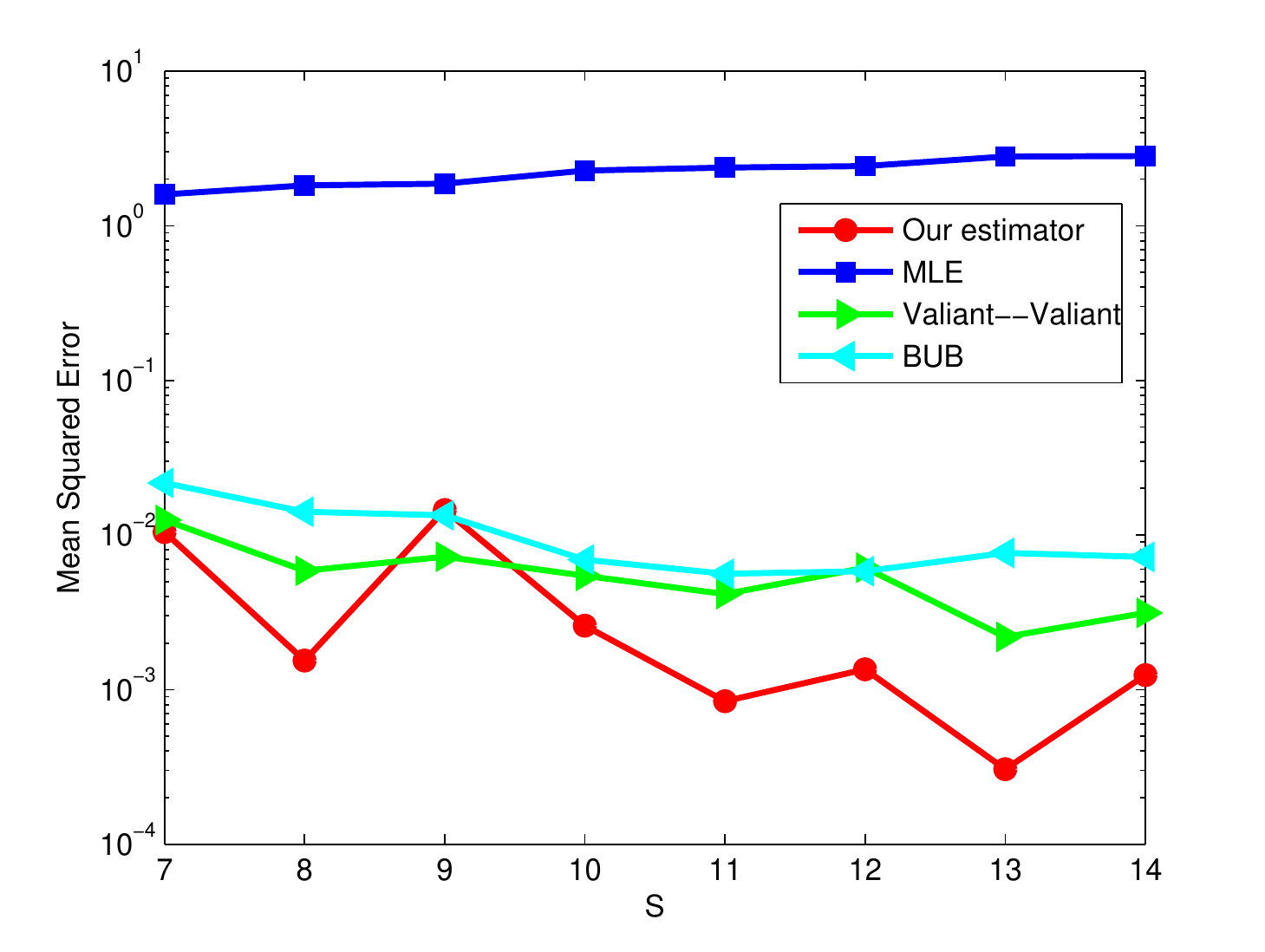}\\
    \caption{The empirical MSE of our estimator, the MLE, the estimator in \cite{Valiant--Valiant2013estimating} and the BUB estimator, where $D=4$ is fixed, $7\le S\le 14$, and $n=1.5S^{D+1}/\ln(S^{D+1})$. The horizontal line is $S$, and the vertical line is the MSE in logarithmic scale. The goal is the estimate the entropy rate $H(X_{D+1}|X_1^D)$.}
    \label{fig.entropyrate}
\end{figure}

It is clear from Figure \ref{fig.entropyrate} that our estimator performs most favorably in estimation of entropy rate for most $S$. As before, it takes our estimator less time ($6.62$s) to complete all 160 simulations compared with the estimator in \cite{Valiant--Valiant2013estimating} ($24.65$s) and the BUB estimator ($13.91$s). Due to the high accuracy and the linear complexity, our estimator is an efficient tool when dealing with high dimensional data.

\subsection{Application in learning graphical models}
Given $n$ i.i.d. samples of a random vector $\mathbf{X} = (X_1,X_2,\ldots,X_d)$, where $X_i \in \mathcal{X}, |\mathcal{X}|<\infty$, we are interested in estimating the joint distribution of $\mathbf{X}$. It was shown~\cite{Han--Jiao--Weissman2014Distribution} that one needs to take $n \gg |\mathcal{X}|^d$ samples to consistently estimate the joint distribution \cite{Paninski2004variational}, which blows up quickly with growing $d$. Practically, it is convenient and necessary to impose some structure on the joint distribution $P_{\mathbf{X}}$ to reduce the required sample complexity. Chow and Liu \cite{Chow--Liu1968} considered this problem under the constraint that the joint distribution of $\mathbf{X}$ satisfies order-one dependence. To be precise, Chow and Liu assumed that $P_{\mathbf{X}}$ can be factorized as:
\begin{equation}
P_{\mathbf{X}} = \prod_{i = 1}^d P_{X_{m_i}|X_{m_{j(i)}}},\quad 0\leq j(i)<i,
\end{equation}
where $(m_1,m_2,\ldots,m_d)$ represents an unknown permutation of the integers $(1,2,\ldots,n)$. This dependence structure can be written as a tree with the random variables as nodes.

Towards estimating $P_{\mathbf{X}}$ from $n$ i.i.d. samples, Chow and Liu~\cite{Chow--Liu1968} considered solving for the MLE under the constraint that it factors as a tree. Interestingly, this optimization problem can be efficiently solved after being transformed into a Maximum Weight Spanning Tree (MWST) problem. Chow and Liu~\cite{Chow--Liu1968} showed that the MLE of the tree structure boils down to the following expression:
\begin{align}\label{eqn.CL}
E_{\mathrm{ML}} & = \argmax_{E_Q: Q\textrm{ is a tree}} \sum_{e\in E_Q} I(\hat{P}_e),
\end{align}
where $I(\hat{P}_e)$ is the mutual information associated with the empirical distribution of the two nodes connected via edge $e$, and $E_Q$ is the set of edges of distribution $Q$ that factors as a tree. In words, it suffices to first compute the empirical mutual information between any two nodes (in total $\binom{d}{2}$ pairs), and the maximum weight spanning tree is the tree structure that maximizes the likelihood. To obtain estimates of distributions on each edge, Chow and Liu~\cite{Chow--Liu1968} simply assigned the empirical distribution.

The Chow--Liu algorithm is widely used in machine learning and statistics as a tool for dimensionality reduction, classification, and as a foundation for algorithm design in more complex dependence structures~\cite{Zhou2011structure} in the theory of learning graphical models~\cite{Wainwright--Jordan2008,Koller--Friedman2009}. It has also been widely adopted in applied research, and is particularly popular in systems biology. For example, the Chow--Liu algorithm is extensively used in the reverse engineering of transcription regulatory networks from gene expression data~\cite{Meyer--Kontos--Lafitte--Bontempi2007information}.

Considerable work has been dedicated to the theoretical properties of the CL algorithm. For example, Chow and Wagner~\cite{Chow--Wagner1973consistency} showed that the CL algorithm is consistent as $n\to \infty$. Tan et al.~\cite{Tan--Anandkumar--Tong--Willsky2011large} studies the large deviation properties of CL. However, no study justified the use of the CL in practical scenarios involving finitely many samples. Indeed, the fact that CL solves MLE does not imply it is optimal: recall Section~\ref{subsubsec.reviewmle} for the discussion on MLE. As we elaborate in what follows, this is no coincidence, as the CL can be considerably improved on in practice. To explain the insights underlying our improved algorithm, we revisit equation~(\ref{eqn.CL}) and note that if we were to replace the empirical mutual information with the true mutual information, the output of the MWST would be the true edges of the tree. In light of this, the CL algorithm can be viewed as a ``plug-in'' estimator that replaces the true mutual information with an estimate of it, namely the empirical mutual information. Naturally then, it is to be expected that a better estimate of the mutual information would lead to smaller probability of error in identifying the tree. It is thus natural to suspect that using our estimator for mutual information in lieu of the empirical mutual information in the CL algorithm would lead to performance boosts. It is gratifying to find this intuition confirmed in all the experiments that we conducted. In the following experiment, we fix $d = 7,|\mathcal{X}| = 200$, construct a star tree (i.e. all random variables are conditionally independent given $X_1$), and generate a random joint distribution by assigning independent $\mathsf{Beta}(0.5,0.5)$-distributed random variables to each entry of the marginal distribution $P_{X_1}$ and the transition probabilities $P_{X_k|X_1}, 2\leq k \leq d$ (with normalization). Then, we increase the sample size $n$ from $10^3$ to $5.5\times 10^4$, and for each $n$ we conduct $20$ Monte Carlo simulations.

Note that the true tree has $d-1 = 6$ edges, and any estimated set of edges will have at least one overlap with these $6$ edges because the true tree is a star graph. We define the wrong-edges-ratio in this case as the number of edges different from the true set of edges divided by $d-2 = 5$. Thus, if the wrong-edges-ratio equals one, it means that the estimated tree is maximally different from the true tree and, in the other extreme, a ratio of zero corresponds to perfect reconstruction. We compute the expected wrong-edges-ratio over $20$ Monte Carlo simulations for each $n$, and the results are exhibited in Figure~\ref{fig.CL_ours}.

\begin{figure}[!htbp]
\centerline{\includegraphics[width=\linewidth]{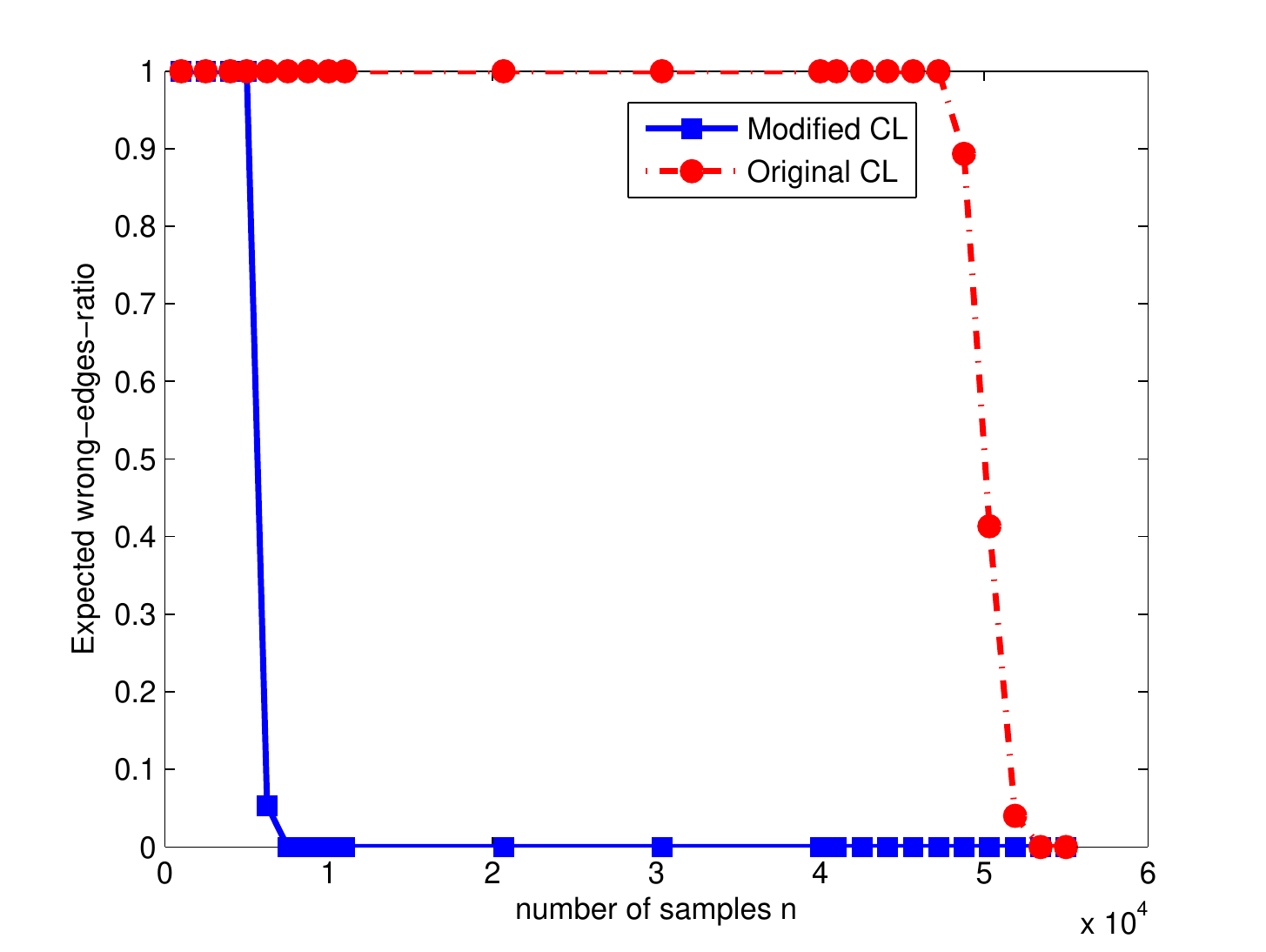}}
\caption{The expected wrong-edges-ratio of our modifed algorithm and the original CL algorithm for sample sizes ranging from $10^3$ to $5.5\times 10^4$.}\label{fig.CL_ours}
\end{figure}

Figure~\ref{fig.CL_ours} reveals intriguing phase transitions for both the modified and the original CL algorithm. When we have fewer than $3\times 10^3$ samples, both algorithms yield a wrong-edges-ratio of $1$, but soon after the sample size exceeds $6\times 10^3$, the modified CL algorithm begins to reconstruct the network perfectly, while the original CL algorithm continues to fail maximally until the sample size exceeds $47\times 10^3$, $8$ times the sample size required by the modified algorithm. The theoretical properties of these sharp phase transitions remain for future work. 

\section{Conclusions and future work}

The risk of any statistical estimator under $L_2$ loss can be decomposed into squared bias and variance. \emph{Shrinkage}~\cite{Stein1956inadmissibility}\cite{Donoho--Johnstone1994ideal}\cite{Donoho--Johnstone1994minimax} is particularly useful in reducing the variance via slightly sacrificing the bias. \emph{Approximation}, which is introduced in this paper, turns out to be the counterpart of shrinkage. The methodology of approximation has demonstrated its efficacy in reducing the bias via slightly sacrificing the variance in constructing minimax rate-optimal estimators for $H(P)$ and $F_\alpha(P)$ in this paper. We remark that in estimating functionals, especially low dimensional functionals of high dimensional parameters, the bias usually dominates the risk~\cite{Jiao--Venkat--Weissman2014MLE}, and the methodology proposed in this paper has proved to be extremely effective in properly reducing bias, and consequently achieving the minimax rates.

We show that the bias of a statistical estimator for a functional can be interpreted as the approximation error of a certain operator in approximating the function, which exhibits an intimate connection between statistics and approximation theory, the latter being a mature mathematical field studied for several centuries. Designing an estimator with smaller bias is equivalent to designing an approximation operator with small approximation error. However, the functional estimation problem is far more subtle than this connection, since one has to control the bias and variance simultaneously. We have made significant efforts to understand this delicate trade-off, which is the foundation for our general methodology in functional estimation in Section~\ref{sec.motivation}. However, we remark that it still remains fertile ground for research.

This paper constitutes a first step towards a more comprehensive understanding of functional estimation. We have partially answered Paninski's question in~[Section 8.3]\cite{Paninski2003}, since we demonstrated that the minimax rates are highly dependent on the properties of the functionals to be estimated. A general theory characterizing the minimax rates in functional estimation using approximation theoretic quantities is yet unavailable. An ambitious goal might be constructing the counterpart of what we know in nonparametric linear functional estimation~\cite{Donoho--Liu1991geometrizing2}, or nonparametric function estimation~\cite{Yang--Barron1999information}.

\section{Acknowledgments}\label{sec.ack}

We are grateful to Gregory Valiant for introducing us to the entropy estimation problem, which motivated this work. We thank many approximation theorists for very helpful discussions, in particular, Dany Leviatan, Kirill Kopotun, Feng Dai, Volodymyr Andriyevskyy, Gancho Tachev, Radu Paltanea, Paul Nevai, Doron Lubinsky, Dingxuan Zhou, and Allan Pinkus. We thank Marcelo Weinberger for bringing up the question of comparing the difficulty between entropy estimation and data compression, and Thomas Courtade for proposing to view $F_\alpha(P)$ as moment generating functions of the information density random variable and suggesting some tricks in the proof of Lemma~\ref{lemma.varentropy}. We thank Liam Paninski for interesting discussions related to Paninski \cite{Paninski2003}. We thank Martin Vinck for interesting discussions related to entropy estimation in physics and neuroscience. We thank Jayadev Acharya,
Alon Orlitsky, Ananda Theertha Suresh, and Himanshu Tyagi for stimulating discussions. We thank Yihong Wu for insightful discussions, in particular, regarding a non-rigorous step in the proof of Lemma~\ref{lemma.poissonmultinomial} in a previous version of the manuscript. We thank Maya Gupta for inspiring discussions, in particular for raising the question of whether Dirichlet prior smoothed plug-in entropy estimation can achieve the minimax rates, whose answer was shown to be negative in~\cite{Han--Jiao--Weissman2015Bayes}. Finally, we thank Dmitri Sergeevich Pavlichin for translating articles from Bernstein's collected works~\cite{Bernstein1964collected}, whose English versions are unavailable.

\appendices

\section{Auxiliary Lemmas}

\begin{lemma}\label{lemma.varentropy}
If the support of distribution $P$ is of size $S$, then
\begin{equation}
\mathsf{Var}(-\ln P(X)) \leq \begin{cases} (\ln S+1)^2 & S<56 \\ \frac{3}{4} \left( \ln S \right)^2 & S \geq 56 \end{cases}.
\end{equation}
\end{lemma}

The next lemma relates the minimax risk under the Poissonized model and that under the Multinomial model. We define the minimax risk for Multinomial model with $n$ observations on support size $S$ for estimating functional $F$ as
\begin{equation}\label{eq:minimaxrisk.multi}
R(S,n) \triangleq \inf_{\hat{F}} \sup_{P \in \mathcal{M}_S} \bE_{\mathrm{Multinomial}} \left( \hat{F} - F(P) \right)^2,
\end{equation}
and the counterpart for the Poissonized model as
\begin{equation}
R_P(S,n) \triangleq \inf_{\hat{F}} \sup_{P \in \mathcal{M}_S} \bE_{\mathrm{Poisson}} \left( \hat{F} - F(P) \right)^2.
\end{equation}

The next lemma is an extension of Wu and Yang~\cite{Wu--Yang2014minimax}.
\begin{lemma}\label{lemma.poissonmultinomial}
The minimax risks under the Poissonized model and the Multinomial model are related via the following inequalities:
\begin{equation}
R_P(S,2n) - e^{-n/4} \sup_{P \in \mathcal{M}_S}|F(P)|^2 \leq  R(S,n) \leq 2 R_P(S,n/2).
\end{equation}
\end{lemma}

The following lemma characterizes the best polynomial approximation error of $x^\alpha$ over $[0,1]$ in a very precise sense. Concretely, denoting the best polynomial approximation error with order at most $n$ for function $f$ as $E_n[f]$, we have the following lemma.
\begin{lemma}\label{lemma.nonasympxa}
The following limit exists for any $\alpha>0$:
\begin{equation}
\lim_{n\to \infty} n^{2\alpha}E_n[x^\alpha]_{[0,1]} = \frac{\mu(2\alpha)}{2^{2\alpha}},
\end{equation}
where $\mu(p) \triangleq \lim_{n\to \infty} n^{p} E_n[|x|^p]_{[-1,1]}, p>0$ is the Bernstein function introduced by \cite{Bernstein1938}. For an non-asymptotic bound, denote the best polynomial approximation of $x^{\alpha},\alpha>0$ to the $n$-th degree by $\sum_{k=0}^n g_{k,\alpha}x^k$, and define $R_{n,\alpha}(x) \triangleq \sum_{k=1}^n g_{k,\alpha}x^k$ as the best approximation polynomial without the constant term. Then for $0<\alpha<1$, we have the norm bound
\begin{align}
    \max_{0\le x\le 1}|R_{n,\alpha}(x) - x^\alpha| \le 2\left(\frac{\pi}{2n}\right)^{2\alpha}.
\end{align}

 For $1<\alpha<3/2$, we have the norm bound
  \begin{align}
    \max_{0\le x\le 1}|R_{n,\alpha}(x) - x^\alpha| \le 3\left(\frac{\pi}{n}\right)^{2\alpha},
  \end{align}
  and the pointwise bound
  \begin{align}
    \left|R_{n,\alpha}(x)-x^\alpha\right| \le \frac{D_1x}{n^{2(\alpha-1)}},
  \end{align}
  where $D_1>0$ is a universal positive constant.

  Furthermore, for any $\alpha>0$, we have
\begin{align}
  |g_{k,\alpha}| \le 2^{3n}, \quad |g_{k,H}| \leq 2^{3n}, \qquad k=1,2,\cdots,n,
\end{align}
where the coefficients $g_{k,H}$ are defined in (\ref{eqn.gkhdefine}).
\end{lemma}

Although the Bernstein function $\mu(p)$ seems hard to analyze, we can compute it fairly easily using well-developed machinery in numerical analysis. For example, \cite{Varga--Carpenter1985bernstein} showed the following bound on $\mu(1)$ using analytical methods:
\begin{equation}
0.28016 85460... \leq \mu(1) \leq   0.28017 33791,
\end{equation}
but we can easily obtain it numerically in the Chebfun system \cite{Trefethen--chebfunv5} using polynomial approximation order roughly $100$.

We have the following result by Ibragimov \cite{Ibragimov1946}:
\begin{lemma}\label{lemma.ibragimovnu12}
The following limits exists:
\begin{equation}
\lim_{n\to \infty} n^2 E_n[-x \ln x]_{[0,1]} = \frac{\nu_1(2)}{2} <\frac{1}{2}.
\end{equation}
The function $\nu_1(p)$ was introduced by Ibragimov \cite{Ibragimov1946} as the following limit for $p$ positive even integer and $m$ positive integer:
\begin{equation}
\lim_{n\to \infty} \frac{n^p}{(\ln n)^{m-1}} E_n[|x|^p \ln^m |x|]_{[-1,1]} = \nu_1(p).
\end{equation}
\end{lemma}

This Lemma follows from Ibragimov \cite[Thm. $9\delta$]{Ibragimov1946}. Note that Ibragimov \cite{Ibragimov1946} contained a small mistake where the limit of $n^2 E_n[(1-x)\ln (1-x)]_{[-1,1]}$ was wrongly computed to be $4\nu_1(2)$, but it is supposed to be $\nu_1(2)$.  Using numerical computation provided by the Chebfun \cite{Trefethen--chebfunv5} toolbox, we obtain that
\begin{equation}
\nu_1(2) \approx 0.453,
\end{equation}
and this asymptotic result starts to be very accurate for small $n$ such as $5$.

The following two lemmas characterize the approximation error of $x^\alpha$ and $-x \ln x$ when $x$ is small.
\begin{lemma}\label{lemma.approsmall}
For all $x \in [0,4\Delta]$, the following bound holds for $0<\alpha<3/2,\alpha\neq1$:
\begin{equation}
\left|\sum_{k = 1}^{K} g_{k,\alpha} \Delta^{-k + \alpha} x^k - x^\alpha \right| \leq   \frac{c_3}{(n \ln n)^\alpha},
\end{equation}
where $c_3 = 2 \left( \frac{\pi^2 c_1}{c_2^2}\right)^\alpha$ for $0<\alpha<1$, and $c_3=3\left(\frac{4\pi^2 c_1}{c_2^2}\right)^\alpha$ for $1<\alpha<3/2$. When $n$ (or equivalently, $K$) is large enough, we could take
\begin{equation}
c_3 = \frac{2 \mu(2\alpha) c_1^\alpha}{c_2^{2\alpha}},
\end{equation}
where the function $\mu(\cdot)$ is the Bernstein function introduced in Theorem~\ref{thm.bernsteinxp}.
\end{lemma}

\begin{lemma}\label{lemma.approsmallentropy}
For all $x\in [0,4\Delta]$, there exists a constant $C>0$ such that
\begin{equation}
\left|\sum_{k = 1}^{K} g_{k,H} (4\Delta)^{-k + 1} x^k + x\ln x \right| \leq   \frac{C}{n \ln n}.
\end{equation}
Moreover, when $n$ (equivalently, $K$) is large enough, we could take $C$ to be
\begin{equation}
C =\frac{4c_1 \nu_1(2)}{c_2^2} \approx \frac{1.81 c_1}{c_2^2},
\end{equation}
where the function $\nu_1(p)$ is introduced in Lemma~\ref{lemma.ibragimovnu12}.
\end{lemma}

According to Lemma~\ref{lemma.ibragimovnu12}, the asymptotic result $C \approx \frac{1.81 c_1}{c_2^2}$ starts to become very accurate even from very small values of $K$ such as $5$. The following lemma gives some tails bounds for Poisson and Binomial random variables.
\begin{lemma}\cite[Exercise 4.7]{mitzenmacher2005probability}\label{lemma.poissontail}
If $X\sim \spo(\lambda)$, or $X\sim \mathsf{B}(n,p), np = \lambda$, then for any $\delta>0$, we have
\begin{align}
\bP(X \geq (1+\delta) \lambda) & \leq \left( \frac{e^\delta}{(1+\delta)^{1+\delta}} \right)^\lambda \\
\bP(X \leq (1-\delta)\lambda) & \leq  \left( \frac{e^{-\delta}}{(1-\delta)^{1-\delta}} \right)^\lambda \leq e^{-\delta^2 \lambda/2}.
\end{align}
\end{lemma}

Next lemma gives an upper bound on the $k$-th moment of a Poisson random variable.

\begin{lemma}\label{lemma.poissonmoment}
Let $X \sim \spo(\lambda)$, $k$ be an positive integer. Taking $M = \max\{\lambda,k\}$, we have
\begin{equation}
\bE X^k \leq (2M)^k.
\end{equation}
\end{lemma}

The next two lemmas from Cai and Low \cite{Cai--Low2011} are simple facts we will utilize in the analysis of our estimators.
\begin{lemma}\cite[Lemma 4]{Cai--Low2011}\label{lemma.cailow1}
Suppose $\mathbbm{1}(A)$ is an indicator random variable independent of $X$ and $Y$, then
\begin{align}
& \mathsf{Var}(X \mathbbm{1}(A) + Y \mathbbm{1}(A^c)) \nonumber \\
& \quad = \mathsf{Var}(X) \bP(A) + \mathsf{Var}(Y) \bP(A^c) + (\bE X - \bE Y)^2 \bP(A) \bP(A^c).
\end{align}
\end{lemma}

\begin{lemma}\cite[Lemma 5]{Cai--Low2011}\label{lemma.cailow2}
For any two random variables $X$ and $Y$,
\begin{equation}
\mathsf{Var}(\min\{X,Y\}) \leq \mathsf{Var}(X) + \mathsf{Var}(Y).
\end{equation}
In particular, for any random variable $X$ and any constant $C$,
\begin{equation}
\mathsf{Var}(\min\{X,C\}) \leq \mathsf{Var}(X).
\end{equation}
\end{lemma}

\section{Proof of Theorem~\ref{th_3},~\ref{cor_3} and main lemmas}

\subsection{Proof of Theorem~\ref{th_3}}
The convexity of $x^\alpha,\alpha>1$ yields
\begin{align}
  F_\alpha(P) \ge \sum_{i=1}^S \left(\frac{1}{S}\right)^\alpha = S^{1-\alpha},
\end{align}
hence for any $\delta>0$,
\begin{align}
 & \left\{{\bf Z}:\left|\frac{\ln \hat{F}_\alpha({\bf Z})}{1-\alpha}-H_\alpha\right|\ge \delta\right\} \nonumber \\ 
  &\quad \subseteq \left\{{\bf Z}: \left|\hat{F}_\alpha({\bf Z}) - F_\alpha\right|\ge \left(1-e^{(1-\alpha)\delta}\right)F_\alpha\right\} \\
  &\quad \subseteq \left\{{\bf Z}: \left|\hat{F}_\alpha({\bf Z}) - F_\alpha\right|\ge \left(1-e^{(1-\alpha)\delta}\right)S^{1-\alpha}\right\}.
\end{align}

Theorem \ref{th_2} implies that there exists a constant $0<C_\alpha<\infty$ such that
\begin{align}
  \sup_{P}\mathbb{E}_P \left(\hat{F}_\alpha-F_\alpha(P)\right)^2 \le \frac{C_\alpha}{(n\ln n)^{2\alpha-2}},
\end{align}
where the supremum is taken over all discrete distributions supported on countably infinite alphabet. Using this estimator and applying Chebychev's inequality,
\begin{align}
 & \sup_{P}\bP\left(\left|\frac{\ln \hat{F}_\alpha}{1-\alpha}-H_\alpha\right|\ge \delta\right) \nonumber \\
  &\quad\le \sup_{P}\bP\left(\left|\hat{F}_\alpha - F_\alpha\right|\ge \left(1-e^{(1-\alpha)\delta}\right)S^{1-\alpha}\right)\\
  &\quad \le \frac{\sup_P \mathbb{E}_P\left|\hat{F}_\alpha - F_\alpha\right|^2}{\left(1-e^{(1-\alpha)\delta}\right)^2S^{2-2\alpha}}\\
  &\quad \le \frac{C_\alpha}{\left(1-e^{(1-\alpha)\delta}\right)^2}\left(\frac{S}{n\ln n}\right)^{2\alpha-2}.
\end{align}
The proof is finished by choosing
\begin{align}
  c_\alpha(\delta,\epsilon)  = \left(\frac{\epsilon\left(1-e^{(1-\alpha)\delta}\right)^2}{C_\alpha}\right)^{-\frac{1}{2\alpha-2}}.
\end{align}

\subsection{Proof of Theorem~\ref{cor_3}}
Since the central limit theorem claims that
\begin{align}
  \frac{\textsf{Poi}(\lambda)-\lambda}{\sqrt{\lambda}} \rightsquigarrow \mathcal{N}(0,1)
\end{align}
as $\lambda\to\infty$, there exists $\lambda_0>0$ such that
\begin{align}
  \mathbb{P}\left(\mathsf{Poi}(\lambda)>\lambda+1\right) \ge \frac{1}{3}, \quad \forall \lambda\ge\lambda_0.
\end{align}
Denoting $c_m\triangleq\max\{c,\lambda_0\}$, we set $S_0=\lceil\frac{n}{c_m}\rceil\le \lceil\frac{n}{c}\rceil\le S$ and consider the distribution $P=(1/S_0,1/S_0,\ldots,1/S_0,0,0,\ldots,0)$, then $H_\alpha(P)=\ln S_0$. Under the Poissonized model $n\hat{p}_i\sim\mathsf{Poi}(np_i),1\le i\le S$, we have $\hat{p}_i=0$ for $i>S_0$, and
\begin{align}
  p\triangleq\mathbb{P}\left(\hat{p}_i> \frac{c_m+1}{n}\right)\ge\frac{1}{3},\quad \forall i=1,2,\cdots,S_0.
\end{align}
Defining
\begin{align}
  N = \sum_{i=1}^{S_0} \mathbbm{1}\left(\hat{p}_i> \frac{c_m+1}{n}\right),
\end{align}
then the random variable $N$ follows a Binomial distribution $N\sim \mathsf{B}(S_0,p)$, and by the central limit theorem again we have
\begin{align}
  \liminf_{n\to\infty} \mathbb{P}\left(N\ge \frac{S_0}{6}\right) = 1.
\end{align}

Given $\eta\triangleq N/S_0\ge1/6$, it follows from the convexity of $x^\alpha,\alpha>1$ that
\begin{align}
  \sum_{i=1}^S\hat{p}_i^\alpha &= \sum_{1\le i\le S_0:\hat{p}_i>\frac{c_m+1}{n}}\hat{p}_i^\alpha +  \sum_{1\le i\le S_0:\hat{p}_i\le\frac{c_m+1}{n}}\hat{p}_i^\alpha \\
  &\ge \eta S_0\cdot \left(\frac{M}{S_0}\right)^\alpha + (S_0-\eta S_0)\cdot \left(\frac{1-\eta M}{S_0-\eta S_0}\right)^\alpha \\
  & \triangleq S_0^{1-\alpha}\cdot f(\eta,M),
\end{align}
where
\begin{align}
  \frac{1}{\eta}\ge\frac{S_0}{N}\ge M & \triangleq S_0\cdot\frac{1}{N}\sum_{1\le i\le S_0:\hat{p}_i>\frac{c_m+1}{n}}\hat{p}_i \\
  & \ge S_0\cdot\frac{c_m+1}{n} \ge 1+\frac{1}{c_m}.
\end{align}
It can be easily checked that
\begin{align}
  f(x,y) &= x y^\alpha + \frac{(1-xy)^\alpha}{(1-x)^{\alpha-1}}\\
  \frac{\partial f}{\partial y} &= \alpha xy^{\alpha-1} - \frac{\alpha x(1-xy)^{\alpha-1}}{(1-x)^{\alpha-1}} \\
  & =\frac{\alpha x}{(1-x)^{\alpha-1}} \cdot\left((y-xy)^{\alpha-1} - (1-xy)^{\alpha-1}\right) > 0, \nonumber \\
  & \qquad 0<x<xy\le 1.
\end{align}
Hence, due to $0<1/6\le\eta<1<M\le1/\eta$, we conclude that $f(\eta,M)\ge f(\eta,1+c_m^{-1})>f(\eta,1)=1$. Since $f(\eta,1+c_m^{-1})$ is continuous with respect to $\eta\in[1/6,c_m/(c_m+1)]$, we have
\begin{align}
  \hat{H}_\alpha(P_n) & \le \ln S_0 - \frac{\ln f(\eta,M)}{\alpha-1} \\
  & \le  \ln S_0 - \frac{\ln f(\eta,1+c_m^{-1})}{\alpha-1} \\
  & \le \ln S_0 - \frac{\min_{\eta\in[1/6,c_m/(c_m+1)]} \ln f(\eta,1+c_m^{-1})}{\alpha-1} \\
  &  < \ln S_0.
\end{align}
Then the proof is completed by choosing
\begin{align}
  \delta_\alpha(c) = \frac{\min_{\eta\in[1/6,c_m/(c_m+1)]} \ln f(\eta,1+c_m^{-1})}{\alpha-1} > 0.
\end{align}

\subsection{Proof of Lemma~\ref{lemma.largebias}}

For $p\geq \Delta$, we do Taylor expansion of $U_\alpha(x)$ around $x = p$. We have
\begin{align}
U_\alpha(x) & = U_\alpha(p) + U_\alpha(p)(x-p) + \frac{1}{2} U''_\alpha(p)(x-p)^2  \nonumber \\
& \quad + \frac{1}{6} U'''_\alpha(p)(x-p)^3 + R(x;p),\label{eqn.utaylor}
\end{align}
where the remainder term enjoys the following representations:
\begin{align}
R(x;p) & = \frac{1}{6}\int_p^x (x-u)^3 U_\alpha^{(4)}(u)du \\ 
& = \frac{U_\alpha^{(4)}(\xi_x)}{24}(x-p)^4,\quad \xi_x \in [\min\{x,p\}, \max\{x,p\}].
\end{align}
The first remainder is called the integral representation of Taylor series remainders, and the second remainder is called the Lagrange remainder.

Since $p\geq \Delta$, we know that
\begin{align}
U'_\alpha(p) & = \alpha p^{\alpha-1} + \frac{\alpha(1-\alpha)}{2n} (\alpha-1)p^{\alpha-2} \\
U''_\alpha(p) & = \alpha (\alpha-1)p^{\alpha-2} + \frac{\alpha(1-\alpha)(\alpha-1)(\alpha-2)}{2n} p^{\alpha-3} \\
U_\alpha^{(3)}(p) & = \alpha(\alpha-1)(\alpha-2) p^{\alpha-3} \nonumber \\
& \quad + \frac{\alpha(1-\alpha)(\alpha-1)(\alpha-2)(\alpha-3)}{2n} p^{\alpha-4} \\
U_\alpha^{(4)}(p) & = \alpha(\alpha-1)(\alpha-2)(\alpha-3) p^{\alpha-4} \nonumber \\
& \quad + \frac{\alpha(1-\alpha)(\alpha-1)(\alpha-2)(\alpha-3)(\alpha-4)}{2n} p^{\alpha-5}
\end{align}

Replacing $x$ by random variable $X$ in (\ref{eqn.utaylor}), where $nX \sim \mathsf{Poi}(np), p\geq \Delta$, and taking expectations on both sides, we have
\begin{align}
\bE U_\alpha(X) & = U_\alpha(p) + \frac{1}{2} U''_\alpha(p) \frac{p}{n} + \frac{1}{6} U'''_\alpha(p) \frac{p}{n^2} + \bE[R(X;p)]\\
& = p^\alpha + \frac{\alpha(\alpha-1)(\alpha-2)(5-3\alpha)}{12n^2}   p^{\alpha-2} \nonumber \\
& \quad - \frac{\alpha(1-\alpha)^2(2-\alpha)(3-\alpha)}{12n^3}p^{\alpha-3}  +  \bE[R(X;p)] \label{eqn.biascorrectbasic}
\end{align}
where we have used the fact that if $nX \sim \mathsf{Poi}(np)$, then $\bE (X-p)^2 = \frac{p}{n}, \bE(X-p)^3 = \frac{p}{n^2}$.

Since the representation of $R(x;p)$ involves $U_\alpha^{(4)}(\xi_x)$, it would be helpful to obtain some estimates of $U_\alpha^{(4)}(x)$ over $[0,1]$. Denoting $U_\alpha(x) = I_n(x) f(x)$, where $f(x) = x^\alpha + \frac{\alpha(1-\alpha)}{2n}x^{\alpha-1}$, we have
\begin{equation}\label{eqn.ualphataylor4}
U_\alpha^{(4)}(x) = I_n^{(4)} f + 4 I_n^{(3)} f^{(1)} + 6 I_n^{(2)} f^{(2)} + 4 I_n^{(1)} f^{(3)} + I_n f^{(4)}.
\end{equation}

Hence, it suffices to bound each term in (\ref{eqn.ualphataylor4}) separately.

For $x\in [0, t]$, $U_\alpha(x) \equiv 0$, so we do not need to consider this regime. For $x \in [2t,1]$, $U_\alpha(x) = f(x)$, hence
\begin{align}
|U^{(4)}_\alpha(x)| & = |f^{(4)}(x)| \\
& = \Bigg | \alpha(\alpha-1)(\alpha-2)(\alpha-3) x^{\alpha-4} \nonumber \\
& \qquad  + \frac{\alpha(1-\alpha)(\alpha-1)(\alpha-2)(\alpha-3)(\alpha-4)}{2n}x^{\alpha-5}\Bigg |,
\end{align}
which implies that for $x\geq 2t$,
\begin{equation}
\sup_{z \in [x,1]} |U^{(4)}_\alpha(z)| \leq 6x^{\alpha-4} + \frac{12}{n}x^{\alpha-5}.
\end{equation}

Finally we consider $x\in (t, 2t)$. Denoting $y = x-t$, the derivatives of $I_n(x)$ for $x\in (t, 2t)$ are as follows:
\begin{align}
I_n'(x) & = \frac{630 y^4(t-y)^4}{t^9}\\
I_n''(x) & = \frac{2520 y^3(t-2y)(t-y)^3 }{t^9}\\
I_n^{(3)}(x) & = \frac{2520 y^2(t-y)^2 (3t^2-14t y + 14 y^2)}{t^9}\\
I_n^{(4)}(x) & = \frac{15120 y(t-2y)(t-y)(t^2-7t y + 7y^2)}{t^9}.
\end{align}

Considering the fact that $y/t \in [0,1]$, we can maximize $|I_n^{(i)}(x)|$ over $x\in (t,2t)$ for $1\leq i \leq 4$. With the help of $\mathsf{Mathematica}$ \cite{mathematica}, we could show that for $x\in (t, 2t)$,
\begin{align}
|I_n'(x)| & \leq  \frac{4}{t}\\
|I_n''(x)| & \leq \frac{20}{t^2} \\
|I_n^{(3)}(x)| & \leq \frac{100}{t^3} \\
|I_n^{(4)}(x)| & \leq \frac{1000}{t^4}.
\end{align}

Plugging these upper bounds in (\ref{eqn.ualphataylor4}), we know for $x\in (t, 2t)$
\begin{align}
|U_\alpha^{(4)}(x)| & \leq \frac{1000}{t^4} t^\alpha + \frac{4\times 100}{t^3} t^{\alpha-1} + 6\times\frac{20}{t^2} t^{\alpha-2} \nonumber \\
& \quad  + 4\times \frac{4}{t}\times 2t^{\alpha-3} + 6  t^{\alpha-4} \\
&  \leq 1558t^{\alpha-4} \\
& \leq 1558(x/2)^{\alpha-4} \\
& \leq 24928 x^{\alpha-4}.
\end{align}

Now we proceed to upper bound $|\bE [R(X;p)]|, p\geq \Delta$. We consider the following two cases:
\begin{enumerate}
\item Case 1:$x\geq p/2$. In this case,
\begin{align}
|R(x;p)| & = \left|\frac{U_\alpha^{(4)}(\xi_x)}{24}(x-p)^4\right| \\
& \leq \sup_{x\in [p/2,1]}|U_\alpha^{(4)}(x)| \frac{(x-p)^4}{24} \\
& \leq \left(  6(p/2)^{\alpha-4} + \frac{12}{n}(p/2)^{\alpha-5} \right)\frac{(x-p)^4}{24}.
\end{align}
\item Case 2: $0\leq x<p/2$. In this case, denoting $y = \max\{x,\Delta/4\}$,
\begin{align}
& |R(x;p)| \nonumber \\
& \leq  \frac{1}{6}\int_y^p (u-x)^3 |U_\alpha^{(4)}(u)|du \\
& \leq \frac{1}{6} \int_y^p (u-x)^3 24928 u^{\alpha-4} du \\
& \leq  4155 \int_y^p \frac{(u-x)^3}{u^{4-\alpha}}du \\
& = 4155  \int_y^p \left( u^{\alpha-1}-3x u^{\alpha-2} + 3x^2 u^{\alpha-3} - x^3 u^{\alpha-4} \right)du\\
& = 4155  \Big( \frac{1}{\alpha}(p^\alpha - y^\alpha) - \frac{3x}{\alpha-1}\left( p^{\alpha-1}-y^{\alpha-1}\right) \nonumber \\
& \quad + \frac{3x^2}{\alpha-2}\left(p^{\alpha-2}-y^{\alpha-2}\right) -\frac{x^3}{\alpha-3}\left( p^{\alpha-3}-y^{\alpha-3}\right)\Big) \\
& \leq 4155  \left(  \frac{1}{\alpha}(p^\alpha - y^\alpha) + \frac{3x^2}{\alpha-2}\left(p^{\alpha-2}-y^{\alpha-2}\right) \right) \\
& = 4155 \left(  \frac{1}{\alpha}(p^\alpha - y^\alpha) + \frac{3}{\alpha-2}\left(p^\alpha\frac{x^2}{p^2}-y^\alpha\frac{x^2}{y^2}\right) \right) \\
& \leq 4155  \left( \frac{1}{\alpha}p^\alpha + \frac{3}{2-\alpha}p^\alpha\right) \\
& = \frac{8310 (1+\alpha)}{\alpha(2-\alpha)}p^\alpha.
\end{align}
\end{enumerate}

Now we have
\begin{align}
\bE[|R(X;p)|] & = \bE [|R(X;p)|\mathbbm{1}(X\geq p/2)] \nonumber \\
& \quad  + \bE[|R(X;p)|\mathbbm{1}(X<p/2)] \\
& \triangleq B_1 + B_2 .
\end{align}

For the term $B_1$, we have
\begin{align}
B_1 & = \bE [|R(X;p)|\mathbbm{1}(X\geq p/2)] \\
& \leq \left(  6(p/2)^{\alpha-4} + \frac{12}{n}(p/2)^{\alpha-5} \right) \bE[(X-p)^4]/24 \\
&  \leq  \left(  \frac{1}{4}(p/2)^{\alpha-4} + \frac{1}{2n}(p/2)^{\alpha-5} \right)\left( \frac{p}{n^3} + \frac{3p^2}{n^2}\right),
\end{align}
where we have used the fact that if $nX \sim \mathsf{Poi}(np)$, then $\bE (X-p)^4 = (np + 3n^2p^2)/n^4$.

For the term $B_2$, we have
\begin{align}
B_2 & = \bE[|R(X;p)|\mathbbm{1}(X<p/2)] \\
& \leq \frac{8310 (1+\alpha)}{\alpha(2-\alpha)}p^\alpha \bP(nX < np/2).
\end{align}

Applying Lemma~\ref{lemma.poissontail}, we have
\begin{equation}
B_2 \leq \frac{8310 (1+\alpha)}{\alpha(2-\alpha)}p^\alpha e^{-np/8} \leq \frac{8310 (1+\alpha)}{\alpha(2-\alpha)}p^\alpha n^{-c_1/8}.
\end{equation}

Hence, we have
\begin{align}
\bE[R(X;p)] & \leq \bE[|R(X;p)|] \\
& \leq  \left(  \frac{1}{4}(p/2)^{\alpha-4} + \frac{1}{2n}(p/2)^{\alpha-5} \right)\left( \frac{p}{n^3} + \frac{3p^2}{n^2}\right) \nonumber \\
& \quad +\frac{8310(1+\alpha)}{\alpha(2-\alpha)} p^\alpha n^{-c_1/8}.
\end{align}

Plugging this into (\ref{eqn.biascorrectbasic}), we have for $p\geq \Delta$,
\begin{align}
& |\bE U_\alpha(X) - p^\alpha | \nonumber \\
& \quad \leq \frac{\alpha(\alpha-1)(\alpha-2)(5-3\alpha)}{12n^2}   p^{\alpha-2}  \nonumber \\
&\quad \quad + \left(  \frac{1}{4}(p/2)^{\alpha-4} + \frac{1}{2n}(p/2)^{\alpha-5} \right)\left( \frac{p}{n^3} + \frac{3p^2}{n^2}\right) \nonumber \\
& \quad \quad +\frac{8310 (1+\alpha)}{\alpha(2-\alpha)}p^\alpha n^{-c_1/8} \\
& \quad \leq  \frac{17p^{\alpha-2}}{n^2} + \frac{8310 (1+\alpha)}{\alpha(2-\alpha)}p^\alpha n^{-c_1/8} \\
& \quad = \frac{17}{n^\alpha(c_1 \ln n)^{2-\alpha}} + \frac{8310(1+\alpha)}{\alpha(2-\alpha)}p^\alpha n^{-c_1/8}.
\end{align}

For the upper bound on the variance $\mathsf{Var}(U_\alpha(X))$, denoting $f(p) = p^\alpha + \frac{\alpha(1-\alpha)}{2n} p^{\alpha-1}$, for $p\geq \Delta$, we have
\begin{align}
\mathsf{Var}(U_\alpha(X)) & = \bE U_\alpha^2(X) - (\bE U_\alpha(X))^2 \\
& = \bE U_\alpha^2(X) - f^2(p) + f^2(p) - (\bE U_\alpha(X))^2 \\
& \leq | \bE U_\alpha^2(X) - f^2(p)| \nonumber \\
& \quad + | f^2(p) - (\bE U_\alpha(X)-f(p) + f(p))^2 |\\
& = | \bE U_\alpha^2(X) - f^2(p)|+ |(\bE U_\alpha(X)-f(p))^2 \nonumber \\
& \quad + 2f(p) (\bE U_\alpha(X)-f(p))| \\
& \leq | \bE U_\alpha^2(X) - f^2(p)| + |\bE U_\alpha(X)-f(p)|^2 \nonumber \\
& \quad + 2f(p)|\bE U_\alpha(X)-f(p)|. \label{eqn.varianceestimatelarge}
\end{align}

Hence, it suffices to obtain bounds on $| \bE U_\alpha^2(X) - f^2(p)|$ and $|\bE U_\alpha(X) - f(p)|$. Denoting $r(x) = U^2_\alpha(x)$, we know that $r(x) \in C^4[0,1]$, and it follows from Taylor's formula and the integral representation of the remainder term that
\begin{equation}
r(X) = f^2(p) + r'(p)(X-p) + R_1(X;p),
\end{equation}
\begin{align}
R_1(X;p) & =\int_p^X (X-u)r''(u)du \\
& = \frac{1}{2}r''(\eta_X)(X-p)^2, \nonumber \\
& \qquad \eta_X \in [\min\{X,p\},\max\{X,p\}].
\end{align}

Similarly, we have
\begin{equation}
U_\alpha(X) = f(p) + f'(p)(X-p) + R_2(X;p),
\end{equation}
\begin{align}
\quad R_2(X;p) & = \int_p^X (X-u)U_\alpha''(u)du \\
& = \frac{1}{2} U_\alpha''(\nu_X)(X-p)^2,\nonumber \\
& \qquad \nu_X\in [\min\{X,p\},\max\{X,p\}].
\end{align}

Taking expectation on both sides with respect to $X$, where $nX \sim \mathsf{Poi}(np), p\geq \Delta$, we have
\begin{equation}\label{eqn.u2biaslarge}
|\bE U_\alpha^2(X) - f^2(p)| = |\bE R_1(X;p) |.
\end{equation}
Similarly, we have
\begin{equation}
|\bE U_\alpha(X) - f(p)| =  | \bE R_2(X;p)|.
\end{equation}

As we did for function $U_\alpha(x)$, now we give some upper estimates for $|r''(x)|$ over $[0,1]$. Over regime $[0,t]$, $r(x)\equiv 0$, so we ignore this regime. Over regime $[2t,1]$, since $U_\alpha(x) = f(x), f(x) = x^\alpha + \frac{\alpha(1-\alpha)}{2n}x^{\alpha-1}$, we have
\begin{align}
r'(x) & = 2f f'\\
r''(x) & = 2(f')^2 + 2 f f''.
\end{align}
Hence, for $x\geq 2t$,
\begin{equation}
\sup_{z\in [x,1]}|r''(z)| \leq 4 x^{2\alpha-2}.
\end{equation}
\begin{equation}
\sup_{z\in [x,1]}|U_\alpha''(z)| \leq x^{\alpha-2}.
\end{equation}

Over regime $[t,2t]$, we have
\begin{align}
r'(x) & = 2 f f' I_n^2 + 2 I_n I_n' f^2 \\
r''(x) & = 2 \Big( (f')^2 I_n^2 + f f'' I_n^2 + 2 f f' I_n I_n' \nonumber \\
& \quad + (I_n')^2 f^2 + I_n I_n'' f^2 + 2 f f' I_n I_n' \Big).
\end{align}
Hence, we have for $x\in [t,2t]$,
\begin{align}
|r''(x)| & \leq 2 \Big( t^{2\alpha-2} +t^{2\alpha-2} + 2 t^{2\alpha-1} \frac{4}{t} + \left( \frac{4}{t} \right)^2 t^{2\alpha} \nonumber \\
& \qquad + \frac{20}{t^2}t^{2\alpha} +  2 t^{2\alpha-1} \frac{4}{t} \Big) \\
& \leq 108t^{2\alpha-2} \\
& \leq 108(x/2)^{2\alpha-2} \\
& = 432 x^{2\alpha-2}
\end{align}
Also, over regime $[t,2t]$,
\begin{equation}
U_\alpha''(x) = I_n'' f + I_n f'' + 2I_n' f',
\end{equation}
hence for $x\in [t,2t]$,
\begin{align}
|U_\alpha''(x)| & \leq \frac{20}{t^2}t^{\alpha} + t^{\alpha-2} + 2 \frac{4}{t}t^{\alpha-1} \\
& \leq 30 t^{\alpha-2} \\
& \leq 30(x/2)^{\alpha-2} \\
& \leq 120x^{\alpha-2}.
\end{align}

Now we are in the position to bound $|\bE R_1(X;p)|$ and $|\bE R_2(X;p)|$.

We have
\begin{align}
|\bE R_1(X;p)| & \leq \bE |R_1(X;p)| \\
& = \bE [|R_1(X;p)\mathbbm{1}(X\geq p/2)|] \nonumber \\
& \quad+ \bE[R_1(X;p)\mathbbm{1}(X<p/2)] \\
& \leq \bE \left[ \frac{1}{2} 4 (p/2)^{2\alpha-2} (X-p)^2\right] \nonumber \\
& \quad + \bE[R_1(X;p)\mathbbm{1}(X<p/2)] \\
& = 8\frac{p^{2\alpha-1}}{n} + \sup_{x\leq p/2}|R_1(x;p)|\bP(nX<np/2) \\
& \leq 8\frac{p^{2\alpha-1}}{n} + \sup_{x\leq p/2}|R_1(x;p)| n^{-c_1/8},
\end{align}
where in the last step we have applied Lemma~\ref{lemma.poissontail}.

Regarding $\sup_{x\leq p/2}|R_1(x;p)|$, for any $x\leq p/2$, denoting $y = \max\{x,\Delta/4\}$, we have
\begin{align}
R_1(x;p) & = \int_x^p (u-x)r''(u)du \\
& \leq \int_y^p (u-x) 432 u^{2\alpha-2}du \\
& \leq 432 \int_y^p u^{2\alpha-1}du \\
& = \frac{432}{2\alpha} (p^{2\alpha}-y^{2\alpha})\\
& \leq \frac{432}{2\alpha}p^{2\alpha} \\
& \leq \frac{216}{\alpha}p^{2\alpha}.
\end{align}

Hence, we have
\begin{equation}
|\bE R_1(X;p)| \leq \frac{8p^{2\alpha-1}}{n} + \frac{216}{\alpha}p^{2\alpha}n^{-c_1/8}.
\end{equation}

Analogously, we obtain the following bound for $|\bE R_2(X;p)|$:
\begin{equation}
|\bE R_2(X;p)|\leq 2\frac{p^{\alpha-1}}{n} + \frac{120}{\alpha}p^\alpha n^{-c_1/8}.
\end{equation}

Plugging these estimates of $|\bE R_1(X;p)|$ and $|\bE R_2(X;p)|$ into (\ref{eqn.varianceestimatelarge}), we have for $p\geq \Delta,c_1 \ln n\geq 1$,
\begin{align}
& \mathsf{Var}(U_\alpha(X)) \nonumber \\
& \leq \frac{8p^{2\alpha-1}}{n} + \frac{216}{\alpha}p^{2\alpha}n^{-c_1/8} \nonumber \\
& \quad +  \left(2\frac{p^{\alpha-1}}{n} + \frac{120}{\alpha} p^\alpha n^{-c_1/8} \right)^2 \nonumber \\
& \quad  +2f(p)\left(2\frac{p^{\alpha-1}}{n} + \frac{120}{\alpha}p^\alpha n^{-c_1/8} \right).
\end{align}

We need to distinguish two cases: $0<\alpha\leq 1/2$, and $1/2<\alpha<1$.
\begin{enumerate}
\item $0<\alpha\leq 1/2$: in this case, we have
\begin{align}
& \mathsf{Var}(U_\alpha(X)) \nonumber \\
& \leq \frac{8p^{2\alpha-1}}{n} + \frac{216}{\alpha}p^{2\alpha} n^{-c_1/8} \nonumber \\
& \quad +  \left(2\frac{p^{\alpha-1}}{n} + \frac{120}{\alpha }p^\alpha n^{-c_1/8} \right)^2 \nonumber \\
& \quad +2f(p)\left(2\frac{p^{\alpha-1}}{n} + \frac{120}{\alpha}p^\alpha n^{-c_1/8} \right) \\
& \leq \frac{8}{n^{2\alpha}(c_1\ln n)^{1-2\alpha}} + \frac{216}{\alpha}p^{2\alpha} n^{-c_1/8} \nonumber \\
& \quad + 2\left( \frac{4p^{2\alpha-2}}{n^2} + \frac{14400}{\alpha^2}p^{2\alpha} n^{-c_1/4}\right)\\
& \quad  + 2p^\alpha \left( 1+\frac{1}{8c_1 \ln n}\right) \left( 2\frac{p^{\alpha-1}}{n} + \frac{120}{\alpha}p^\alpha n^{-c_1/8} \right) \\
& \leq \frac{16}{n^{2\alpha}(c_1\ln n)^{1-2\alpha}} + \frac{8}{n^{2\alpha}(c_1\ln n)^{2-2\alpha}} \nonumber \\
& \quad + \frac{576}{\alpha} p^{2\alpha}n^{-c_1/8} + \frac{28800}{\alpha^2} p^{2\alpha}n^{-c_1/4} \\
& \leq  \frac{24}{n^{2\alpha}(c_1\ln n)^{1-2\alpha}}  + \frac{576}{\alpha}p^{2\alpha}n^{-c_1/8} \nonumber \\
& \quad + \frac{28800}{\alpha^2}p^{2\alpha} n^{-c_1/4} .
\end{align}
\item $1/2<\alpha<1$: in this case, we have
\begin{align}
& \mathsf{Var}(U_\alpha(X)) \nonumber \\
 & \leq \frac{8p^{2\alpha-1}}{n} + \frac{216}{\alpha}p^{2\alpha} n^{-c_1/8}  \nonumber \\
 & \quad +  \left(2\frac{p^{\alpha-1}}{n} + \frac{120}{\alpha}p^\alpha n^{-c_1/8} \right)^2 \nonumber \\
 &  +2f(p)\left(2\frac{p^{\alpha-1}}{n} + \frac{120}{\alpha}p^\alpha n^{-c_1/8} \right) \\
& \leq \frac{8p^{2\alpha-1}}{n} + \frac{216}{\alpha}p^{2\alpha} n^{-c_1/8} + \frac{8p^{2\alpha-2}}{n^2} \nonumber \\
& \quad  + \frac{28800}{\alpha^2}p^{2\alpha} n^{-c_1/4} + 3p^\alpha \left(2\frac{p^{\alpha-1}}{n} + \frac{120}{\alpha}p^\alpha n^{-c_1/8} \right)\\
& \leq \frac{14p^{2\alpha-1}}{n} + \frac{576}{\alpha}p^{2\alpha} n^{-c_1/8} + \frac{28800}{\alpha^2}p^{2\alpha} n^{-c_1/4} \nonumber \\
& \quad  + \frac{8}{n^{2\alpha}(c_1\ln n)^{2-2\alpha}}.
\end{align}
\end{enumerate}

For $1<\alpha<3/2$, following the same procedures, we obtain some upper bounds on $|r''(x)|$ and $|U_\alpha''(x)|$. Over regime $[0,t]$, $r(x)=U_\alpha^2(x)\equiv0$, we have $r''(x)=U_\alpha''(x)=0$. Over regime $[2t,1]$, since $U_\alpha(x)=f(x)$, we have
\begin{align}
  |r''(x)| &= |2(f')^2+2ff''| \\
  & \le |2(\alpha x^{\alpha-1})^2 + 2x^\alpha\cdot \alpha x^{\alpha-2}| \\
  &\le 8x^{2\alpha-2}\\
  |U_\alpha''(x)| &= |f''(x)| \le \alpha x^{\alpha-2} \le 2x^{\alpha-2}.
\end{align}

Over regime $[t,2t]$, we have
\begin{align}
&  |r''(x)| \nonumber \\
 &= 2|(f')^2I_n^2 + ff''I_n^2 + 2ff'I_nI_n' \nonumber \\
 & \quad + (I_n')^2f^2 + I_nI_n'' f^2 + 2ff'I_nI_n'|\\
  &\le 2\Big(\alpha^2x^{2\alpha-2}+\alpha x^{2\alpha-2} + 2\alpha x^{2\alpha-1} \cdot\frac{4}{t} + \left(\frac{4}{t}\right)^2x^{2\alpha} \nonumber \\
  & \quad 
 + \frac{20}{t^2}\cdot x^{2\alpha}+2\alpha x^{2\alpha-1}\cdot\frac{4}{t}\Big)\\
  &\le 400x^{2\alpha-2},
\end{align}
and
\begin{align}
  |U_\alpha''(x)| & = |I_n''f+2I_n'f'+I_nf''| \\
  & \le \frac{20}{t^2}\cdot x^\alpha + 2\alpha x^\alpha\cdot\frac{4}{t} + \alpha x^{\alpha-2} \\
  & \le 120x^{\alpha-2},
\end{align}
where we have used the inequality
\begin{align}
  |I_n(x)| \le 1, \quad |I_n'(x)|\le \frac{4}{t},\quad |I_n''(x)|\le \frac{20}{t^2},\qquad \forall x\in[t,2t].
\end{align}

Noting that we have obtained a norm bound for $|r''(x)|$ over all regimes expressed as
\begin{align}
  |r''(x)| \le 400x^{2\alpha-2} \le 400, \qquad \forall x\in[0,1],
\end{align}
we have the upper bound
\begin{align}
  |\mathbb{E}R_1(X;p)| & \le \mathbb{E}|R_1(X;p)| \\
  & = \frac{1}{2}\mathbb{E}|f''(\eta_X)(X-p)^2| \\
  & \le 200\mathbb{E}|(X-p)^2| \\
  &= \frac{200p}{n}.
\end{align}
For the upper bound of $|\mathbb{E}R_2(X;p)|$, we first consider the upper bound of $|R_2(x;p)|$ when $x\le p/2$. Denoting $y=\max\{x,\Delta/4\}$, we have
\begin{align}
  R_2(x;p) & = \int_x^p (u-x)U_\alpha''(u)du \\
  & \le \int_y^p (u-x)120u^{\alpha-2}du  \\
  & \le \int_y^p 120u^{\alpha-1}du \\
  & \le \frac{120}{\alpha}p^\alpha,
\end{align}
then
\begin{align}
 & |\mathbb{E}R_2(X;p)| \nonumber \\
  &\le \mathbb{E}|R_2(X;p)|\\
&= \mathbb{E}|R_2(X;p)\mathbbm{1}(X\ge p/2)| + \mathbb{E}|R_2(X;p)\mathbbm{1}(X<p/2)|\\
&\le \mathbb{E}\left[\frac{1}{2}\cdot 2\left(\frac{p}{2}\right)^{\alpha-2}(X-p)^2\right] \nonumber \\
& \quad + \sup_{x< p/2}|R_2(x;p)|\cdot\mathbb{P}\left(nX<\frac{np}{2}\right)\\
&\le \frac{2p^{\alpha-1}}{n} + \frac{120}{\alpha}p^\alpha n^{-c_1/8},
\end{align}
where in the last step we have applied Lemma \ref{lemma.poissontail}. Plugging in the upper bound of $|\mathbb{E}R_1(X;p)|$ and $|\mathbb{E}R_2(X;p)|$ together, we know when $1<\alpha<3/2$,
\begin{align}
 & \mathsf{Var}(U_\alpha(X)) \nonumber \\
  &\le \frac{200p}{n} + \left(\frac{2p^{\alpha-1}}{n} + \frac{120}{\alpha}p^\alpha n^{-c_1/8}\right)^2 \nonumber \\
  & \quad  +
f(p)\left(\frac{2p^{\alpha-1}}{n} + \frac{120}{\alpha}p^\alpha n^{-c_1/8}\right)\\
&\le \frac{200p}{n} + 2\left(\frac{4p^{2(\alpha-1)}}{n^2} + \frac{14400}{\alpha^2}p^{2\alpha} n^{-c_1/4}\right) \nonumber \\
& \quad +
p^\alpha\left(\frac{2p^{\alpha-1}}{n} + \frac{120}{\alpha}p^\alpha n^{-c_1/8}\right)\\
&\le \frac{202p}{n} + \frac{8}{n^2} + \frac{28800}{\alpha^2}p^{2\alpha} n^{-c_1/4} + \frac{120}{\alpha}p^{2\alpha} n^{-c_1/8}.
\end{align}

\subsection{Proof of Lemma~\ref{lemma.entropyupper}}
We have
\begin{equation}
U_H(x) = I_n(x) \left( -x \ln x + \frac{1}{2n} \right) = I_n(x) f(x),
\end{equation}
where $f(x) = -x \ln x + 1/(2n)$.

For $p\geq \Delta$, we do Taylor expansion of $U_H(x)$ around $x = p$. We have
\begin{align}
U_H(x) & = U_H(p) + U_H(p)(x-p) + \frac{1}{2} U''_H(p)(x-p)^2 \nonumber \\
& \quad + \frac{1}{6} U'''_H(p)(x-p)^3 + R(x;p),\label{eqn.utaylorentropy}
\end{align}
where the remainder term enjoys the following representations:
\begin{align}
R(x;p) & = \frac{1}{6}\int_p^x (x-u)^3 U_H^{(4)}(u)du \\
& = \frac{U_H^{(4)}(\xi_x)}{24}(x-p)^4,\quad \xi_x \in [\min\{x,p\}, \max\{x,p\}].
\end{align}
The first remainder is called the integral representation of Taylor series remainders, and the second remainder is called the Lagrange remainder.

Since $p\geq \Delta$, we know that
\begin{align}
U'_H(p) & = - \ln p -1\\
U''_H(p) & = -1/p \\
U_H^{(3)}(p) & = 1/p^2 \\
U_H^{(4)}(p) & = -2/p^3
\end{align}

Replacing $x$ by random variable $X$ in (\ref{eqn.utaylorentropy}), where $nX \sim \mathsf{Poi}(np), p\geq \Delta$, and taking expectations on both sides, we have
\begin{align}
\bE U_H(X) & = U_H(p) + \frac{1}{2} U''_H(p) \frac{p}{n} + \frac{1}{6} U'''_H(p) \frac{p}{n^2} + \bE[R(X;p)]\\
& = -x \ln x + \frac{1}{6pn^2} +  \bE[R(X;p)] \label{eqn.biascorrectbasicentropy}
\end{align}
where we have used the fact that if $nX \sim \mathsf{Poi}(np)$, then $\bE (X-p)^2 = \frac{p}{n}, \bE(X-p)^3 = \frac{p}{n^2}$.

Since the representation of $R(x;p)$ involves $U_H^{(4)}(\xi_x)$, it would be helpful to obtain some estimates of $U_H^{(4)}(x)$ over $[0,1]$. We have
\begin{equation}\label{eqn.ualphataylor4entropy}
U_H^{(4)}(x) = I_n^{(4)} f + 4 I_n^{(3)} f^{(1)} + 6 I_n^{(2)} f^{(2)} + 4 I_n^{(1)} f^{(3)} + I_n f^{(4)}.
\end{equation}

Hence, it suffices to bound each term in (\ref{eqn.ualphataylor4entropy}) separately.

For $x\in [0, t]$, $U_H(x) \equiv 0$, so we do not need to consider this regime. For $x \in [2t,1]$, $U_H(x) = f(x)$, hence
\begin{equation}
|U^{(4)}_H(x)| = |f^{(4)}(x)| = 2/x^3,
\end{equation}
which implies that for $x\geq 2t$,
\begin{equation}
\sup_{z \in [x,1]} |U^{(4)}_\alpha(z)| \leq 2/x^3.
\end{equation}

Finally we consider $x\in (t, 2t)$. Denoting $y = x-t$, the derivatives of $I_n(x)$ for $x\in (t, 2t)$ are as follows:
\begin{align}
I_n'(x) & = \frac{630 y^4(t-y)^4}{t^9}\\
I_n''(x) & = \frac{2520 y^3(t-2y)(t-y)^3 }{t^9}\\
I_n^{(3)}(x) & = \frac{2520 y^2(t-y)^2 (3t^2-14t y + 14 y^2)}{t^9}\\
I_n^{(4)}(x) & = \frac{15120 y(t-2y)(t-y)(t^2-7t y + 7y^2)}{t^9}.
\end{align}

Considering the fact that $y/t \in [0,1]$, we can maximize $|I_n^{(i)}(x)|$ over $x\in (t,2t)$ for $1\leq i \leq 4$. With the help of $\mathsf{Mathematica}$ \cite{mathematica}, we could show that for $x\in (t, 2t)$,
\begin{align}
|I_n'(x)| & \leq  \frac{4}{t}\\
|I_n''(x)| & \leq \frac{20}{t^2} \\
|I_n^{(3)}(x)| & \leq \frac{100}{t^3} \\
|I_n^{(4)}(x)| & \leq \frac{1000}{t^4}.
\end{align}

Plugging these upper bounds in (\ref{eqn.ualphataylor4entropy}), we know for $x\in (t, 2t)$
\begin{align}
& |U_H^{(4)}(x)| \nonumber \\
& \quad \leq \frac{1000}{t^4} t \ln(1/t) + \frac{4\times 100}{t^3}\ln(1/t)+ 6\times\frac{20}{t^2} 1/t \nonumber \\
& \qquad + 4\times \frac{4}{t}\times 1/t^2 + 2/t^3 \\
& \quad \leq 1538\frac{\ln(1/t)}{t^3} \\
& \quad \leq 1538\frac{\ln(2/x)}{(x/2)^3} \\
& \quad \leq 12304 \frac{1+\ln(1/x)}{x^3}.
\end{align}

Now we proceed to upper bound $|\bE [R(X;p)]|, p\geq \Delta$. We consider the following two cases:
\begin{enumerate}
\item Case 1:$x\geq p/2$. In this case,
\begin{align}
|R(x;p)| & = \left|\frac{U_H^{(4)}(\xi_x)}{24}(x-p)^4\right| \\
& \leq \sup_{x\in [p/2,1]}|U_H^{(4)}(x)| \frac{(x-p)^4}{24} \\ & \leq \frac{2}{(p/2)^3}\frac{(x-p)^4}{24} \\
& = \frac{2(x-p)^4}{3p^3}.
\end{align}
\item Case 2: $0\leq x<p/2$. In this case, denoting $y = \max\{x,\Delta/4\}$,
\begin{align}
&|R(x;p)|\nonumber \\
 & \leq  \frac{1}{6}\int_y^p (u-x)^3 |U_H^{(4)}(u)|du \\
& \leq \frac{1}{6} \int_y^p (u-x)^3 12304 \frac{1+\ln(1/u)}{u^3} du \\
& \leq  2051 \int_y^p \frac{(u-x)^3(1+\ln(1/u))}{u^3} du \\
& =  2051 \int_y^p \frac{(u^3 - 3xu^2 + 3x^2 u - x^3)(1+\ln(1/u))}{u^3} du \\
& \leq 2051 \int_y^p \frac{(u^3 + 3x^2 u )(1+\ln(1/u))}{u^3} du \\
& = 2051 \int_y^p \left( 1+ \frac{3x^2}{u^2}\right)\left(1 + \ln(1/u)\right) du \\
& \leq 8204 \int_y^p \left(1 + \ln(1/u)\right) du \\
& \leq 8024 p \left( \ln(1/p) + 2 \right) .
\end{align}
\end{enumerate}

Now we have
\begin{align}
\bE[|R(X;p)|] & = \bE [|R(X;p)|\mathbbm{1}(X\geq p/2)] \nonumber \\
& \quad + \bE[|R(X;p)|\mathbbm{1}(X<p/2)] \\
& \triangleq B_1 + B_2 .
\end{align}

For the term $B_1$, we have
\begin{align}
B_1 & = \bE [|R(X;p)|\mathbbm{1}(X\geq p/2)] \\
& \leq  \bE\left[\frac{2(X-p)^4}{3p^3}\right] \\
&  = \frac{2}{3p^2n^3} + \frac{2}{pn^2},
\end{align}
where we have used the fact that if $nX \sim \mathsf{Poi}(np)$, then $\bE (X-p)^4 = (np + 3n^2p^2)/n^4$.

For the term $B_2$, we have
\begin{align}
B_2 & = \bE[|R(X;p)|\mathbbm{1}(X<p/2)] \\
& \leq 8024 p \left( \ln(1/p) + 2\right) \bP(nX < np/2).
\end{align}

Applying Lemma~\ref{lemma.poissontail}, we have
\begin{align}
B_2 & \leq 8024\left( p \ln(1/p) + 2p\right) e^{-np/8} \\
& = 8024\left( p \ln(1/p) + 2p\right)n^{-c_1/8}.
\end{align}

Hence, we have
\begin{align}
\bE[R(X;p)] & \leq \bE[|R(X;p)|] \\
& \leq  \frac{2}{3p^2n^3} +\frac{2}{pn^2}+ 8024\left( p \ln(1/p) + 2p\right)n^{-c_1/8}.
\end{align}

Plugging this into (\ref{eqn.biascorrectbasicentropy}), we have for $p\geq \Delta$,
\begin{align}
& |\bE U_H(X) + p \ln p | \nonumber \\
 & \leq \frac{1}{6pn^2} +  \frac{2}{3p^2n^3} +\frac{2}{pn^2}+ 8024\left( p \ln(1/p) + 2p\right)n^{-c_1/8} \\
& \leq \frac{3}{pn^2} + \frac{2}{3p^2n^3} +8024\left( p \ln(1/p) + 2p\right)n^{-c_1/8} \\
& \leq \frac{3}{c_1 n \ln n} + \frac{2}{3 (c_1 \ln n)^2 n} + 8024\left( p \ln(1/p) + 2p\right)n^{-c_1/8} .
\end{align}

For the upper bound on the variance $\mathsf{Var}(U_H(X))$, recalling that $f(p) = -x \ln x + \frac{1}{2n}$, for $p\geq \Delta$, we have
\begin{align}
& \mathsf{Var}(U_H(X)) \nonumber \\
& = \bE U_H^2(X) - (\bE U_H(X))^2 \\
& = \bE U_H^2(X) - f^2(p) + f^2(p) - (\bE U_H(X))^2 \\
& \leq | \bE U_H^2(X) - f^2(p)|+ | f^2(p) - (\bE U_H(X)-f(p) + f(p))^2 |\\
& = | \bE U_H^2(X) - f^2(p)|+ |(\bE U_H(X)-f(p))^2 \nonumber \\
& \quad + 2f(p) (\bE U_H(X)-f(p))| \\
& \leq | \bE U_H^2(X) - f^2(p)| + |\bE U_H(X)-f(p)|^2 \nonumber \\
& \quad + 2f(p)|\bE U_H(X)-f(p)|. \label{eqn.varianceestimatelargeentropy}
\end{align}

Hence, it suffices to obtain bounds on $| \bE U_H^2(X) - f^2(p)|$ and $|\bE U_H(X) - f(p)|$. Denoting $r(x) = U^2_H(x)$, we know that $r(x) \in C^4[0,1]$, and it follows from Taylor's formula and the integral representation of the remainder term that
\begin{equation}
r(X) = f^2(p) + r'(p)(X-p) + R_1(X;p),
\end{equation}
\begin{align}
R_1(X;p) & =\int_p^X (X-u)r''(u)du \\
& = \frac{1}{2}r''(\eta_X)(X-p)^2, \nonumber \\
& \qquad \eta_X \in [\min\{X,p\},\max\{X,p\}].
\end{align}

Similarly, we have
\begin{equation}
U_H(X) = f(p) + f'(p)(X-p) + R_2(X;p),
\end{equation}
\begin{align}
R_2(X;p) & = \int_p^X (X-u)U_H''(u)du \\
& = \frac{1}{2} U_H''(\nu_X)(X-p)^2, \nonumber \\
& \qquad \nu_X\in [\min\{X,p\},\max\{X,p\}].
\end{align}

Taking expectation on both sides with respect to $X$, where $nX \sim \mathsf{Poi}(np), p\geq \Delta$, we have
\begin{equation}\label{eqn.u2biaslargeentropy}
|\bE U_H^2(X) - f^2(p)| = |\bE R_1(X;p) |.
\end{equation}
Similarly, we have
\begin{equation}
|\bE U_H(X) - f(p)| =  | \bE R_2(X;p)|.
\end{equation}

As we did for function $U_H(x)$, now we give some upper estimates for $|r''(x)|$ over $[0,1]$. Over regime $[0,t]$, $r(x)\equiv 0$, so we ignore this regime. Over regime $[2t,1]$, since $U_H(x) = f(x)$, we have
\begin{align}
r'(x) & = 2f f'\\
r''(x) & = 2(f')^2 + 2 f f''.
\end{align}
Hence, for $x\geq 2t$,
\begin{equation}
\sup_{z\in [x,1]}|r''(z)| \leq 4 (\ln x)^2.
\end{equation}
\begin{equation}
\sup_{z\in [x,1]}|U_\alpha''(z)| \leq 1/x.
\end{equation}

Over regime $[t,2t]$, we have
\begin{align}
r'(x) & = 2 f f' I_n^2 + 2 I_n I_n' f^2 \\
r''(x) & = 2 \Big( (f')^2 I_n^2 + f f'' I_n^2 + 2 f f' I_n I_n' \nonumber \\
& \quad + (I_n')^2 f^2 + I_n I_n'' f^2 + 2 f f' I_n I_n' \Big).
\end{align}
Hence, we have for $x\in [t,2t]$,
\begin{align}
|r''(x)| & \leq 2 \Big( (\ln t)^2 + \ln(1/t) +2 (t\ln(1/t))(\ln(1/t)) \frac{4}{t} \nonumber \\
& \quad + (4/t)^2 t^2 (\ln t)^2 + \frac{20}{t^2} t^2 (\ln t)^2 \nonumber \\
& \quad + 2 (t \ln(1/t)) \ln(1/t) \frac{4}{t}   \Big) \\
& \leq 108 (\ln t)^2 \\
& \leq 108 (\ln (x/2))^2\\
& = 216 ((\ln x)^2 + 1),
\end{align}
where we have used the fact that $|\ln 2|\approx 0.69 <1$. Also, over regime $[t,2t]$,
\begin{equation}
U_H''(x) = I_n'' f + I_n f'' + 2I_n' f',
\end{equation}
hence for $x\in [t,2t]$,
\begin{align}
|U_H''(x)| & \leq  \frac{20}{t^2} (t \ln(1/t)) + \frac{1}{t} + 2 \ln(1/t) \frac{4}{t} \\
&  \leq \frac{30}{t} \ln(1/t) \\
& \leq \frac{30}{x/2} \ln (2/x) \\
& \leq \frac{60}{x}(\ln(1/x)  +1).
\end{align}

Now we are in the position to bound $|\bE R_1(X;p)|$ and $|\bE R_2(X;p)|$.

We have
\begin{align}
& |\bE R_1(X;p)| \nonumber \\
& \leq \bE |R_1(X;p)| \\
& = \bE [|R_1(X;p)\mathbbm{1}(X\geq p/2)|] + \bE[R_1(X;p)\mathbbm{1}(X<p/2)] \\
& \leq \bE \left[ \frac{1}{2} \times 4(\ln (p/2))^2  (X-p)^2\right] + \bE[R_1(X;p)\mathbbm{1}(X<p/2)] \\
& = 2p(\ln p - \ln 2)^2/n + \sup_{x\leq p/2}|R_1(x;p)|\bP(nX<np/2) \\
& = 2p(\ln p - \ln 2)^2/n  + \sup_{x\leq p/2}|R_1(x;p)| n^{-c_1/8},
\end{align}
where in the last step we have applied Lemma~\ref{lemma.poissontail}.

Regarding $\sup_{x\leq p/2}|R_1(x;p)|$, for any $x\leq p/2$, denoting $y = \max\{x,\Delta/4\}$, we have
\begin{align}
R_1(x;p) & = \int_x^p (u-x)r''(u)du \\
& \leq \int_y^p (u-x) 216 \left( (\ln u)^2 + 1\right) du \\
& \leq 216\int_y^p u\left(  (\ln u)^2 + 1\right) du \\
& \leq 54 p^2 \left( 2 (\ln p)^2 -2 \ln p + 3\right).
\end{align}

Hence, we have
\begin{equation}
|\bE R_1(X;p)| \leq 2p(\ln p - \ln 2)^2/n + 54 p^2 \left| 2 (\ln p)^2 -2 \ln p + 3\right| n^{-c_1/8}.
\end{equation}

Analogously, we obtain the following bound for $|\bE R_2(X;p)|$:
\begin{equation}
|\bE R_2(X;p)|\leq \frac{1}{n} + 60 \left( p\ln(1/p) + 2p\right)n^{-c_1/8}.
\end{equation}

Plugging these estimates of $|\bE R_1(X;p)|$ and $|\bE R_2(X;p)|$ into (\ref{eqn.varianceestimatelargeentropy}), we have for $p\geq \Delta$,
\begin{align}
& \mathsf{Var}(U_H(X)) \nonumber \\
 &  \leq  2p(\ln p - \ln 2)^2/n + 54 p^2 \left| 2 (\ln p)^2 -2 \ln p + 3\right|n^{-c_1/8} \nonumber \\
 & \quad + \left( \frac{1}{n} + 60 \left( p\ln(1/p) + 2p\right)n^{-c_1/8}\right)^2 \\
& \quad + 2 \left( p\ln(1/p) + \frac{1}{2n}\right) \left( \frac{1}{n} + 60 \left( p\ln(1/p) + 2p\right)n^{-c_1/8}\right).
\end{align}

\subsection{Proof of Lemma~\ref{lemma.lowpartbiasvariance}}

We first bound the bias term. It follows from differentiating the moment generating function of the Poisson distribution that if $X \sim \mathsf{Poi}(\lambda)$, then
\begin{equation}
\bE X(X-1)\ldots(X-r+1) = \lambda^r,
\end{equation}
for any $r$ positive integer.

Then, we know that for $nX \sim \spo(np)$,
\begin{equation}
\bE S_{K,\alpha}(X) = \sum_{k =1 }^{K} g_{k,\alpha} (4\Delta)^{-k + \alpha} p^k.
\end{equation}

Applying Lemma~\ref{lemma.approsmall}, we know that for all $p\leq 4\Delta$,
\begin{equation}
|\bE S_{K,\alpha}(X) - p^\alpha| \leq \frac{c_3}{(n \ln n)^\alpha}.
\end{equation}

Now we bound the second moment of $S_{K,\alpha}(X)$. Denote
\begin{equation}
E_{k,n}(x) = \prod_{r = 0}^{k-1} (x-r/n),
\end{equation}
we have
\begin{align}
\bE S_{K,\alpha}^2(X) & \leq \left( \sum_{k = 1}^{K} |g_{k,\alpha}| (4\Delta)^{-k + \alpha}  \left(  \bE E_{k,n}^2(X) \right)^{1/2} \right)^2 \\
& \leq 2^{6K} \left( \sum_{k = 1}^{K} (4\Delta)^{-k + \alpha}  \left(  \bE E_{k,n}^2(X) \right)^{1/2} \right)^2
\end{align}
Here we have used Lemma~\ref{lemma.nonasympxa}.

Since $K \leq 4 n \Delta$, applying Lemma~\ref{lemma.poissonmoment},
\begin{align}
\bE E_{k,n}^2(X) & = \frac{1}{n^{2k}} \bE \prod_{r=0}^{k-1} (nX - r)^2 \\
& \leq \frac{1}{n^{2k}} \bE \prod_{r=0}^{k-1} (nX)^2 \\
& = \frac{1}{n^{2k}} \bE (nX)^{2k} \\
& \leq \frac{1}{n^{2k}} (8 c_1 \ln n)^{2k} \\
& = \left( \frac{8 c_1 \ln n}{n} \right)^{2k},
\end{align}
we know
\begin{align}
\bE S_{K,\alpha}^2(X) & \leq 2^{6K}\left( \sum_{k = 1}^{K} (4\Delta)^{-k + \alpha} \left( \frac{8 c_1 \ln n}{n} \right)^{k}   \right)^2 \\
& \leq 2^{6K} 2^{2K} \left(\sum_{k = 1}^{K} (4\Delta)^{-k+\alpha} (4\Delta)^{k} \right)^{2} \\
& \leq 2^{6K} 2^{2K} \left( \sum_{k = 1}^{K}(4\Delta)^{\alpha} \right)^{2} \\
& = 2^{8K} K^2 (4\Delta)^{2\alpha} \\
& \leq n^{8c_2 \ln 2} (c_2 \ln n)^2 (4c_1 \ln n /n)^{2\alpha}\\
& \leq n^{8c_2 \ln 2} \frac{(4c_1 \ln n)^{2+2\alpha}}{n^{2\alpha}}
\end{align}

The proof for the $S_{K,H}(x)$ case is essentially the same as that for $S_{K,\alpha}(x)$ via replacing $\alpha$ by $1$ and applying Lemma~\ref{lemma.approsmallentropy} rather than Lemma~\ref{lemma.approsmall}.

\subsection{Proof of Lemma \ref{lem_L_small}}
We first bound the bias term. It follows from the property of Poisson distribution that
\begin{align}\label{eq:appro_mean}
  \mathbb{E}S_{K,\alpha}(X) = \sum_{k=1}^K g_{k,\alpha}(4\Delta)^{-k+\alpha}p^k.
\end{align}
It follows from a variation of the pointwise bound in Lemma \ref{lemma.nonasympxa} that
\begin{align}
  |\mathbb{E}S_{K,\alpha}(X)-p^\alpha| \le \frac{D_1(4\Delta)^\alpha}{K^{2(\alpha-1)}}\cdot\frac{p}{4\Delta} = D_1\left(\frac{4c_1}{c_2^2n\ln n}\right)^{\alpha-1}p,
\end{align}
which completes the proof of the first part of Lemma \ref{lem_L_small}. For the variance, denote
\begin{align}
  E_{k,n}(x) = \prod_{r=0}^{k-1}\left(x-\frac{r}{n}\right),
\end{align}
we have
\begin{align}
  \mathbb{E}E_{k,n}^2(X) &= \frac{1}{n^{2k}}\mathbb{E}\prod_{r=0}^{k-1}(nX-r)^2 \le \frac{1}{n^{2k}}\mathbb{E}[nX]^{2k}\\
&= \frac{1}{n^{2k}}\sum_{i=1}^{2k} \left\{\begin{matrix}
                                         2k \\
                                         i
                                       \end{matrix}
\right\}(np)^{i}\\
&\le \frac{1}{n^{2k}}\sum_{i=1}^{2k} \binom{2k}{i}i^{2k-i}(np)^i\\
&\le \frac{1}{n^{2k}}\sum_{i=1}^{2k} \binom{2k}{i}(2k)^{2k-i}np\\
&\le \frac{(2k+1)^{2k}p}{n^{2k-1}},
\end{align}
where $\left\{\begin{matrix}
k \\
i
\end{matrix}
\right\}$ is the Stirling numbers of the second kind, and we have used the inequality \cite{Rennie1969stirling}
\begin{align}\label{eq:stirling_inequality}
  \left\{\begin{matrix}
k \\
i
\end{matrix}
\right\} \le \binom{k}{i}i^{k-i}.
\end{align}

Hence, we can bound the second moment of $S_{K,\alpha}(X)$ as
\begin{align}
  \mathbb{E}S_{K,\alpha}^2(X) &\le \left(\sum_{k=1}^K |g_{k,\alpha}|(4\Delta)^{-k+\alpha}(\mathbb{E}E_{k,n}^2(X))^{1/2}\right)^2\\
&\le 2^{6K}\left(\sum_{k=1}^K (4\Delta)^{-k+\alpha}\frac{(2k+1)^k\sqrt{p}}{n^{k-\frac{1}{2}}}\right)^2\\
&\le 2^{10K}\left(\sum_{k=1}^K \left(\frac{4c_1\ln n}{n}\right)^{-k+\alpha}\frac{K^k\sqrt{p}}{n^{k-\frac{1}{2}}}\right)^2\\
&\le \frac{2^{10K}(4c_1\ln n)^{2\alpha}}{n^{2\alpha-1}}\left(\sum_{k=1}^K(\frac{c_2}{4c_1})^k\right)^2p\\
&\le 2^{10c_2\ln 2}\frac{(4c_1\ln n)^{2\alpha}K^2}{n^{2\alpha-1}}p\\
&\le 2^{10c_2\ln 2}\frac{(4c_1\ln n)^{2\alpha+2}p}{n^{2\alpha-1}},
\end{align}
given $c_2<4c_1$, where we have used Lemma \ref{lemma.nonasympxa}.


\subsection{Proof of Lemma~\ref{lemma.preparemainf}}
We apply Lemma~\ref{lemma.cailow1} and Lemma~\ref{lemma.cailow2} to calculate the bias and variance of $\xi$.
\begin{enumerate}

\item Case 1: $p\leq \Delta$

\subsubsection*{Claim} when $p\leq \Delta$, we have
\begin{align}
|B(\xi)| & \lesssim \frac{1}{(n \ln n)^\alpha} \\
\mathsf{Var}(\xi) & \lesssim \frac{(\ln n)^{2+2\alpha}}{n^{2\alpha - \epsilon}}.
\end{align}
Now we prove this claim. In this regime, we write $L_\alpha(X) = S_{K,\alpha} - (S_{K,\alpha}(X)-1) \mathbbm{1}(S_{K,\alpha}(X)\geq 1)$. We have
\begin{align}
& |B(\xi)| \nonumber \\
 & = \Bigg | \bE S_{K,\alpha}(X) \bP(Y \leq 2\Delta) \nonumber \\
 & \quad - \left[ \bE (S_{K,\alpha}(X)-1)\mathbbm{1}(S_{K,\alpha}\geq 1) \right] \bP(Y\leq 2\Delta) \nonumber \\
 & \quad + \bE U_\alpha(X) \bP(Y>2\Delta) - p^\alpha \Bigg |\\
& = \Bigg | \bE S_{K,\alpha}(X) - p^\alpha \nonumber \\
& \quad - \left[ \bE (S_{K,\alpha}(X)-1)\mathbbm{1}(S_{K,\alpha}\geq 1) \right] \bP(Y\leq 2\Delta) \nonumber \\
&\quad + \left( \bE U_\alpha(X) - \bE S_{K,\alpha}(X)\right) \bP(Y>2\Delta) \Bigg |\\
& \leq |\bE S_{K,\alpha}(X) - p^\alpha| \nonumber \\
& \quad+ \bE (S_{K,\alpha}(X)-1)\mathbbm{1}(S_{K,\alpha}\geq 1) \nonumber \\
& \quad + \left( |\bE U_\alpha(X)| + |\bE S_{K,\alpha}(X)|\right) \bP(Y>2\Delta)  \\
& \equiv B_1  + B_2 + B_3.
\end{align}

Now we bound $B_1,B_2,B_3$ separately. It follows from Lemma~\ref{lemma.lowpartbiasvariance} that
\begin{equation}
B_1 = |\bE S_{K,\alpha}(X) - p^\alpha| \leq \frac{c_3}{(n\ln n)^\alpha} \lesssim \frac{1}{(n \ln n)^\alpha}.
\end{equation}

Now consider $B_2$. Note that for any random variable $Z$ and any constant $\lambda>0$,
\begin{equation}
\bE (Z \mathbbm{1}(Z\geq \lambda)) \leq \lambda^{-1} \bE (Z^2 \mathbbm{1}(X\geq \lambda)) \leq \lambda^{-1} \bE Z^2.
\end{equation}
Hence, we have
\begin{align}
B_2 & = \bE (S_{K,\alpha}(X)-1)\mathbbm{1}(S_{K,\alpha}\geq 1) \\
& \leq \bE S_{K,\alpha}\mathbbm{1}(S_{K,\alpha}\geq 1)\\
& \leq \bE S_{K,\alpha}^2 \\
& \leq n^{8c_2 \ln 2} \frac{(4c_1 \ln n)^{2+2\alpha}}{n^{2\alpha}} \\
&  \lesssim \frac{(\ln n)^{2+2\alpha}}{n^{2\alpha - \epsilon}},\label{eqn.b2method}
\end{align}
where we used Lemma~\ref{lemma.lowpartbiasvariance} in the last step.

Now we deal with $B_3$. We have
\begin{align}
|\bE S_{K,\alpha}(X)| & \leq p^\alpha + \frac{c_3}{(n\ln n)^\alpha} \\
& \leq \left(\frac{c_1\ln n}{n} \right)^{\alpha} + \frac{c_3}{(n\ln n)^\alpha} \\
\bE |U_\alpha(X)| & \leq \sup_{x\in [0,1]} |U_\alpha(x)| \leq  1+ \frac{\alpha(1-\alpha)}{2c_1  \ln n} \\
& \leq 1+ \frac{1}{8c_1  \ln n} \\
\bP(Y \geq 2\Delta) & = \bP(nY \geq 2n\Delta) \\
&  \leq (e/4)^{c_1 \ln n} \\
 & = n^{-c_1 \ln(4/e)},
\end{align}
where we have used Lemma~\ref{lemma.lowpartbiasvariance} and Lemma~\ref{lemma.poissontail}. Thus, we have
\begin{align}
B_3 & = \left( |\bE S_{K,\alpha}(X)|  + \bE |U_\alpha(X)|\right)\bP(Y \geq 2\Delta) \\
&  \lesssim n^{-c_1 \ln(4/e)}.
\end{align}

To sum up, we have the following bound on $|B(\xi)|$:
\begin{equation}
|B(\xi)|\lesssim \frac{1}{(n \ln n)^\alpha} +\frac{(\ln n)^{2+2\alpha}}{n^{2\alpha - \epsilon}} + n^{-16\alpha} \lesssim \frac{1}{(n \ln n)^\alpha}.
\end{equation}

We now consider the variance. It follows from Lemma~\ref{lemma.cailow1} and Lemma~\ref{lemma.cailow2} that
\begin{align}
& \mathsf{Var}(\xi) \nonumber \\
 & \leq \mathsf{Var}(S_{K,\alpha}(X)) + \mathsf{Var}(U_\alpha(X)) \bP(Y>2\Delta) \nonumber \\
 & \quad + (\bE L_\alpha(X) - \bE U_\alpha(X))^2 \bP(Y > 2\Delta) \\
& \leq \bE S_{K,\alpha}^2(X) + \left( \bE U_\alpha^2(X) + 1 + 2 |\bE U_\alpha(X)| \right) \bP(Y > 2\Delta) \\
&  \leq n^{8c_2 \ln 2} \frac{(4c_1 \ln n)^{2+2\alpha}}{n^{2\alpha}} + \left( 1 + \frac{1}{8c_1  \ln n}\right)^2 n^{-c_1 \ln (4/e)} \\
& \lesssim \frac{(\ln n)^{2+2\alpha}}{n^{2\alpha - \epsilon}} + n^{-8\alpha} \\
& \lesssim \frac{(\ln n)^{2+2\alpha}}{n^{2\alpha - \epsilon}}
\end{align}

\item Case 2: $\Delta < p \leq 4\Delta$.
\subsubsection*{Claim} when $\Delta < p \leq 4\Delta$, we have
\begin{align}
|B(\xi)| & \lesssim \frac{1}{(n \ln n)^\alpha} \\
\mathsf{Var}(\xi) & \lesssim  \begin{cases} \frac{(\ln n)^{2+2\alpha}}{n^{2\alpha - \epsilon}}  & 0<\alpha\leq 1/2 \\ \frac{(\ln n)^{2+2\alpha}}{n^{2\alpha - \epsilon}} + \frac{p^{2\alpha-1}}{n} & 1/2 < \alpha <1  \end{cases}
\end{align}
Now we prove this claim. In this case,
\begin{align}
& |B(\xi)| \nonumber \\
 & = \Bigg | (\bE L_\alpha(X) - p^\alpha) \bP(Y \leq 2\Delta)  \nonumber \\
 & \quad + (\bE U_\alpha(X) - p^\alpha) \bP(Y>2\Delta) \Bigg |\\
& \leq |\bE L_\alpha(X) - p^\alpha| + |\bE U_\alpha(X) - p^\alpha| \\
& \leq |\bE S_{K,\alpha}(X) - p^\alpha| + \bE (S_{K,\alpha}(X)-1)\mathbbm{1}(S_{K,\alpha}\geq 1) \nonumber \\
& \quad + |\bE U_\alpha(X) - p^\alpha| \\
& \equiv B_1 + B_2 + B_3.
\end{align}

It follows from Lemma~\ref{lemma.lowpartbiasvariance} that
\begin{equation}
B_1 = |\bE S_{K,\alpha}(X) - p^\alpha| \leq \frac{c_3}{(n \ln n)^\alpha} \lesssim \frac{1}{(n \ln n)^\alpha}.
\end{equation}

As in (\ref{eqn.b2method}), we have
\begin{align}
B_2 & = \bE (S_{K,\alpha}(X)-1)\mathbbm{1}(S_{K,\alpha}\geq 1) \\
& \leq n^{8c_2 \ln 2} \frac{(4c_1 \ln n)^{2+2\alpha}}{n^{2\alpha}} \\
& \lesssim \frac{(\ln n)^{2+2\alpha}}{n^{2\alpha - \epsilon}}.
\end{align}

Regarding $B_3$, applying Lemma~\ref{lemma.largebias}, we have
\begin{align}
B_3 & \leq \frac{17}{n^\alpha(c_1 \ln n)^{2-\alpha}} + \frac{8310 (1+\alpha)}{\alpha(2-\alpha)}p^\alpha n^{-c_1/8} \\
& \lesssim \frac{1}{n^\alpha (\ln n)^{2-\alpha}} + n^{-2\alpha }.
\end{align}

To sum up, we have
\begin{align}
|B(\xi)| & \lesssim \frac{1}{(n \ln n)^\alpha} +  \frac{(\ln n)^{2+2\alpha}}{n^{2\alpha - \epsilon}}+  \frac{1}{n^\alpha (\ln n)^{2-\alpha}} + n^{-2\alpha} \\
& \lesssim  \frac{1}{(n \ln n)^\alpha} .
\end{align}

For the variance, we have
\begin{align}
\mathsf{Var}(\xi) & \leq \mathsf{Var}(S_{K,\alpha}(X)) + \mathsf{Var}(U_\alpha(X)) \nonumber \\
& \quad + \left( \bE L_\alpha(X) -\bE U_\alpha(X) \right)^2.
\end{align}

Applying Lemma~\ref{lemma.lowpartbiasvariance}, we have
\begin{equation}
\mathsf{Var}(S_{K,\alpha}(X)) \leq n^{8c_2 \ln 2} \frac{(4c_1 \ln n)^{2+2\alpha}}{n^{2\alpha}} \lesssim \frac{(\ln n)^{2+2\alpha}}{n^{2\alpha - \epsilon}}.
\end{equation}
Lemma~\ref{lemma.largebias} implies that when $0<\alpha \leq 1/2$, 
\begin{align}
& \mathsf{Var}(U_\alpha(X)) \nonumber \\
& \leq \frac{24}{n^{2\alpha}(c_1\ln n)^{1-2\alpha}}  + \frac{576}{\alpha}p^{2\alpha} n^{-c_1/8} + \frac{28800}{\alpha^2}p^{2\alpha} n^{-c_1/4},
\end{align}
and when $1/2<\alpha<1$,
\begin{align}
& \mathsf{Var}(U_\alpha(X)) \nonumber \\
& \leq  \frac{14p^{2\alpha-1}}{n} + \frac{576}{\alpha}p^{2\alpha} n^{-c_1/8} + \frac{28800}{\alpha^2}p^{2\alpha} n^{-c_1/4} \nonumber \\
& \quad  + \frac{8}{n^{2\alpha}(c_1\ln n)^{2-2\alpha}} 
\end{align}
Regarding $\left( \bE L_\alpha(X) -\bE U_\alpha(X) \right)^2$, we have
\begin{align}
& \left( \bE L_\alpha(X) -\bE U_\alpha(X) \right)^2 \nonumber \\
& \leq \Big[ |\bE S_{K,\alpha}(X)-p^\alpha| + \bE (S_{K,\alpha}(X)-1)\mathbbm{1}(S_{K,\alpha}\geq 1) \nonumber \\
& \quad + |\bE U_\alpha(X) - p^\alpha| \Big]^2 \\
& \leq \Big[ \frac{c_3}{(n \ln n)^\alpha} +   n^{8c_2 \ln 2} \frac{(4c_1 \ln n)^{2+2\alpha}}{n^{2\alpha}} + \frac{17}{n^\alpha(c_1 \ln n)^{2-\alpha}} \nonumber \\
& \quad + \frac{8310 (1+\alpha)}{\alpha(2-\alpha)}p^\alpha n^{-c_1/8} \Big]^2\\
& \lesssim \frac{1}{(n\ln n)^{2\alpha}}.
\end{align}

To sum up, we have
\begin{equation}
\mathsf{Var}(\xi) \lesssim \begin{cases} \frac{(\ln n)^{2+2\alpha}}{n^{2\alpha - \epsilon}}  & 0<\alpha\leq 1/2 \\ \frac{(\ln n)^{2+2\alpha}}{n^{2\alpha - \epsilon}} + \frac{p^{2\alpha-1}}{n} & 1/2 < \alpha <1  \end{cases}
\end{equation}

\item Case 3: $p>4\Delta$.
\subsubsection*{Claim} when $p>4\Delta$, we have
\begin{align}
|B(\xi)| & \lesssim  \frac{1}{n^\alpha(\ln n)^{2-\alpha}}\\
\mathsf{Var}(\xi) & \lesssim  \begin{cases} \frac{1}{n^{2\alpha}(\ln n)^{1-2\alpha}} & 0<\alpha\leq 1/2 \\ \frac{1}{n^{2\alpha}(\ln n)^{1-2\alpha}} + \frac{p^{2\alpha-1}}{n} & 1/2 < \alpha <1 \end{cases}
\end{align}
Now we prove this claim. In this case,
\begin{align}
& |B(\xi)| \nonumber \\
& \leq |\bE U_\alpha(X)-p^\alpha| + \left(|\bE L_\alpha(X)|+p^\alpha \right) \bP(Y \leq 2\Delta) \\
& \leq \frac{17}{n^\alpha(c_1 \ln n)^{2-\alpha}} + \frac{8310 (1+\alpha)}{\alpha(2-\alpha)}p^\alpha n^{-c_1/8} \nonumber \\
& \quad  + 2 \bP(Y \leq 2\Delta) \\
& \leq \frac{17}{n^\alpha(c_1 \ln n)^{2-\alpha}} + \frac{8310 (1+\alpha)}{\alpha(2-\alpha)}p^\alpha n^{-c_1/8} \nonumber \\
& \quad  + 2 \bP(nY \leq 2n\Delta) \\
& \leq \frac{17}{n^\alpha(c_1 \ln n)^{2-\alpha}} + \frac{8310 (1+\alpha)}{\alpha(2-\alpha)}p^\alpha n^{-c_1/8}\nonumber \\
& \quad  + 2 e^{-c_1/2 \ln n} \\
& = \frac{17}{n^\alpha(c_1 \ln n)^{2-\alpha}} + \frac{8310 (1+\alpha)}{\alpha(2-\alpha)}p^\alpha n^{-c_1/8} + 2 n^{-c_1/2}\\
& \lesssim \frac{1}{n^\alpha(\ln n)^{2-\alpha}}.
\end{align}

Regarding the variance, we have
\begin{align}
& \mathsf{Var}(\xi) \nonumber \\
& \leq \mathsf{Var}(U_\alpha(X)) \nonumber \\
& \quad + \left( \mathsf{Var}(L_\alpha(X)) +(\bE L_\alpha(X) - \bE U_\alpha(X))^2 \right) \bP(Y \leq 2\Delta). 
\end{align}

When $0<\alpha\leq 1/2$, we have
\begin{align}
& \mathsf{Var}(\xi) \nonumber \\
& \leq \frac{24}{n^{2\alpha}(c_1\ln n)^{1-2\alpha}}  + \frac{576}{\alpha}p^{2\alpha} n^{-c_1/8} \nonumber \\
& \quad + \frac{28800}{\alpha^2}p^{2\alpha} n^{-c_1/4}+ 3 \bP(Y \leq 2\Delta) \\
& \leq \frac{24}{n^{2\alpha}(c_1\ln n)^{1-2\alpha}}  + \frac{576}{\alpha}p^{2\alpha}n^{-c_1/8} \nonumber \\
& \quad + \frac{28800}{\alpha^2} p^{2\alpha}n^{-c_1/4} + 3 n^{-c_1/2} \\
& \lesssim \frac{1}{n^{2\alpha}(\ln n)^{1-2\alpha}}. 
\end{align}

When $1/2<\alpha<1$, we have
\begin{align}
& \mathsf{Var}(\xi) \nonumber \\
& \leq \frac{14p^{2\alpha-1}}{n} + \frac{576}{\alpha}p^{2\alpha}n^{-c_1/8} + \frac{28800}{\alpha^2}p^{2\alpha}n^{-c_1/4} \nonumber \\
& \quad + \frac{8}{n^{2\alpha}(c_1\ln n)^{2-2\alpha}} + 3 \bP(Y \leq 2\Delta)\\
& \leq \frac{14p^{2\alpha-1}}{n} + \frac{576}{\alpha}p^{2\alpha}n^{-c_1/8} + \frac{28800}{\alpha^2}p^{2\alpha}n^{-c_1/4} \nonumber \\
& \quad + \frac{8}{n^{2\alpha}(c_1\ln n)^{2-2\alpha}}  + 3 n^{-c_1/2}\\
& \lesssim \frac{1}{n^{2\alpha}(\ln n)^{1-2\alpha}} + \frac{p^{2\alpha-1}}{n}. 
\end{align}
\end{enumerate}

\subsection{Proof of Lemma \ref{lem_indivi_bound}}
We use Lemma \ref{lemma.largebias}, Lemma \ref{lemma.lowpartbiasvariance} and Lemma \ref{lem_L_small} to compute the bias and variance of $\xi$. We distinguish four cases.

\begin{enumerate}

\item Case 1: $p\le \frac{1}{n\ln n}$.

In this regime, we write $L_\alpha(X)=S_{K,\alpha}(X)-(S_{K,\alpha}(X)-1)\mathbbm{1}(S_{K,\alpha}(X)\ge1)$ and bound the bias as
\begin{align}
  &|B(\xi)| \nonumber \\
  &= \Big|\mathbb{E}S_{K,\alpha}(X)\mathbb{P}(Y\le 2\Delta) \nonumber \\
  & \quad -[\mathbb{E}(S_{K,\alpha}(X)-1)\mathbbm{1}(S_{K,\alpha}(X)\ge1)]\mathbb{P}(Y\le 2\Delta) \nonumber \\
  & \quad +\mathbb{E}U_\alpha(X)\mathbb{P}(Y> 2\Delta)-p^\alpha\Big|\\
&\le |\mathbb{E}S_{K,\alpha}(X)-p^\alpha| \nonumber \\
& \quad + \mathbb{E}(S_{K,\alpha}(X)-1)\mathbbm{1}(S_{K,\alpha}(X)\ge1) \nonumber \\
& \quad  + (|\mathbb{E}U_\alpha(X)|+|\mathbb{E}S_{K,\alpha}(X)|)\mathbb{P}(Y> 2\Delta)\\
&\equiv B_1+B_2+B_3.
\end{align}
Now we bound $B_1,B_2,B_3$ separately. It follows from Lemma \ref{lem_L_small} that
\begin{align}
  B_1 &= |\mathbb{E}S_{K,\alpha}(X)-p^\alpha| \\
  & \le D_1\left(\frac{4c_1}{c_2^2n\ln n}\right)^{\alpha-1}p \\
  & \lesssim \frac{p}{(n\ln n)^{\alpha-1}},\\
  B_2 &= \mathbb{E}(S_{K,\alpha}(X)-1)\mathbbm{1}(S_{K,\alpha}(X)\ge1) \\
  & \le \mathbb{E}S_{K,\alpha}^2 \\
  & \le n^{10c_2\ln 2}\frac{(4c_1\ln n)^{2\alpha+2}p}{n^{2\alpha-1}}\\
  & \lesssim \frac{(\ln n)^{2\alpha+2}p}{n^{2\alpha-1-\epsilon}}.
\end{align}
Now we consider $B_3$. Since
\begin{align}
  |\mathbb{E}U_\alpha(X)| &\le |\mathbb{E}X^\alpha| \\
  &  \le \frac{1}{n^\alpha}|\mathbb{E}(nX)^\alpha| \\
  & \le  \frac{1}{n^\alpha}|\mathbb{E}(nX)^2| \\
  & = \frac{np+(np)^2}{n^{\alpha}}\lesssim \frac{p}{n^{\alpha-1}}\\
  |\mathbb{E}S_{K,\alpha}(X)| &\le |\mathbb{E}S_{K,\alpha}(X)-p^\alpha| + p^\alpha \\
  & \le D_1\left(\frac{4c_1}{c_2^2n\ln n}\right)^{\alpha-1}p + p^\alpha \\
  & \lesssim \frac{p}{(n\ln n)^{\alpha-1}}\\
  \mathbb{P}(Y> 2\Delta) &=  \mathbb{P}(nY> 2n\Delta) \\
  & \le \left(\frac{enp}{2c_1\ln n}\right)^{2c_1\ln n}\\
  & \le \left(\frac{e}{2c_1(\ln n)^2}\right)^{2c_1\ln n}\\
  & \lesssim n^{-4c_1\ln\ln n},
\end{align}
where we have used Lemma \ref{lemma.poissontail} and the pointwise bound in Lemma \ref{lemma.nonasympxa}, thus
\begin{align}
  B_3 & = (|\mathbb{E}U_\alpha(X)|+|\mathbb{E}S_{K,\alpha}(X)|)\mathbb{P}(Y> 2\Delta) \\
  & \lesssim pn^{-4c_1\ln\ln n}.
\end{align}

To sum up, we have the following bound on $|B(\xi)|$:
\begin{align}
  |B(\xi)| & \lesssim \frac{p}{(n\ln n)^{\alpha-1}} + \frac{(\ln n)^{2\alpha+2}p}{n^{2\alpha-1-\epsilon}} + pn^{-4c_1\ln\ln n}\\
  & \lesssim  \frac{p}{(n\ln n)^{\alpha-1}}
\end{align}
As for variance, it follows from Lemma \ref{lemma.cailow1} and Lemma \ref{lemma.cailow2} that
\begin{align}
&  \mathsf{Var}(\xi) \nonumber \\
 &\le \mathsf{Var}(S_{K,\alpha}(X)) + \mathsf{Var}(U_\alpha)\mathbb{P}(Y>2\Delta) \nonumber \\
 & \quad + (\mathbb{E}L_\alpha(X)-\mathbb{E}U_\alpha(X))^2\mathbb{P}(Y>2\Delta)\\
 &\le \mathbb{E}S_{K,\alpha}^2(X) + 2(\mathbb{E}L_{\alpha}^2(X)+\mathbb{E}U_\alpha^2(X))\mathbb{P}(Y>2\Delta)\\
 &\lesssim \frac{(\ln n)^{2\alpha+2}p}{n^{2\alpha-1-\epsilon}} + 2\left(\frac{(\ln n)^{2\alpha+2}p}{n^{2\alpha-1-\epsilon}}+\frac{p}{n^{2\alpha-1}}\right)n^{-4c_1\ln\ln n}\\
 &\lesssim \frac{(\ln n)^{2\alpha+2}p}{n^{2\alpha-1-\epsilon}},
\end{align}
where we have used the fact that when $p\le \frac{1}{n\ln n}$,
\begin{align}
  |\mathbb{E}U_\alpha^2(X)| & \le |\mathbb{E}X^{2\alpha}|\\
  & \le \frac{1}{n^{2\alpha}}|\mathbb{E}(nX)^{2\alpha}|\\
  & \le  \frac{1}{n^{2\alpha}}|\mathbb{E}(nX)^3|\\
  & = \frac{(np)^3+3(np)^2+np}{n^{2\alpha}}\\
  & \lesssim \frac{p}{n^{2\alpha-1}}.
\end{align}

\item Case 2:$\frac{1}{n\ln n}<p\le \Delta$.

To bound the bias, we use the same definition of $B_1,B_2,B_3$ as in Case 1. It follows from Lemma \ref{lemma.lowpartbiasvariance} that
\begin{align}
  B_1 &= |\mathbb{E}S_{K,\alpha}(X)-p^\alpha| \le \frac{c_3}{(n\ln n)^\alpha} \lesssim \frac{1}{(n\ln n)^\alpha},\\
  B_2 &= \mathbb{E}(S_{K,\alpha}(X)-1)\mathbbm{1}(S_{K,\alpha}(X)\ge1) \le \mathbb{E}S_{K,\alpha}^2 \\
  & \le n^{8c_2\ln 2}\frac{(4c_1\ln n)^{2\alpha+2}}{n^{2\alpha}}\\
  & \lesssim \frac{(\ln n)^{2\alpha+2}}{n^{2\alpha-\epsilon}}.
\end{align}
To deal with $B_3$, we use
\begin{align}
  |\mathbb{E}U_\alpha(X)| &\le \sup_{x\in[0,1]} |U_\alpha(x)| \le 1\\
  |\mathbb{E}S_{K,\alpha}(X)| &\le |\mathbb{E}S_{K,\alpha}(X)-p^\alpha| + p^\alpha \\
  & \le  \frac{c_3}{(n\ln n)^\alpha} + \left(\frac{c_1\ln n}{n}\right)^\alpha \\
  & \lesssim \left(\frac{\ln n}{n}\right)^\alpha\\
   \mathbb{P}(Y> 2\Delta) &=  \mathbb{P}(nY> 2n\Delta) \\
   & \le \left(\frac{e}{4}\right)^{c_1\ln n} \\
   & \lesssim n^{-c_1\ln(4/e)},
\end{align}
to obtain that
\begin{align}
  B_3 &= (|\mathbb{E}U_\alpha(X)|+|\mathbb{E}S_{K,\alpha}(X)|)\mathbb{P}(Y> 2\Delta)\\
  &  \lesssim n^{-c_1\ln(4/e)}.
\end{align}
Hence, the bias is upper bounded by
\begin{align}
  |B(\xi)| & \lesssim \frac{1}{(n\ln n)^\alpha} + \frac{(\ln n)^{2\alpha+2}}{n^{2\alpha-\epsilon}} + n^{-c_1\ln(4/e)}\\
  &  \lesssim \frac{1}{(n\ln n)^\alpha}.
\end{align}

Similar to the analysis in Case 1, the variance is upper bounded by
\begin{align}
 & \mathsf{Var}(\xi) \nonumber \\
  &\le \mathsf{Var}(S_{K,\alpha}(X)) + \mathsf{Var}(U_\alpha(X))\mathbb{P}(Y>2\Delta) \nonumber \\
  & \quad + (\mathbb{E}L_\alpha(X)-\mathbb{E}U_\alpha(X))^2\mathbb{P}(Y>2\Delta)\\
 &\le \mathbb{E}S_{K,\alpha}^2(X) + 2(\mathbb{E}L_{\alpha}^2(X)+\mathbb{E}U_\alpha^2(X))\mathbb{P}(Y>2\Delta)\\
 &\lesssim \frac{(\ln n)^{2\alpha+2}}{n^{2\alpha-\epsilon}} + 2\left(1+1\right)n^{-c_1\ln(4/e)}\\
 &\lesssim \frac{(\ln n)^{2\alpha+2}}{n^{2\alpha-\epsilon}}.
\end{align}

\item Case 3: $\Delta< p< 4\Delta$.

In this case,
\begin{align}
  |B(\xi)| &= \Big|(\mathbb{E}L_\alpha(X)-p^\alpha)\mathbb{P}(Y\le 2\Delta)\nonumber \\
  & \quad +(\mathbb{E}U_\alpha(X)-p^\alpha)\mathbb{P}(Y> 2\Delta)\Big|\\
&\le |\mathbb{E}L_\alpha(X)-p^\alpha| + |\mathbb{E}U_\alpha(X)-p^\alpha|\\
&\le |\mathbb{E}S_{K,\alpha}(X)-p^\alpha| \nonumber \\
& \quad + \mathbb{E}(S_{K,\alpha}(X)-1)\mathbbm{1}(S_{K,\alpha}(X)\ge1) \nonumber \\
& \quad + |\mathbb{E}U_\alpha(X)-p^\alpha|\\
&\equiv B_1+B_2+B_3.
\end{align}
It follows from Lemma \ref{lemma.largebias} and Lemma \ref{lemma.lowpartbiasvariance} that
\begin{align}
    B_1 &= |\mathbb{E}S_{K,\alpha}(X)-p^\alpha| \\
    & \le \frac{c_3}{(n\ln n)^\alpha} \\
    & \lesssim \frac{1}{(n\ln n)^\alpha},\\
  B_2 &= \mathbb{E}(S_{K,\alpha}(X)-1)\mathbbm{1}(S_{K,\alpha}(X)\ge1) \\
  & \le \mathbb{E}S_{K,\alpha}^2 \\
  & \le n^{8c_2\ln 2}\frac{(4c_1\ln n)^{2\alpha+2}}{n^2} \\
  & \lesssim \frac{(\ln n)^{2\alpha+2}}{n^{2-\epsilon}},\\
  B_3 &= |\mathbb{E}U_\alpha(X)-p^\alpha| \\
  & \le \frac{17}{n^\alpha(c_1\ln n)^{2-\alpha}} + \frac{8310(1+\alpha)}{\alpha(2-\alpha)}p^\alpha n^{-c_1/8} \\
 & \lesssim \frac{1}{n^\alpha(\ln n)^{2-\alpha}}.
\end{align}
To sum up, the total bias is upper bounded by
\begin{align}
  |B(\xi)| & \lesssim \frac{1}{(n\ln n)^\alpha} + \frac{(\ln n)^{2\alpha+2}}{n^{2-\epsilon}} + \frac{1}{n^\alpha(\ln n)^{2-\alpha}} \\
  & \lesssim \frac{1}{n^\alpha(\ln n)^{2-\alpha}}.
\end{align}
For the variance, we have
\begin{align}
 & \mathsf{Var}(\xi)\nonumber \\
  &\le \mathsf{Var}(S_{K,\alpha}(X)) + \mathsf{Var}(U_\alpha(X)) \nonumber \\
  & \quad  + (\mathbb{E}L_\alpha(X)-\mathbb{E}U_\alpha(X))^2\\
 &\le \mathbb{E}S_{K,\alpha}^2(X) + \mathsf{Var}(U_\alpha(X)) + (\mathbb{E}L_\alpha(X)-\mathbb{E}U_\alpha(X))^2\\
 &\lesssim \frac{(\ln n)^{2\alpha+2}}{n^{2\alpha-\epsilon}} + \left(\frac{p}{n} + \frac{1}{n^2}\right) + \frac{1}{n^{2\alpha}(\ln n)^{4-2\alpha}}\\
 &\lesssim \frac{1}{n^2} + \frac{p}{n},
\end{align}
where we have used Lemma \ref{lemma.largebias} and
\begin{align}
&  |\mathbb{E}L_\alpha(X)-\mathbb{E}U_\alpha(X)| \nonumber \\
 &\le |\mathbb{E}S_{K,\alpha}(X)-p^\alpha| \nonumber \\
 & \quad +  \mathbb{E}(S_{K,\alpha}(X)-1)\mathbbm{1}(S_{K,\alpha}(X)\ge1) \nonumber \\
 & \quad + |\mathbb{E}U_\alpha(X)-p^\alpha|\\
  &= B_1 + B_2 + B_3\\
  &\lesssim \frac{1}{n^\alpha(\ln n)^{2-\alpha}}.
\end{align}

\item Case 4: $p\ge 4\Delta$.

In this case, the bias is upper bounded by
\begin{align}
&  |B(\xi)| \nonumber \\
&\le |\mathbb{E}U_\alpha(X)-p^\alpha|+(|\mathbb{E}L_\alpha(X)|+p^\alpha)\mathbb{P}(Y\le 2\Delta)\\
  &\le \frac{17}{n^\alpha(c_1\ln n)^{2-\alpha}} + \frac{8310(1+\alpha)}{\alpha(2-\alpha)}p^\alpha n^{-c_1/8} \nonumber \\
  & \quad  + 2\mathbb{P}(nY\le 2n\Delta)\\
  &\le \frac{17}{n^\alpha(c_1\ln n)^{2-\alpha}} + \frac{8310(1+\alpha)}{\alpha(2-\alpha)}p^\alpha n^{-c_1/8} + 2e^{-c_1\ln n/2}\\
  &= \frac{17}{n^\alpha(c_1\ln n)^{2-\alpha}} + \frac{8310(1+\alpha)}{\alpha(2-\alpha)}p^\alpha n^{-c_1/8} + 2n^{-c_1/2}\\
  &\lesssim \frac{1}{n^\alpha(\ln n)^{2-\alpha}},
\end{align}
where we have used Lemma \ref{lemma.poissontail} to bound $\mathbb{P}(nY\le 2n\Delta)$. The variance is then upper bounded by
\begin{align}
 & \mathsf{Var}(\xi) \nonumber \\
 &\le \mathsf{Var}(U_\alpha(X)) \nonumber \\
 & \quad + \left(\mathsf{Var}(L_\alpha(X))+(\mathbb{E}L_\alpha(X)-\mathbb{E}U_\alpha(X))^2\right)\mathbb{P}(nY\le 2n\Delta)\\
  &\le \frac{202p}{n} + \frac{8}{n^2} + \frac{28800}{\alpha^2}p^{2\alpha} n^{-c_1/4} \nonumber \\
  & \quad + \frac{120}{\alpha}p^{2\alpha} n^{-c_1/8} + 2\mathbb{P}(nY\le 2n\Delta)\\
  &\le \frac{202p}{n} + \frac{8}{n^2} + \frac{28800}{\alpha^2}p^{2\alpha} n^{-c_1/4} \nonumber \\
  & \quad + \frac{120}{\alpha}p^{2\alpha} n^{-c_1/8} + 2n^{-c_1/2}\\
  &\lesssim \frac{p}{n} + \frac{1}{n^2},
\end{align}
where we have used Lemma \ref{lemma.largebias}.

\end{enumerate}

\subsection{Proof of Lemma~\ref{lemma.priorconstruct}}

The existence of the two prior distributions $\nu_0$ and $\nu_1$ follows directly from a standard functional analysis argument proposed by Lepski, Nemirovski, and Spokoiny \cite{Lepski--Nemirovski--Spokoiny1999estimation}, and elaborated in best polynomial approximation by Cai and Low \cite[Lemma 1]{Cai--Low2011}. It suffices to replace the interval $[-1,1]$ with $[0,1]$ and the function $|x|$ with $x^\alpha$ in the proof of Lemma 1 in \cite{Cai--Low2011}.

\subsection{Proof of Lemma~\ref{lemma.lowermeanvar}}

We compute the difference of the expectations as follows,
\begin{align}
& \bE_{\mu_1^{S'}} \underline{F_\alpha}(P) - \bE_{\mu_0^{S'}} \underline{F_\alpha}(P) \nonumber \\
&= \sum_{i = 1}^{S'} \left( \bE_{\mu_1} p_i^\alpha - \bE_{\mu_0} p_i^\alpha \right) \\
& = 2S'M^\alpha E_L[x^\alpha]_{[0,1]} \\
& = 2 \left( \frac{\alpha}{c} \right)^\alpha n^\alpha (\ln n)^\alpha \left( \frac{d_1 \ln n}{n} \right)^\alpha \frac{\mu(2\alpha)}{2^{2\alpha}} \frac{1}{(d_2 \ln n)^{2\alpha}} (1 + o(1)) \\
& = 2 \left( \frac{\alpha}{c} \right)^\alpha \frac{\mu(2\alpha) d_1^\alpha}{(2 d_2)^{2\alpha}}(1 + o(1)),
\end{align}
where $\mu(2\alpha)$ is the constant given by Lemma~\ref{lemma.nonasympxa}.

We also have bounds for the variance:
\begin{align}
\mathsf{Var}_{\mu_j^{S'}}(\underline{F_\alpha}(P)) & = \bE_{\mu_j^{S'}} \left( \underline{F_\alpha}(P) - \bE_{\mu_j^{S'}} \underline{F_\alpha}(P) \right)^2 \\
& = \sum_{i = 1}^{S'} \bE_{\mu_j}\left(  p_i^\alpha - \bE_{\mu_j} p_i^\alpha\right)^2 \\
& \leq S' \bE_{\mu_j} p_i^{2\alpha} \\
& \leq S' M^{2\alpha} \\
& \leq \left( \frac{\alpha d_1^2}{c} \right)^\alpha \frac{(\ln n)^{3\alpha}}{n^\alpha}   \to 0, \quad j = 0,1.
\end{align}

%

For any integer $y \geq 0$,
\begin{align}
 F_{1,M}(y) - F_{0,M}(y) & = \int \frac{e^{-np}(np)^y}{y!} \left( \mu_1(dp)- \mu_0(dp)\right)\\
 & = \int \sum_{i = 0}^\infty \frac{(-1)^i (np)^{i+y}}{i!y!}\left( \mu_1(dp)- \mu_0(dp)\right),
\end{align}
where we have used the Taylor expansion of $e^{-x}$.

Now we proceed to bound the total variation distance between the marginal distributions under two priors $\mu_0,\mu_1$.

\begin{align}
& \sum_{y = 0}^\infty |F_{1,M}(y) - F_{0,M}(y)| \nonumber \\
 & = \sum_{y = 0}^{\frac{d_2 \ln n}{2}}  |F_{1,M}(y) - F_{0,M}(y)| + \sum_{y > \frac{d_2 \ln n}{2}}  |F_{1,M}(y) - F_{0,M}(y)| \\
& \triangleq D_1 + D_2.
\end{align}

Note that we take $d_1 = 1, d_2 = 10e$ in the assumption. We bound $D_2$ in the following way:
\begin{align}
D_2 & = \sum_{y > \frac{d_2 \ln n}{2}} \left|\int \frac{e^{-np}(np)^y}{y!} \left( \mu_1(dp)- \mu_0(dp)\right)\right| \\
& \leq \int\bP\left(\mathsf{Poi}(np) > \frac{d_2 \ln n}{2}\right)\left( \mu_1(dp)+ \mu_0(dp)\right) \\
& \leq 2\cdot\bP\left (\mathsf{Poi}(d_1 \ln n) > \frac{d_2 \ln n}{2}\right) \\
& \leq 2\cdot \left( \frac{e^{5e-1}}{(5e)^{5e}} \right)^{\ln n} \\
& = \frac{2}{n^{1+5e\ln 5}} \leq \frac{2}{n^{22}},
\end{align}
where in the fourth step we have applied Lemma~\ref{lemma.poissontail}.

We bound $D_1$ as follows:
\begin{align}
D_1 & = \sum_{y = 0}^{\frac{d_2 \ln n}{2}}  |F_{1,M}(y) - F_{0,M}(y)| \\
& = \sum_{y = 0}^{\frac{d_2 \ln n}{2}} \frac{1}{y!} \left | \sum_{i = 0}^\infty \int \frac{(-1)^i (np)^{i+y}}{i!} \left( \mu_1(dp)- \mu_0(dp)\right) \right|\\
& = \sum_{y = 0}^{\frac{d_2 \ln n}{2}} \frac{1}{y!} \left | \sum_{i > d_2 \ln n -y} \int \frac{(-1)^i (np)^{i+y}}{i!} \left( \mu_1(dp)- \mu_0(dp)\right) \right| \\
& \leq  \sum_{y = 0}^{\frac{d_2 \ln n}{2}} \frac{1}{y!}  \sum_{i > d_2 \ln n -y} \frac{ (d_1 \ln n)^{i+y}}{i!}  \\
& = \sum_{y = 0}^{\frac{d_2 \ln n}{2}} \frac{(d_1 \ln n)^y}{y!}  \sum_{i > d_2 \ln n -y} \frac{ (d_1 \ln n)^{i}}{i!},
\end{align}
where in the third step we have used the fact that $\mu_1$ and $\mu_0$ have matching moments up to order $d_2 \ln n$. The Lagrangian remainder for Taylor series of $e^x, x>0$ shows that
\begin{equation}
\sum_{j> m} \frac{x^j}{j!}  = e^{\xi} \frac{x^{m+1}}{(m+1)!},
\end{equation}
where $0\leq \xi \leq x$. Applying this result, we have
\begin{align}
D_1 & \leq \sum_{y = 0}^{\frac{d_2 \ln n}{2}} \frac{(d_1 \ln n)^y}{y!}  \sum_{i > d_2 \ln n -y} \frac{ (d_1 \ln n)^{i}}{i!} \\
& \leq \sum_{y = 0}^{\frac{d_2 \ln n}{2}} \frac{(d_1 \ln n)^y}{y!}  e^{d_1 \ln n} \frac{(d_1 \ln n)^{d_2 \ln n -y+1}}{(d_2 \ln n -y +1)!} \\
& \leq \sum_{y = 0}^{\frac{d_2 \ln n}{2}} \frac{(d_1 \ln n)^y}{y!}  e^{d_1 \ln n} \frac{(d_1 \ln n)^{d_2 \ln n -y+1} e^{d_2 \ln n -y + 1}}{(d_2 \ln n -y +1)^{d_2 \ln n -y + 1}} \\
& \leq  e^{d_1 \ln n + d_2 \ln n + 1} \left( \frac{d_1 \ln n}{\frac{d_2 \ln n}{2} + 1} \right)^{\frac{d_2 \ln n}{2}+1} \sum_{y = 0}^{\frac{d_2 \ln n}{2}} \frac{(d_1 \ln n)^y}{y!} \\
& \leq e^{2d_1 \ln n + d_2 \ln n + 1} \left( \frac{d_1 \ln n}{\frac{d_2 \ln n}{2} + 1} \right)^{\frac{d_2 \ln n}{2}+1}
\end{align}
where in the third step we have used the fact that $n! \geq \left( \frac{n}{e} \right)^n$.

Plugging $d_1 = 1, d_2 = 10 e$ in, we have
\begin{equation}
D_1 \leq \frac{1}{5 n^6}.
\end{equation}

Combining bounds on $D_1$ and $D_2$ together, we have
\begin{align}
V(F_{1,M}, F_{0,M})  & = \frac{1}{2}\sum_{y = 0}^\infty |F_{1,M}(y) - F_{0,M}(y)| \\
&  \leq \min\left\{\frac{1}{2}(D_1 + D_2),1\right\} \\
& \leq \frac{1}{n^6}.
\end{align}

\subsection{Proof of Lemma \ref{lem_conti}}
For $\eta\in(0,1/2)$, define
\begin{align}
  f_\eta(x) = \left(\frac{1-\eta}{2}x+\frac{1+\eta}{2}\right)^\beta,
\end{align}
then $E_L[f_\eta]_{[-1,1]}=E_L[x^\beta]_{[\eta,1]}$. For $\varphi(x)=\sqrt{1-x^2}$, denote the second-order Ditzian-Totik modulus of smoothness by
\begin{align}
  \omega_\varphi^2(f,t) & \triangleq \sup\Big\{\Big|f(u)+f(v)-2f\Big(\frac{u+v}{2}\Big)\Big|:\nonumber \\
  & \quad u,v\in[-1,1],|u-v|\le 2t\varphi\Big(\frac{u+v}{2}\Big)\Big\},
\end{align}
then it is straightforward to obtain that for all $n\le(4\eta)^{\frac{\beta}{2}-1}$, 
\begin{align}\label{eq:modulus_smoo}
  \omega_\varphi^2(f_\eta,n^{-1}) & = \left|\eta^\beta + \left(\eta+\frac{2(1-\eta)}{n^2+1}\right)^\beta - 2\left(\eta+\frac{1-\eta}{n^2+1}\right)^\beta\right|.
\end{align}
It follows directly from (\ref{eq:modulus_smoo}) that when $n\le \min\{\frac{1}{\sqrt{\eta}},(4\eta)^{\frac{\beta}{2}-1}\}$, 
\begin{align}
  \frac{2\cdot 5^\beta-4^\beta-6^\beta}{(2n)^{2\beta}}\le \omega_\varphi^2(f_\eta,n^{-1}) \le \frac{1}{n^{2\beta}}. 
\end{align}

The relationship between $\omega_\varphi^2(f_\eta,n^{-1})$ and $E_n[f_\eta]_{[-1,1]}$ was shown in \cite[Thm. 7.2.1, Thm. 7.2.4]{Ditzian--Totik1987} that there exists two universal positive constants $M_1,M_2$ such that
\begin{align}\label{eq:appro_direct}
  E_n[f_\eta]_{[-1,1]} &\le M_1\omega_\varphi^2(f_\eta,n^{-1}),\\
  \frac{1}{n^2}\sum_{k=0}^n (k+1)E_k[f_\eta]_{[-1,1]} &\ge M_2\omega_\varphi^2(f_\eta,n^{-1}).\label{eq:appro_converse}
\end{align}
Applying (\ref{eq:appro_direct}) and (\ref{eq:appro_converse}) and setting the approximation order $N=DL$ with a positive constant $D>1$ to be specified later, then given $\eta=1/N^2$, the non-increasing property of $E_n[f_\eta]_{[-1,1]}$ with respect to $n$ yields that
\begin{align}
 & E_L[f_\eta]_{[-1,1]} \nonumber \\
  &\ge \frac{1}{N-L}\sum_{k=L+1}^N E_k[f_\eta]_{[-1,1]} \\
  &\ge \frac{1}{N^2}\sum_{k=L+1}^N (k+1)E_k[f_\eta]_{[-1,1]}\\
&\ge M_2\omega_\varphi^2(f_\eta,N^{-1}) - \frac{E_0[f_\eta]_{[-1,1]}}{N^2} \nonumber \\
& \quad - \frac{1}{N^2}\sum_{k=1}^L (k+1)E_k[f_\eta]_{[-1,1]}\\
&\ge \frac{M_2(2\cdot 5^\beta-4^\beta-6^\beta)}{(2N)^{2\beta}} - \frac{1}{N^2} - \frac{2M_1}{DN}\sum_{k=1}^L \omega_\varphi^2(f_\eta,k^{-1})\\
&\ge \frac{M_2(2\cdot 5^\beta-4^\beta-6^\beta)}{(2DL)^{2\beta}} - \frac{1}{(DL)^2} - \frac{2M_1}{DN}\sum_{k=1}^L \frac{1}{k^{2\beta}}\\
&\ge \frac{M_2(2\cdot 5^\beta-4^\beta-6^\beta)}{(2DL)^{2\beta}} - \frac{1}{D^2L^{2\beta}} - \frac{2M_1}{D^2L}\int_0^L \frac{dx}{x^{2\beta}}\\
&= L^{-2\beta}\left[\frac{M_2(2\cdot 5^\beta-4^\beta-6^\beta)}{(2D)^{2\beta}} - \frac{1}{D^2} - \frac{M_1}{D^2(1-2\beta)}\right].
\end{align}
Due to $0<2\beta<1$, for a sufficiently large universal constant $D$ we can obtain that
\begin{align}
  \liminf_{L\to\infty} L^{2\beta}E_L[x^\beta]_{[\eta,1]} = \liminf_{L\to\infty} L^{2\beta}E_L[f_\eta]_{[-1,1]} > 0.
\end{align}

\subsection{Proof of Lemma \ref{lem_equiv}}
Fix $\delta>0$. Let $\hat{F}({\bf Z})$ be a near-minimax estimator of $F_\alpha(P)$ under the Multinomial model. The estimator $\hat{F}({\bf Z})$ obtains the number of samples $n$ from observation $\bf Z$. By definition, we have
  \begin{align}
    \sup_{P\in \mathcal{M}_S} \mathbb{E}_{\mathrm{Multinomial}}|\hat{F}(\mathbf{Z})-F_\alpha(P)|^2 < R(S,n)+\delta,
  \end{align}
  where $R(S,n)$ is the minimax $L_2$ risk under the Multinomial model. Given $P\in\mathcal{M}_S(\gamma)$, let $\mathbf{Z}=[Z_1,\cdots,Z_S]^T$ with $Z_i\sim \mathsf{Poi}(np_i)$ and let $n'=\sum_{i=1}^S Z_i\sim \mathsf{Poi}(n\sum_{i=1}^Sp_i)$, we use the estimator $d_1^\alpha\cdot\hat{F}(\mathbf{Z})$ to estimate $F_\alpha(P)$.

The triangle inequality gives
  \begin{align}
   & \frac{1}{2}R_P(S,n,\gamma) \nonumber \\
    &\le \frac{1}{2}\mathbb{E}_P|d_1^\alpha\cdot\hat{F}({\bf Z})-F_\alpha(P)|^2 \\
    &\le \mathbb{E}_P\left|d_1^\alpha\cdot\hat{F}({\bf Z})-d_1^\alpha F_\alpha\left(\frac{P}{\sum_{i=1}^S p_i}\right)\right|^2 \nonumber \\
    & \quad + \left|\left(\sum_{i=1}^Sp_i\right)^{\alpha}-d_1^\alpha\right|^2F_\alpha^2\left(\frac{P}{\sum_{i=1}^S p_i}\right)\\
    &\le d_1^{2\alpha}\mathbb{E}_P\left[\left.\left|\hat{F}(\mathbf{Z})-F_\alpha\left(\frac{P}{\sum_{i=1}^S p_i}\right)\right|^2\right|n'=m\right]\mathbb{P}(n'=m) \nonumber \\
    & \quad +\frac{4d_1^{2\alpha-2}}{(\ln n)^{2\gamma}}\cdot \left(\frac{M}{d_1}\right)^{2\alpha-2}\\
    &\le d_1^{2\alpha}\sum_{m=0}^\infty R(S,m)\mathbb{P}(n'=m) + \delta + \frac{4M^{2\alpha-2}}{(\ln n)^{2\gamma}}\\
    &\le d_1^{2\alpha}R(S,\frac{d_1n}{2})\mathbb{P}(n'\ge \frac{d_1n}{2}) + \mathbb{P}(n'\le \frac{d_1n}{2}) + \delta +\frac{4M^{2\alpha-2}}{(\ln n)^{2\gamma}}\\
    &\le d_1^{2\alpha}R(S,\frac{d_1n}{2})+ \exp(-\frac{d_1n}{8}) + \delta + \frac{4M^{2\alpha-2}}{(\ln n)^{2\gamma}},
  \end{align}
  where we have used the fact that conditioned on $n'=m$, $\mathbf{Z}\sim \mathsf{Multinomial}(m,\frac{P}{\sum_ip_i})$, and the last step follows from Lemma~\ref{lemma.poissontail}. The proof is completed by the arbitrariness of $\delta$.

\section{Proof of auxiliary lemmas}

\subsection{Proof of Lemma~\ref{lemma.gxa}}

We obtain the polynomial $g(x;a)$ via the Hermite interpolation formula. Concretely, the following $\mathsf{WolframAlpha}$ (\url{http://www.wolframalpha.com/}) command will give us $g(x;a)$:

\textsf{InterpolatingPolynomial[\{\{\{0\}, 0, 0, 0, 0, 0\}, \{\{a\}, 1, 0, 0, 0, 0\}\}, x]}.


\subsection{Proof of Lemma~\ref{lemma.varentropy}}
For brevity, denote $\mathsf{Var}(-\ln P(X))$ as $V(P)$, we have
\begin{equation}
V(P) = \sum_{i = 1}^S p_i (\ln p_i)^2 - \left( \sum_{i = 1}^S p_i \ln p_i \right)^2.
\end{equation}

We note that
\begin{equation}
V(P) \leq \sum_{i = 1}^S p_i (\ln p_i)^2 \leq \sum_{i = 1}^S p_i (\ln p_i -1)^2.
\end{equation}

The function $x(\ln x-1)^2$ on $[0,1]$ has second derivative $\frac{2\ln x}{x}$, hence is concave. Thus, the expression $\sum_{i = 1}^S p_i (\ln p_i -1)^2$ attains its maximum when $P$ is uniform distribution. In other words, we have shown that
\begin{equation}
V(P) \leq (\ln S+1)^2.
\end{equation}

A tighter bound which gives better constant can be constructed as follows. We define the Lagrangian:
\begin{equation}
\cL = \sum_{i = 1}^S p_i (\ln p_i)^2 - \left( \sum_{i = 1}^S p_i \ln p_i \right)^2 + \lambda \left( \sum_{i = 1}^S p_i -1 \right).
\end{equation}

Taking derivatives with respect to $p_i$, we obtain
\begin{align}
\frac{\partial \cL}{\partial p_i} & = (\ln p_i)^2 + p_i (2 \ln p_i) \times \frac{1}{p_i} - 2 \left( \sum_{i = 1}^S p_i \ln p_i \right)(1+\ln p_i) + \lambda \label{eqn.dev0} \\
& = 0, \quad \text{for all } i.
\end{align}

It is equivalent to
\begin{equation}
(\ln p_i)^2 + 2 \ln p_i + 2 H(P)(1 + \ln p_i) + \lambda = 0,\quad \forall i
\end{equation}

Note that it is a quadratic form for $\ln p_i$ with the same coefficients. Solving for $\ln p_i$, we obtain that
\begin{equation}
\ln p_i = -(1+H(P))\pm \sqrt{ 1 + H^2(P) - \lambda}.
\end{equation}

It implies that components of the maximum achieving distribution can only take two values. Assume $p_i \in \{q_1,q_2\}, \forall i$. Suppose $q_1$ appears $k$ times, we have
\begin{equation}
kq_1  + (S-k) q_2 = 1.
\end{equation}

Now we compute the functional
\begin{align}
& V(P)\nonumber \\
 & = k q_1 (\ln q_1)^2 + (S-k) q_2 (\ln q_2)^2 - \left( k q_1 \ln q_1 + (S-k) q_2 \ln q_2 \right)^2 \\
& = k q_1 (\ln q_1)^2 + (1-kq_1) (\ln q_2)^2 - k^2 q_1^2(\ln q_1)^2 \nonumber \\
& \quad - (1-k q_1)^2 (\ln q_2)^2 - 2k q_1 (1- kq_1) (\ln q_1) (\ln q_2) \\
& = k q_1 (1-kq_1) \left( \ln \frac{q_2}{q_1} \right)^2.
\end{align}

Since $q_2 = \frac{1-kq_1}{S-k}$, we have
\begin{equation}
V(P) =  k q_1 (1-kq_1)  \left( \ln \frac{1-k q_1}{S q_1 - k q_1} \right)^2.
\end{equation}

Denote $x = kq_1, y = k/S$, we have
\begin{equation}
V(P) = x(1-x) \left( \ln \frac{1-x}{x} - \ln \frac{1-y}{y} \right)^2.
\end{equation}

Fixing $x$, we see $V(P)$ is a monotone function of $y$. Without loss of generality, by symmetry we assume $x \leq 1/2$. Then, the maximum achieving $y = \frac{S-1}{S}$, and $V(P)$ as a function of $x$ is
\begin{equation}
V(P) = x(1-x) \left(\ln \frac{1-x}{x} + \ln(S-1) \right)^2.
\end{equation}

Taking derivatives with respect to $x$, ignoring the minimum achieving $x$, we obtain the following equation for maximum achieving value of $x$, which is denoted as $x_1$:
\begin{equation}
(1-2x_1) \left( \ln \left( \frac{1}{x_1} - 1 \right) + \ln(S-1)\right) = 2.
\end{equation}

Denoting $S-1$ by $m$, we have
\begin{equation}
(1-2x_1) \ln \frac{m(1-x_1)}{x_1} = 2,
\end{equation}
which is equivalent to
\begin{equation}
(2-2x_1) \ln \frac{m(1-x_1)}{x_1} = 2 + \ln \frac{m(1-x_1)}{x_1}.
\end{equation}

Multiplying both sides by $\frac{x_1}{2} \ln \frac{m(1-x_1)}{x_1}$, we obtain
\begin{align}
V_{\max} & = x_1(1-x_1)\left( \ln \frac{m(1-x_1)}{x_1} \right)^2 \\
&  = x_1 \left( \ln \frac{m(1-x_1)}{x_1} + \frac{1}{2} \left(\ln \frac{m(1-x_1)}{x_1} \right)^2 \right).
\end{align}

Note that for $x\in (0,1/2]$,
\begin{equation}
\ln \frac{m(1-x)}{x} \in [\ln m, \infty),
\end{equation}
and if $S\geq 4$, we have $\ln m = \ln (S-1) >1$.

Using the bound $z \leq z^2, z \geq 1$, we have
\begin{equation}\label{eqn.vupper}
V_{\max} \leq \frac{3}{2} x_1 \left(\ln \frac{m(1-x_1)}{x_1} \right)^2.
\end{equation}

Taking derivatives with respect to $z$ for $ z \left(\ln \frac{m(1-z)}{z} \right)^2, z\in (0,1/2]$, we have
\begin{align}
& \frac{d}{dz}\left( z \left(\ln \frac{m(1-z)}{z} \right)^2 \right) \nonumber \\
& = \ln \left( \frac{m(1-z)}{z} \right) \times \left( \ln \frac{m(1-z)}{z} + \frac{2}{z-1} \right) \\
& \geq \ln \left( \frac{m(1-z)}{z} \right) \times \left( \ln m - 4 \right),
\end{align}
which is always nonnegative if $m \geq e^4$, i.e., $S\geq 56$. Hence, we know that when $S\geq 56$, the function $z \left(\ln \frac{m(1-z)}{z} \right)^2$ is an increasing function of $z$ for $z\in (0,1/2]$, thus achieves its maximum at $z = 1/2$.

Then, we obtain
\begin{equation}
V_{\max} \leq \frac{3}{4}(\ln m)^2 \leq \frac{3}{4} (\ln S)^2.
\end{equation}

\subsection{Proof of Lemma~\ref{lemma.poissonmultinomial}}

Denote by $\hat{F}_{n}^P, \hat{F}_n$ the estimator for $F(P)$ under the Poissonized model and the Multinomial model with sample size $n$, respectively. By the minimax theorem~\cite{Wald1950statistical}, the minimax risk is the supremum of Bayes risk under all priors, i.e.,
\begin{align}
  R_P(S,n) &= \sup_\pi \inf_{\hat{F}_n^P} \int \bE_P|\hat{F}_n^P(\mathbf{Z}) - F(P)|^2 \pi(dP) \\
  &  \triangleq \sup_\pi R_B^P(S,n,\pi)\\
  R(S,n) &= \sup_\pi \inf_{\hat{F}_n} \int \bE_P|\hat{F}_n(\mathbf{Z}) - F(P)|^2 \pi(dP) \\
  & \triangleq \sup_\pi R_B(S,n,\pi)
\end{align}
where the supremum is taken over all priors on $\mathcal{M}_S$, and we denote the Bayes risk under prior $\pi$ in Poissonized and Multinomial models by $R_B^P(S,n,\pi)$ and $R_B(S,n,\pi)$, respectively. Using the property that independent $Z_i\sim \mathsf{Poi}(np_i), 1\le i\le S$ implies $(Z_1,\cdots,Z_S)|(\sum_{i=1}^S Z_i=n)\sim \mathsf{Multinomial}(n;p_1,\cdots,p_S)$, it is straightforward to show that for any prior $\pi$,
\begin{align}
  R_B^P(S,n,\pi) = \sum_{m=0}^\infty R_B(S,m,\pi)\bP\left(\mathsf{Poi}(n)=m\right).
\end{align}

Since for $n>m$, the Bayes estimator with sample size $m$ can be used for estimation with sample size $n$ by neglecting the last $n-m$ samples, we conclude that the function $R_B(S,n,\pi)$ is non-increasing in $n$. Moreover, it is obvious that $R(S,n,\pi)\le \sup_{P\in\mathcal{M}_S} |F(P)|^2$ by considering the zero estimator. Then we have
\begin{align}
& R_B^P(S,2n,\pi) \nonumber \\
&= \sum_{m=0}^\infty R_B(S,m,\pi)\bP\left(\mathsf{Poi}(2n)=m\right) \\
&= \sum_{m=0}^{n-1} R_B(S,m,\pi)\bP\left(\mathsf{Poi}(2n)=m\right) \nonumber \\
& \quad + \sum_{m=n}^\infty R_B(S,m,\pi)\bP\left(\mathsf{Poi}(2n)=m\right) \\
&\le \sup_{P\in\mathcal{M}_S} |F(P)|^2 \cdot \bP\left(\mathsf{Poi}(2n)<n\right) \nonumber \\
& \quad + R_B(S,n,\pi)\sum_{m=n}^\infty \bP\left(\mathsf{Poi}(2n)=m\right) \\
&\le \sup_{P\in\mathcal{M}_S} |F(P)|^2\cdot e^{-n/4} + R_B(S,n,\pi)
\end{align}
and
\begin{align}
R_B^P(S,n/2,\pi) & = \sum_{m=0}^\infty R_B(S,m,\pi)\bP\left(\mathsf{Poi}(n/2)=m\right) \\
& \ge \sum_{m=0}^n R_B(S,m,\pi)\bP\left(\mathsf{Poi}(n/2)=m\right) \\
& \geq R_B(S,n,\pi) \bP(\mathsf{Poi}(n/2) \leq n) \\
& \geq \frac{1}{2}R_B(S,n,\pi),
\end{align}
where we have applied Lemma~\ref{lemma.poissontail} to show that $\bP\left(\mathsf{Poi}(2n)<n\right) \le e^{-n/4}$ and the Markov inequality to show that $\bP(\mathsf{Poi}(n/2) \leq n)\ge 1/2$. Then we are done by taking the supremum over all prior $\pi$ at both sides of these two inequalities.

\subsection{Proof of Lemma~\ref{lemma.nonasympxa}}

We first show the limiting result. Defining $y^2 = x$, we know
\begin{equation}
E_n[x^\alpha]_{[0,1]} = E_{2n}[y^{2\alpha}]_{[-1,1]}.
\end{equation}
Applying Theorem~\ref{thm.bernsteinxp} to our settings, for any $\alpha>0$ we have
\begin{align}
\lim_{n\to \infty} n^{2\alpha} E_n[x^\alpha]_{[0,1]} & = \lim_{n\to \infty} n^{2\alpha} E_{2n}[y^{2\alpha}]_{[-1,1]} \\
& = \frac{1}{2^{2\alpha}} \lim_{n\to \infty} (2n)^{2\alpha}E_{2n}[y^{2\alpha}]_{[-1,1]} \\
& = \frac{\mu(2\alpha)}{2^{2\alpha}} .
\end{align}

Korneichuk \cite[Sec. 6.2.5]{Korneichuk1991} showed the inequality
\begin{equation}
E_{n}[f] \leq \omega\left(f, \frac{\pi}{n+1}\right)
\end{equation}
for all $f \in C[-1,1]$, where $\omega(f,\delta)$ is the first order modulus of smoothness, defined as
\begin{equation}
\omega(f,\delta) \triangleq \sup \{ |f(x) - f(x+\delta)|: x \in [-1,1], x+\delta \in [-1,1] \}.
\end{equation}

Bernstein \cite[Pg. 171]{Bernstein1958collected} showed
\begin{equation}
E_{n+1}[f(x)] \leq \frac{\pi}{2(n+1)}E_n[f'(x)],
\end{equation}
where $f\in C^1[-1,1]$.

For $0<\alpha \leq 1/2$, $\omega(y^{2\alpha}, \delta) \leq \delta^{2\alpha}, \delta\leq 2$, hence we know
\begin{equation}\label{eqn.useappendixlemma9}
E_n[x^\alpha]_{[0,1]} = E_{2n}[y^{2\alpha}]_{[-1,1]} \leq \left( \frac{\pi}{2n} \right)^{2\alpha}.
\end{equation}

For $1/2 < \alpha <1$, noting that
\begin{equation}
\left( (y^2)^\alpha \right)' = 2 \alpha y (y^2)^{\alpha-1},
\end{equation}
and that
\begin{equation}
\omega( 2 \alpha y (y^2)^{\alpha-1}, \delta) \leq 2 \alpha \delta^{2\alpha-1}, \quad \delta\leq 2,
\end{equation}
we know for $1/2 <\alpha<1$,
\begin{align}
E_{2n}[y^{2\alpha}]_{[-1,1]} & \leq \frac{\pi}{2(2n)} E_{2n-1}[2 \alpha y (y^2)^{\alpha-1}]_{[-1,1]} \\
&  \leq 2 \alpha\frac{\pi}{4n} \left( \frac{\pi}{2n} \right)^{2\alpha-1} \\
& = \alpha \left( \frac{\pi}{2n} \right)^{2\alpha} \\
& \leq \left( \frac{\pi}{2n} \right)^{2\alpha}.
\end{align}
Plugging in $x=0$ yields $|g_{0,\alpha}|<(\pi/2n)^{2\alpha}$, hence
\begin{align}
   \max_{0\le x\le 1}|R_{n,\alpha}(x) - x^\alpha| \le E_n[x^\alpha]_{[0,1]} + |g_{0,\alpha}| \le 2\left(\frac{\pi}{2n}\right)^{2\alpha}.
\end{align}

%
%

For $1<\alpha<3/2$, by defining $y^2=x$ we know that
\begin{align}
  E_n[x^\alpha]_{[0,1]} & = E_{2n}[y^{2\alpha}]_{[-1,1]} \\
  & \le \frac{\pi}{4n}E_{2n-1}[2\alpha y^{2\alpha-1}]_{[-1,1]} \\
  &  \le \frac{\pi^2}{8n(2n-1)}E_{2n-2}[2\alpha(2\alpha-1) y^{2\alpha-2}]_{[-1,1]}
\end{align}
and
\begin{align}
  \omega\left(2\alpha(2\alpha-1)y^{2\alpha-2},\frac{\pi}{2n-1}\right) = \frac{2\alpha(2\alpha-1)\pi^{2\alpha-2}}{(2n-1)^{2\alpha-2}}.
\end{align}
Hence we have
\begin{align}
  E_n[x^\alpha]_{[0,1]} & \le \frac{\pi^2}{8n(2n-1)}\cdot\frac{2\alpha(2\alpha-1)\pi^{2\alpha-2}}{(2n-1)^{2\alpha-2}} \\
  &  < \frac{\alpha(2\alpha-1)}{2}\left(\frac{\pi}{2n-1}\right)^{2\alpha} \\
  &  < \frac{3}{2}\left(\frac{\pi}{n}\right)^{2\alpha}.
\end{align}
Plugging in $x=0$ yields $|g_{0,\alpha}|<\frac{3}{2}(\pi/n)^{2\alpha}$, hence
\begin{align}
   \max_{0\le x\le 1}|R_{n,\alpha}(x) - x^\alpha| \le E_n[x^\alpha]_{[0,1]} + |g_{0,\alpha}| \le 3\left(\frac{\pi}{n}\right)^{2\alpha}.
\end{align}

Moreover, it has been shown in \cite[Pg. 207]{Devore--Lorentz1993} that
\begin{align}\label{eq:derivative_bound}
  \max_{0\le x\le 1}|R_{n,\alpha}'(x)-(x^\alpha)'| \le D \cdot E_n[(x^\alpha)']_{[0,1]} \le D\alpha \left(\frac{\pi}{2n}\right)^{2(\alpha-1)},
\end{align}
where $D>0$ is a positive universal constant, and the last inequality follows directly from (\ref{eqn.useappendixlemma9}). Then integrating on $x$ yields the pointwise bound
\begin{align}
  \left|R_{n,\alpha}(x)-x^\alpha\right| 
  & \le \int_0^x \left|R_{n,\alpha}'(t)-(t^\alpha)'\right|dt \\
  & \le \int_0^xD\alpha \left(\frac{\pi}{2n}\right)^{2(\alpha-1)}dt \\
  & = D\alpha \left(\frac{\pi}{2n}\right)^{2(\alpha-1)}x \\
  &  \triangleq \frac{D_1x}{n^{2(\alpha-1)}}.
\end{align}

In order to bound the coefficients of best polynomial approximations, we need the following result by Qazi and Rahman\cite[Thm. E]{Qazi--Rahman2007} on the maximal coefficients of polynomials on a finite interval.

\begin{lemma}\label{lemma.chebyshev}
Let $p_n(x) = \sum_{\nu=0}^n a_\nu x^\nu$ be a polynomial of degree at most $n$ such that $|p_n(x)|\leq 1$ for $x\in [-1,1]$. Then, $|a_{n-2\mu}|$ is bounded above by the modulus of the corresponding coefficient of $T_n$ for $\mu = 0,1,\ldots,\lfloor n/2 \rfloor$, and $|a_{n-1-2\mu}|$ is bounded above by the modulus of the corresponding coefficient of $T_{n-1}$ for $\mu = 0,1,\ldots,\lfloor (n-1)/2 \rfloor$. Here $T_n(x)$ is the $n$-th Chebyshev polynomials of the first kind. 
\end{lemma}

It is shown in Cai and Low\cite[Lemma 2]{Cai--Low2011} that all of the coefficients of Chebyshev polynomial $T_{2m}(x), m\in \mathbb{Z}_+$ are upper bounded by $2^{3m}$. If we view the best polynomial approximation of $x^\alpha$ or $-x \ln x$ over $[0,1]$ as the best polynomial approximation of $y^{2\alpha}$ or $-y^2 \ln y^2, y^2 = x$, then we would obtain an even polynomial over interval $[-1,1]$ represented as
\begin{equation}
\sum_{k = 0}^{n} g_{k,\alpha} y^{2k}\quad \text{or} \quad \sum_{k = 0}^{n} r_{k,H} y^{2k}.
\end{equation}

Applying Lemma~\ref{lemma.chebyshev} and equation~(\ref{eqn.gkhdefine}), we know that for all $k \leq n$, we have
\begin{equation}
|g_{k,\alpha}| \leq 2^{3n}, \quad |g_{k,H}| \leq 2^{3n}.
\end{equation}

\subsection{Proof of Lemma~\ref{lemma.approsmall}}

Define $x' = \frac{x}{4\Delta} \in [0,1]$. For $0<\alpha<1$, applying Lemma~\ref{lemma.nonasympxa}, we have
\begin{equation}
\left| (x')^\alpha - \sum_{k =1}^{K} g_{k,\alpha}(x')^k \right| \leq 2\left(\frac{\pi}{2K}\right)^{2\alpha}.
\end{equation}

Multiplying both sides by $(4\Delta)^\alpha$, we have
\begin{equation}
\left|\sum_{k = 0}^{K} g_{k,\alpha} (4\Delta)^{-k + \alpha} x^k - x^\alpha \right| \leq 2\left(\frac{\pi}{2}\right)^{2\alpha} \frac{(4\Delta)^\alpha}{K^{2\alpha}}=      \frac{c_3}{(n \ln n)^\alpha}.
\end{equation}

For the case $1<\alpha<3/2$, similar results hold for
\begin{align}
  c_3 = 3\left(\frac{4\pi^2c_1}{c_2^2}\right)^\alpha.
\end{align}

\subsection{Proof of Lemma~\ref{lemma.approsmallentropy}}

Define $x' = \frac{x}{4\Delta}$, hence for $x\in [0,4\Delta], x' \in [0,1]$. It follows from the best polynomial approximation result for $-x\ln x$ on $[0,1]$ that there exists a constant $d>0$ such that for all $x'\in [0,1]$,
\begin{equation}
\left| \sum_{k = 0}^K r_{k,H} (x')^k - (-x' \ln x') \right| \leq \frac{d}{K^2}.
\end{equation}
When $n$ is sufficiently large, we could take $d = \frac{\nu_1(2)}{2}$.
Taking $x' = 0$, we have
\begin{equation}
r_{0,H} \leq \frac{d}{K^2},
\end{equation}
hence
\begin{equation}
\left| \sum_{k = 1}^K r_{k,H} (x')^k - (-x' \ln x') \right| \leq \frac{2d}{K^2}.
\end{equation}
Now, multiplying both sides by $4\Delta$, we have
\begin{equation}
\left| \sum_{k = 1}^K r_{k,H} (4\Delta)^{-k+1}x^k  + x \left( \ln x - \ln (4\Delta) \right)\right| \leq \frac{2d(4\Delta)}{K^2}.
\end{equation}

Since we have defined $g_{k,H}$ as
\begin{equation}
g_{k,H} = r_{k,H}, 2\leq k \leq K,\quad g_{1,H} = r_{1,H} - \ln (4\Delta),
\end{equation}
we have
\begin{align}
\left| \sum_{k = 1}^K g_{k,H} (4\Delta)^{-k+1} x^k + x\ln x \right| & \leq  \frac{2d(4\Delta)}{K^2} \\
& = \frac{8dc_1}{c_2^2 n \ln n} \\
& = \frac{C}{n \ln n}.
\end{align}

When $n$ is sufficiently large, we could replace $d$ by $\nu_1(2)/2$, hence obtain
\begin{equation}
C = \frac{4c_1 \nu_1(2)}{c_2^2}.
\end{equation}

\subsection{Proof of Lemma~\ref{lemma.poissonmoment}}

We know that if $X\sim\mathsf{Poi}(\lambda)$, then it follows from \cite{Riordan1937moment} that
\begin{equation}
\bE X^k = \sum_{i=1}^k \lambda^i \left\{\begin{matrix} k \\ i \end{matrix}\right\},
\end{equation}
where $\left\{\begin{matrix} k \\ i \end{matrix}\right\}$ is the Stirling numbers of the second kind.

Using (\ref{eq:stirling_inequality}), we have
\begin{align}
\bE X^k = & \sum_{i=1}^k \lambda^i \left\{\begin{matrix} k \\ i \end{matrix}\right\} \\
& \leq \sum_{i =1}^k \lambda^i \binom{k}{i}i^{k-i} \\
& \leq \sum_{i = 1}^k M^i \binom{k}{i} M^{k-i} \\
& = M^k \sum_{i = 1}^k \binom{k}{i} \\
& \leq M^k 2^k \\
& = (2M)^k.
\end{align}

\bibliographystyle{IEEEtran}
\bibliography{di}

\begin{IEEEbiographynophoto}{Jiantao Jiao}
(S'13) received the B.Eng. degree with the highest honor in Electronic Engineering from Tsinghua University, Beijing, China in 2012, and a Master's degree in Electrical Engineering from Stanford University in 2014. He is currently working towards the Ph.D. degree in the Department of Electrical Engineering at Stanford University. He is a recipient of the Stanford Graduate Fellowship (SGF). His research interests include information theory and statistical signal processing, with applications in communication, control,
computation, networking, data compression, and learning. 
\end{IEEEbiographynophoto}

\begin{IEEEbiographynophoto}{Kartik Venkat}
(S'12) is a Ph.D. candidate in the Department of Electrical
Engineering at Stanford University. His research interests include statistical inference, information theory, machine learning, and their applications in genomics, wireless networks, neuroscience, and quantitative finance. Kartik received a Bachelor’s degree in Electrical Engineering from the Indian Institute of Technology, Kanpur
in 2010, and a Master's degree in Electrical Engineering from Stanford University in 2012. His honors include a Stanford Graduate Fellowship for Engineering and Sciences, the Numerical Technologies Founders Prize, and a Jack Keil Wolf ISIT Student Paper Award at the 2012 International Symposium on Information Theory.
\end{IEEEbiographynophoto}

\begin{IEEEbiographynophoto}{Yanjun Han}
(S'14) is currently working towards the B.Eng. degree in Electronic Engineering from Tsinghua University, Beijing, China. His research interests include information theory and statistics, with applications in communications, data compression, and learning.
\end{IEEEbiographynophoto}

\begin{IEEEbiographynophoto}{Tsachy Weissman}
(S'99-M'02-SM'07-F'13) graduated summa cum laude with a
B.Sc. in electrical engineering from the Technion in 1997, and earned
his Ph.D. at the same place in 2001. He then worked at Hewlett-Packard
Laboratories with the information theory group until 2003, when he joined
Stanford University, where he is Associate Professor of Electrical
Engineering and incumbent of the
STMicroelectronics chair in the School of Engineering.
He has spent leaves at the Technion, and at ETH Zurich.

Tsachy's research is focused on information theory, statistical signal
processing, the interplay between them, and their applications.

He is recipient of several best paper awards, and prizes for excellence in research.

He served on the editorial board of the \textsc{IEEE Transactions on Information Theory} from Sept. 2010 to Aug. 2013, and currently serves on the editorial board of Foundations and Trends in Communications and Information Theory.
\end{IEEEbiographynophoto}
\end{document}